\documentclass[a4paper, 10pt]{article}
\usepackage[utf8x]{inputenc}
\usepackage[english]{babel}
\usepackage{amsfonts, amssymb, amsmath, amsthm, stmaryrd, amsbsy,
  mathrsfs, genyoungtabtikz, yfonts}
\usepackage{pifont, indentfirst, dsfont, relsize}
\usepackage{geometry, graphicx, tabularx, wrapfig}
\usepackage{fancyhdr}
\usepackage{mdframed}
\usepackage[strict]{changepage}
\usepackage{xcolor, fancybox, empheq}
\definecolor{linkcolor}{rgb}{0,0,0.6}
\definecolor{mygreen}{rgb}{0.32,1.65,0.50}
\definecolor{myblue}{rgb}{0.08, 0.38, 0.74}
\usepackage[	pdftex,colorlinks=true,	
			pdfstartview=FitV,
			linkcolor= linkcolor,
			citecolor= linkcolor,
			urlcolor= linkcolor,
			hyperindex=true,
			hyperfigures=false,
                        backref=page]
			{hyperref}

\linethickness{1.2pt}

\usepackage{lastpage, cite}
\date{}
\newcommand{\hs}{\mathfrak{hs}}
\newcommand{\so}{\mathfrak{so}}
\newcommand{\gl}{\mathfrak{gl}}

\newcommand{\iso}{\mathfrak{iso}}
\newcommand{\Pd}[1]{\mathcal{P}^{(#1)}}

\newcommand*\Bell{\ensuremath{\boldsymbol\ell}}
\newcommand{\blambda}{\boldsymbol{\lambda}}
\newcommand{\bmu}{\boldsymbol{\mu}}
\newcommand{\balpha}{\boldsymbol{\alpha}}
\newcommand{\brho}{\boldsymbol{\rho}}
\newcommand{\C}{\mathcal{C}}

\newcommand{\D}{\mathcal{D}}

\newcommand{\W}{\mathcal{W}}
\newcommand{\fW}{\mathfrak{W}}

\newcommand{\N}{\mathbb{N}}
\newcommand{\Z}{\mathbb{Z}}
\newcommand{\Y}{\mathbb{Y}}

\newcommand{\V}{\mathcal{V}}

\newcommand{\U}{\mathcal{U}}
\newcommand{\Sn}{\mathcal{S}}

\newcommand{\flimit}{\underset{\lambda \rightarrow 0}{\longrightarrow}}

\newcommand{\rac}{\mathrm{Rac}}
\newcommand{\di}{\mathrm{Di}}
\newcommand{\rank}{\mathrm{rank}}
\newcommand{\Ann}{\mathrm{Ann}}
\newcommand{\ez}{\epsilon_0}
\newcommand{\eoh}{\epsilon_{1/2}}
\newcommand{\zero}{\boldsymbol{0}}
\newcommand{\half}{\boldsymbol{\tfrac12}}
\newcommand{\blb}{\boldsymbol{(}}
\newcommand{\brb}{\boldsymbol{)}}
\newcommand{\+}{\boldsymbol{+}}
\newcommand{\ogl}{\underset{\gl(d-1)}{\otimes}}

\newcommand{\branching}{\overset{\so(d)}{\underset{\so(d-1)}{\downarrow}}}
\newcommand{\branchingmod}[2]{\overset{\so(#1)}{\underset{\so(#2)}{\downarrow}}}
\newcommand{\branch}{\overset{\so(2,d)}{\underset{\so(2,d-1)}{\downarrow}}}
\newcommand{\branchiso}{\overset{\so(2,d)}{\underset{\iso(1,d-1)}{\downarrow}}}
\newcommand{\branchmod}[2]{\overset{\so(2,#1)}{\underset{\so(2,#2)}{\downarrow}}}

\newtheorem{theorem}{Theorem}[section] \newtheorem*{theorem*}{Theorem}
 \newtheorem*{lemma*}{Lemma}
\newtheorem{proposition}[theorem]{Proposition}
\theoremstyle{definition}
\newtheorem{definition}{Definition}[section]
\newtheorem*{definition*}{Definition}

\definecolor{theRed}{rgb}{0.56,0,0}

\newcommand{\myfontbackref}[1]{
    \hspace*{\fill} \mbox{\textsl {\small #1}}
}

\renewcommand*{\backref}[1]{}
\renewcommand*{\backrefalt}[4]{%
  \ifcase #1 \myfontbackref{(Not cited in the text)}
  \or        \myfontbackref{(Cited on page #2)}
  \else      \myfontbackref{(Cited on pages #2)}
  \fi
}

\mdfdefinestyle{rem}{topline=false, rightline=false, bottomline=false,
  linewidth=2pt}

\newenvironment{remark}
{ \begin{changemargin}{1.2cm}{0.5cm} \begin{mdframed}[style=rem]
      \underline{\it \bfseries Remark.} }
{ \end{mdframed} \end{changemargin} }

\newenvironment{example}
{ \begin{center}\rule{\textwidth}{1.5pt}\end{center} 
    \underline{\sc Example:} \it }
{ \begin{center}\rule{\textwidth}{1.5pt}\end{center} }

\Yboxdim{8pt} 
\Ylinethick{1pt}

\usepackage{tikz, tikz-cd, animate}
\usetikzlibrary{arrows, matrix}

\begin{document}

\thispagestyle{empty}
\pagenumbering{gobble}

\begin{center}
  \rule{0.4\textwidth}{.9pt} \bf \\
  \vspace{.3cm}
  \Large A note on rectangular partially massless fields\\ 
  \rule{0.4\textwidth}{.9pt} 
\end{center}

\vspace*{15pt}

\begin{center}
Thomas Basile\footnote{E-mail address:
  \href{mailto:thomas.basile@umons.ac.be}{\tt
    thomas.basile@khu.ac.kr}}
\end{center}

\vspace*{.5cm}

\begin{footnotesize} 
  \begin{minipage}[c]{0.45\textwidth}
    \begin{center}
      Laboratoire de Math\'ematiques et Physique Th\'eorique\\ Unit\'e
      Mixte de Recherche $7350$ du CNRS\\ F\'ed\'eration de Recherche
      $2964$ Denis Poisson\\ Universit\'e Fran\c{c}ois Rabelais, Parc
      de Grandmont\\ 37200 Tours, France \\
      \vspace{2mm}
    \end{center}
  \end{minipage}
  \hfill
  \begin{minipage}[c]{0.45\textwidth}
    \begin{center}
      Groupe de M\'ecanique et Gravitation\\ Service de Physique
      Th\'eorique et Math\'ematique\\ Universit\'e de Mons --
      UMONS\\ 20 Place du Parc\\ 7000 Mons, Belgique\\
    \vspace*{.3cm}
    \end{center}
  \end{minipage}
  \begin{center}
    Department of Physics and Research Institute of Basic Science,\\
    Kyung Hee University, Seoul 02447, Korea
  \end{center}
\end{footnotesize}

\vspace*{25pt}

\begin{abstract}
  We study a class of non-unitary $\so(2,d)$ representations (for even
  values of $d$), describing mixed-symmetry partially massless fields
  which constitute natural candidates for defining higher-spin
  singletons of higher order. It is shown that this class of
  $\so(2,d)$ modules obeys natural generalisations of a couple of
  defining properties of unitary higher-spin singletons. In
  particular, we find out that upon restriction to the subalgebra
  $\so(2,d-1)$, these representations branch onto a sum of modules
  describing partially massless fields of various depths. Finally,
  their tensor product is worked out in the particular case of $d=4$,
  where the appearance of a variety of mixed-symmetry partially
  massless fields in this decomposition is observed.
\end{abstract}

\vspace{20pt}

\tableofcontents

\vfill

\pagebreak
\setcounter{footnote}{0}
\pagenumbering{arabic}

\section{Introduction}
\label{sec:intro}
The completion of the Bargmann-Wigner program in anti-de Sitter (AdS)
spacetime\,\footnote{In the sense that for every $\so(2,d)$ UIRs of
  the lowest energy type, i.e. in the discrete series of
  representations, a realisation on a space of solutions of a wave
  equation is known. It should be noted however that the recent works
  \cite{Metsaev:2016lhs, Metsaev:2017ytk} uncovered the existence of
  ``continuous spin'' fields on AdS which do not fall into the
  discrete series, but might correspond to $\so(2,d)$ module induced
  from the non-compact subalgebra $\so(1,1) \oplus \so(1,d-1)$
  (instead of the maximal compact subalgebra $\so(2)\oplus\so(d)$) as
  noticed in \cite{Bekaert:2017khg}.} lead to some surprising lessons
concerning the definition of masslessness in other backgrounds than
Minkowski space. If nowadays, the most common way to discriminate
between massless and massive fields in AdS is whether or not they
enjoy some gauge symmetry, other proposals which involve a particular
kind of $\so(2,d)$ representations known as ``singletons'', were put
forward \cite{Angelopoulos:1980wg, Angelopoulos:1997ij}. Indeed, the
proposed notions of ``conformal masslessness'' and ``composite
masslessness'' both rely on two crucial properties of singletons,
namely\,\footnote{See e.g. \cite{Iazeolla:2008ix} where those
  properties were studied using the harmonic expansion of fields in
  (A)dS spacetimes.}:
\begin{itemize}
\item They are unitary and irreducible representations (UIRs) of
  $\so(2,d)$ that remain irreducible when restricted to UIRs of
  $\so(2,d-1)$, or in other words they correspond to the class of
  elementary particles in $d$-dimensional anti-de Sitter space which
  are conformal. This property is the very definition of a conformally
  massless UIR, and it turns out that the singletons are precisely the
  $\so(2,d)$ UIRs to which conformally massless $\so(2,d-1)$ UIRs can
  be lifted.
\item The tensor product of two $\so(2,3)$ singletons contains all
  conformally massless fields in AdS$_4$ \cite{Flato:1978qz}. In any
  dimensions however, the representations appearing in the
  decomposition of the tensor product of two $\so(2,d)$ singletons (of
  spin 0 or $\tfrac12$) are no longer conformally massless but make
  up, by definition, all of the composite massless UIRs of $\so(2,d)$
  \cite{Angelopoulos:1999bz, Vasiliev:2004cm, Dolan:2005wy,
    Iazeolla:2008ix}. In other words, composite massless UIRs are
  those modules which appear in the decomposition of the tensor
  product of two singletons.
\end{itemize}

This tensor product decomposition, called the Flato-Fr\o{}nsdal
theorem, is crucial in the context of Higher-Spin Gauge Theories and
can be summed up as follows (in the case of two scalar singletons):
\begin{equation}
  \rac \otimes \rac = \text{Massive scalar } \oplus
  \bigoplus_{s=1}^\infty \text{Gauge field of spin } s\, ,
  \label{FlatoFronsdalRac}
\end{equation}
where $\rac$ denotes the $\so(2,d)$ scalar singleton. Notice that the
spin-$s \geqslant 1$ gauge fields are both ``composite massless'' by
definition, as well as massless in the modern sense (as they enjoy
some gauge symmetry), whereas the scalar field is considered massive
in the sense of being devoid of said gauge symmetries, despite the
fact that it is also ``composite massless''. The $\so(2,d)$ UIRs on
the right hand side make up the spectrum of fields of Vasiliev's
higher-spin gravity \cite{Vasiliev:1988xc, Vasiliev:1990en,
  Vasiliev:1992av, Vasiliev:2003ev} (see e.g. \cite{Bekaert:2010hw}
and \cite{Bekaert:2005vh, Didenko:2014dwa} for, respectively,
non-technical and technical reviews). This decomposition can also be
interpreted in terms of operators of a free $d$-dimensional Conformal
Field Theory (CFT) as on the left hand side, the tensor product of two
scalar singletons can be thought of as a bilinear operator in the
fundamental scalar field and the right hand side as the various
conserved currents that this CFT possesses. This dual interpretation
of \eqref{FlatoFronsdalRac} is by now regarded as a first evidence in
favor of the AdS/CFT correspondence \cite{Maldacena:1997re,
  Witten:1998qj, Gubser:1998bc} in the context of Higher-Spin theory
\cite{Sezgin:2002rt, Klebanov:2002ja}. This duality relates (the type
A) Vasiliev's bosonic (minimal) higher-spin gravity to the free $U(N)$
($O(N)$) vector model and has passed several non-trivial checks since
it has been proposed, from the computation and matching of the
one-loop partition functions \cite{Giombi:2013fka, Giombi:2014yra} to
three point functions \cite{Sezgin:2003pt, Giombi:2009wh,
  Giombi:2012ms} on both sides of the duality (see
e.g. \cite{Giombi:2016ejx, Giombi:2016pvg, Sleight:2016hyl,
  Sleight:2017krf} and references therein for reviews of this
duality). The possible existence of such an equivalence between the
type-A Higher-Spin (HS) theory and the free vector model opened the
possibility of probing interactions in the bulk using the knowledge
gathered on the CFT side, a program which was tackled in
\cite{Bekaert:2014cea} (improving the earlier works \cite{Ruhl:2004cf,
  Manvelyan:2004ii, Manvelyan:2005fp}). This lead to the derivation of
all cubic vertices and the quartic vertex for four scalar fields in
the bulk \cite{Sleight:2016dba, Bekaert:2015tva}, as dictated by the
holographic duality, while \cite{Boulanger:2015ova} also raising
questions on the locality properties of the bulk HS theory (see
e.g. \cite{Skvortsov:2015lja, Vasiliev:2016xui, Taronna:2017jeq,
  Sleight:2017pcz, Bonezzi:2017vha} and references therein for more
details). \\

The fact that the prospective CFT dual to a HS theory in AdS$_{d+1}$
is free can be understood retrospectively thanks to the
Maldacena-Zhiboedov theorem \cite{Maldacena:2011jn} and its
generalisation \cite{Alba:2013yda, Alba:2015upa}. Indeed, it was shown
in \cite{Maldacena:2011jn} that if a $3$-dimensional CFT which is
unitary, obeys the cluster decomposition axiom and has a (unique)
Lorentz covariant stress-tensor plus at least one higher-spin current,
then this theory is either a CFT of free scalars or free spinors. This
was generalised to arbitrary dimensions in \cite{Alba:2013yda,
  Alba:2015upa}\,\footnote{See also \cite{Boulanger:2013zza} where the
  authors derived the bulk counterpart of this theorem.}, where the
authors showed that this result holds in dimensions $d \geqslant 3$,
up to the additional possibility of a free CFT of
$(\tfrac{d-2}2)$-forms in even dimensions. These free conformal fields
precisely correspond to the singleton representations of spin $0$,
$\tfrac12$ and $1$ in arbitrary dimensions \cite{Siegel:1988gd,
  Angelopoulos:1997ij}\,\footnote{Singletons of spin-$s>1$ do not
  appear here due to the fact that the free CFT based on these fields
  do not possess a conserved current of spin-$2$, i.e. a stress-energy
  tensor and therefore are not covered by the theorem derived in
  \cite{Alba:2015upa}. This can be seen for instance from the
  decomposition of the tensor product of two spin-$s$ singletons
  spelled out in \cite{Dolan:2005wy}, where one can check that there
  are no currents of spin lower than $2s$.}. Due to the fact that,
according to the standard AdS/CFT dictionary, the higher-spin gauge
field making up the spectrum of the Higher-Spin theory in the bulk are
dual to higher-spin conserved current on the CFT side, this CFT should
be free\,\footnote{Notice that by changing the boundary condition of
  fields in the bulk opens the possibility of having an HS theory dual
  to an interacting CFT, see for instance \cite{Klebanov:2002ja,
    Leigh:2003gk, Sezgin:2003pt, Skvortsov:2015pea,
    Giombi:2016ejx}. From the CFT point of view, this corresponds to
  the modification of the previously mentioned Maldacena-Zhiboedov
  theorem studied by the same authors in \cite{Maldacena:2012sf}.
  Instead of the existence of at least one conserved higher spin
  current, it is assumed that the CFT possesses a parameter $N$
  together with a tower of single trace, approximately conserved
  currents of all even spin $s \geqslant 4$, such that the
  conservation law gets corrected by terms of order $1/N$. As a
  consequence, the anomalous dimensions of these higher spin currents
  are of order $1/N$, which translates into the fact that the dual
  higher spin fields in the bulk acquire masses through radiative
  corrections, thereby leading to changes in their boundary
  conditions.}  as it falls under the assumption of the previously
recalled results of \cite{Maldacena:2011jn, Alba:2013yda,
  Alba:2015upa}. Hence, the algebra generated by the set of charges
associated with the conserved currents of the CFT whose fundamental
field is a spin-$s$ singletons corresponds to the HS algebra of the HS
theory in the bulk with a spectrum of field given by the decomposition
of the tensor product of two spin-$s$ singletons. These HS algebras
can be defined as follows:
\begin{equation}
  \hs_s^{(d)} \cong \frac{\U\big(\so(2,d)\big)}{\Ann\big(\D_s^{\rm
      sing.}\big)}\, ,
  \label{def_hs_alg}
\end{equation}
where $\hs_s^{(d)}$ stands for the HS algebra associated with the
spin-$s$ singleton in AdS$_{d+1}$, and $\U\big(\so(2,d)\big)$ denotes
the universal enveloping algebra of $\so(2,d)$ whereas $\D_s^{\rm
  sing.}$ denotes the spin-$s$ singleton module of $\so(2,d)$ and
$\Ann\big(\D_s^{\rm sing.}\big)$ the annihilator of this module. For
more details, see e.g. \cite{Iazeolla:2008ix, Boulanger:2011se,
  Joung:2014qya} where the construction of HS algebras (and their
relation with minimal representations of simple Lie algebras
\cite{Joseph1976}) is reviewed and \cite{Bekaert:2009fg} where HS
algebras associated with HS singletons were studied. Although
Vasiliev's Higher-Spin theory is based on the HS algebra $\hs_0^{(d)}$
and the HS theory based on $\hs_s^{(d)}$ with $s=\tfrac12, 1, \dots$,
is unknown\,\footnote{Notice that this is true only for $d>3$. For the
  AdS$_4$ case however, due to the fact that the annihilator of the
  scalar and the spin-$\tfrac12$ singleton are isomorphic
  (i.e. $\Ann(\D_0^{\rm sing.}) \cong \Ann(\D_{1/2}^{\rm sing.})$, see
  e.g. \cite{Iazeolla:2008ix} for more details), the HS algebra
  $\hs_{1/2}^{(3)}$ is isomorphic to $\hs_0^{(3)}$ and is therefore
  known. This translates into the fact that the two corresponding HS
  theory have almost the same spectrum of fields, the only difference
  being the mass of the bulk scalar field.}, the latter algebras are
quite interesting as they all describe a spectrum containing
mixed-symmetry fields. Even though this last class of massless field
is well understood at the free level (in flat space as well as in AdS)
\cite{Labastida:1987kw, Metsaev:1995re, Metsaev:1997nj,
  Metsaev:1998xg, Burdik:2000kj, Burdik:2001hj, Alkalaev:2003qv,
  Bekaert:2003az, Bekaert:2003zq, Boulanger:2008up, Boulanger:2008kw,
  Skvortsov:2008vs, Alkalaev:2008gi, Skvortsov:2009nv,
  Skvortsov:2009zu, Campoleoni:2009je, Alkalaev:2009vm,
  Alkalaev:2011zv, Campoleoni:2012th}, little is known about their
interaction (see e.g. \cite{Metsaev:2005ar, Metsaev:2007rn,
  Alkalaev:2010af, Boulanger:2011se, Boulanger:2011qt} on cubic
vertices and \cite{Alkalaev:2012rg, Alkalaev:2012ic,
  Chekmenev:2015kzf} where mixed-symmetry fields have been studied in
the context of the AdS/CFT correspondence). \\

A possible extension of the HS algebras associated with singletons can
be obtained by applying the above construction \eqref{def_hs_alg} with
a generalisation of the singleton representations $\D_s^{\rm sing.}$,
referred to as ``higher-order'' singletons. The latter are also
irreducible representations of $\so(2,d)$, which share the property of
describing fields ``confined'' to the boundary of AdS with the usual
singletons but which are non-unitary (as detailed in
\cite{Bekaert:2013zya}). This class of higher-order singletons, which
are of spin $0$ or $\tfrac12$, is labelled by a (strictly) positive
integer $\ell$. In the case of the scalar singleton of order $\ell$,
such a representation describes a conformal scalar $\phi$ obeying the
polywave equation:
\begin{equation}
  \Box^\ell \phi = 0\, .
\end{equation}
When $\ell=1$ one recovers the usual singleton (free, unitary
conformal scalar field), whereas $\ell>1$ leads to non-unitary
CFT. Such CFTs were studied in \cite{Brust:2016zns} for instance, and
were proposed to be dual to HS theories \cite{Bekaert:2013zya} whose
spectrum consists, on top of the infinite tower of (totally symmetric)
higher-spin massless fields, also partially massless (totally
symmetric) fields of arbitrary spin (theories which have been studied
recently in \cite{Alkalaev:2014nsa, Brust:2016xif, Brust:2016gjy}),
thereby extending the HS holography proposal of
Klebanov-Polyakov-Sezgin-Sundell to the non-unitary case. The
corresponding HS algebras were studied in \cite{Eastwood2008} for the
simplest case $\ell=2$ (as the symmetry algebra of the Laplacian
square, thereby generalising the previous characterisation of
$\hs_0^{(d)}$ as the symmetry algebra of the Laplacian
\cite{Eastwood:2002su}) and for general values of $\ell$ in
\cite{Gover2012, Michel2014, Joung:2015jza}. As we already mentionned,
the interesting feature of such HS algebras is that their spectrum,
i.e. the set of fields of the bulk theory, contains partially massless
(totally symmetric) higher-spin fields \cite{Bekaert:2013zya}
(introduced originally in \cite{Deser:1983mm, Deser:1983tm,
  Higuchi:1986wu, Deser:2001us}, and whose free propagation was
described in the unfolded formalism in
\cite{Skvortsov:2006at}). Although non-unitary in AdS background,
partially massless fields of arbitrary spin are unitary in de Sitter
background \cite{Basile:2016aen}, and hence constitute a particularly
interesting generalisation of HS gauge fields to
consider\,\footnote{Notice that partially massless HS fields have been
  shown to also appear in the spectrum of fields resulting from
  various compactifications of anti-de Sitter spacetime
  \cite{Gwak:2016sma, Gwak:2017cyu}.}. Partially massless fields, both
totally symmetric and of mixed-symmetry, also appear in the spectrum
of the HS algebra based on the order-$\ell$ spinor singleton
\cite{Basile:2014wua}. It seems reasonable to expect that the known
spectrum of the HS algebras based on a spin-$s$ singleton is enhanced,
when considering the HS algebras based on their higher-order
extension, with partially massless fields of the same symmetry type as
already present in the case of the original singleton. Therefore, a
natural question is whether or not there exist higher-order
higher-spin singletons. This question is adressed in the present note,
in which we study a class of $\so(2,d)$ modules which is a natural
candidate for defining a higher-order higher-spin singleton. \\

This paper is organised as follows: in
\hyperref[sec:notations]{Section \ref{sec:notations}} we start by
introducing the various notations that will be used throughout this
note, then in \hyperref[sec:hs_singletons]{Section
  \ref{sec:hs_singletons}} we first review the defining properties of
the well-known (unitary) higher-spin singletons before introducing
their would-be higher-order extension and spelling out the counterpart
of the previously recalled characteristic properties. Finally, the
tensor product of two such representations is decomposed in the
low-dimensional case $d=4$ in \hyperref[sec:flato_fronsdal]{Section
  \ref{sec:flato_fronsdal}}. Technical details on the branching and
tensor product rule of $\so(d)$ can be found in
\hyperref[app:sod]{Appendix \ref{app:sod}} while details of the proofs
of Propositions \hyperref[prop:decompo_ho]{\ref{prop:decompo_ho}} and
\hyperref[prop:branching_rule_ho]{\ref{prop:branching_rule_ho}} and
are relegated to \hyperref[app:technical]{Appendix
  \ref{app:technical}}.

\section{Notation and conventions}
\label{sec:notations}
In the rest of this paper, we will use the following symbols:
\begin{itemize}
\item A $\so(2,d)$ (generalised Verma) module is characterised by the
  $\so(2) \oplus \so(d)$ lowest weight $[\Delta; \Bell]$, where
  $\Delta$ is the $\so(2)$ weight (in general a real number, and more
  often in this paper, a positive integer) corresponding to the
  minimal energy of the AdS field described by this representation and
  $\Bell$ is a dominant integral $\so(d)$ weight corresponding to the
  spin of the representation. If irreducible, those modules will be
  denoted $\D(\Delta; \Bell)$, whereas if reducible (or
  indecomposable), they will be denoted $\V(\Delta; \Bell)$.
\item The spin $\Bell \equiv \blb \ell_1, \dots, \ell_r \brb$, with
  $r:=\rank(\so(d)) \equiv [\tfrac d2]$ (and where $[x]$ is the
  integer part of $x$), is a $\so(d)$ integral dominant weight. The
  property that the weight $\Bell$ is integral means that its
  components $\ell_i\, ,\, i=1,\dots,r$ are either {\it all} integers
  or {\it all} half-integers. The fact that $\Bell$ is dominant means
  that the components are ordered in decreasing order, and all
  positive except for the component $\ell_r$ when $d=2r$. More
  precisely,
  \begin{equation}
    \ell_1 \geqslant \ell_2 \geqslant \dots \geqslant \ell_{r-1}
    \geqslant \ell_r \geqslant 0\, , \quad \text{for} \quad
    \so(2r+1)\,,
  \end{equation}
  and
  \begin{equation}
    \ell_1 \geqslant \ell_2 \geqslant \dots \geqslant \ell_{r-1}
    \geqslant \rvert \ell_r \rvert\, , \quad \text{for} \quad
    \so(2r)\,.
  \end{equation}
\item In order to deal more efficiently with weight having several
  identical components, we will use the notation:
  \begin{equation}
    \blb \underbrace{\ell_1, \dots, \ell_1}_{h_1\, \text{times}},
    \underbrace{\ell_2, \dots, \ell_2}_{h_2\, \text{times}}, \dots,
    \underbrace{\ell_k, \dots, \ell_k}_{h_k\, \text{times}} \brb
    \equiv \blb \ell_1^{h_1}, \ell_2^{h_2}, \dots, \ell_k^{h_k} \brb\,
    , \quad \text{with} \quad k \leqslant r\, ,
  \end{equation}
  in other words the number $h$ of components with the same value
  $\ell$ appears as the exponent of the latter. For the special cases
  where all components of the highest weight are equal either to $0$
  or to $\tfrac12$, we will use bold symbols (and forget about the
  brackets), i.e.
  \begin{equation}
    \zero := \blb 0, \dots, 0 \brb\, , \quad \text{and} \quad \half :=
    \blb \tfrac12, \dots, \tfrac12 \brb\, .
  \end{equation}
  We will also write only the non-vanishing components of the various
  $\so(d)$ weight encountered in this paper, i.e.
  \begin{equation}
    \blb s_1, \dots, s_k \brb := \blb s_1, \dots, s_k, \underbrace{0,
      \dots, 0}_{r-k} \brb\, , \quad \text{for} \quad 1 \leqslant k
    \leqslant r\, .
  \end{equation}
\item If the spin is given by an irrep of an even dimensional
  orthogonal algebra, i.e. $\so(d)$ for $d=2r$, the last component
  $\ell_r$ of this highest weight (if non-vanishing) can either be
  positive or negative. Whenever the statements involving such a
  weight does not depend on this sign, we will write $\Bell = \blb
  \ell_1, \dots, \ell_r \brb$ with the understanding that $\ell_r$ can
  be replaced by $-\ell_r$. However, if the sign of the component
  $\ell_r$ matters, we will distinguish the two cases by writting:
  \begin{equation}
    \Bell_\epsilon \equiv \blb \ell_1, \dots, \ell_{r-1}, \ell_r
    \brb_\epsilon := \blb \ell_1, \dots, \ell_{r-1}, \epsilon \ell_r
    \brb\, , \quad \text{with} \quad \epsilon=\pm1\, .
  \end{equation}
  It will also be convenient to consider the direct sum of two modules
  labelled by $\Bell_+$ and $\Bell_-$, which we will denote by:
  \begin{equation}
    \D\big(\Delta; \Bell_0\big) := \D\big(\Delta; \Bell_+\big) \oplus
    \D\big(\Delta; \Bell_-\big)\, .
  \end{equation}
\item Finally, a useful tool that we will use throughout this paper is
  the character $\chi_{\V(\Delta;\Bell)}^{\so(2,d)}(q, \vec x)$ of a
  (possibly reducible) generalised Verma module $\V\big(\Delta\,;\,
  \Bell\big)$:
  \begin{equation}
    \chi_{\V(\Delta;\Bell)}^{\so(2,d)}(q, \vec x) = q^\Delta
    \chi^{\so(d)}_{\Bell}(\vec x) \Pd d (q, \vec x)\, ,
  \end{equation}
  where $q:=e^{-\mu_0}$ and $x_i:=e^{\mu_i}$ for $i=1,\dots,r$ with
  $\{\mu_0, \dots, \mu_r\}$ a basis of the weight space of $\so(2,d)$,
  and
  \begin{equation}
    \Pd d (q, \vec x) := \frac{1}{(1-q)^{d-2r}} \prod_{i=1}^r
    \frac{1}{(1-qx_i)(1-qx_i^{-1})}\, ,
  \end{equation}
  i.e. the prefactor $\tfrac{1}{1-q}$ is absent for $d=2r$, and
  $\chi^{\so(d)}_{\Bell}(\vec x)$ is the character of the irreducible
  $\so(d)$ representation $\Bell$. Any irreducible generalised Verma
  module $\D\big(\Delta\,;\, \Bell\big)$ can be defined as the
  quotient of the (freely generated) generalised Verma module
  $\V\big(\Delta\,;\, \Bell\big)$ by its maximal submodule
  $\D\big(\Delta'\,;\, \Bell'\big)$ (for some $\so(2)\oplus\so(d)$
  weight $[\Delta'\,;\, \Bell']$ related to $[\Delta\,;\, \Bell]$). An
  important property that will be used extensively in this work is the
  fact that given two representation spaces $V$ and $W$ of the same
  algebra with respective characters $\chi_V$ and $\chi_W$, the
  characters of the tensor product, direct sum or quotient of these
  two spaces obey:
  \begin{equation}
    \chi_{V \otimes W} = \chi_V \cdot \chi_W\, , \quad \chi_{V \oplus
      W} = \chi_V + \chi_W\, , \quad \chi_{V/W} = \chi_V - \chi_W\, .
  \end{equation}
  As a consequence, the character of an irreducible generalised Verma
  module $\D(\Delta;\Bell)$ takes the form:
  \begin{equation}
    \chi_{[\Delta;\Bell]}^{\so(2,d)}(q,\vec x) =
    \chi_{\V(\Delta;\Bell)}^{\so(2,d)}(q, \vec x) -
    \chi_{\D(\Delta';\Bell')}^{\so(2,d)}(q, \vec x) \, , 
  \end{equation}
  whenever $\V(\Delta;\Bell)$ possesses a submodule
  $\D(\Delta';\Bell')$. For more details on characters of $\so(2,d)$
  generalised Verma modules, see e.g. \cite{Dolan:2005wy,
    Beccaria:2014jxa}.
\end{itemize}

\section{Higher-order higher-spin singletons}
\label{sec:hs_singletons}
In this section, we start by reviewing the definition of the usual
(unitary) higher-spin singletons (about which more details can be
found in the pedagogical review \cite{Bekaert:2011js}, and in
\cite{Fernando:2015tiu} where they were studied from the point of view
of minimal representations), before moving on to the proposed
higher-order extension which is the main focus of this paper.

\paragraph{Unitary higher-spin singletons.}
Higher-spin singletons have been first considered by Siegel
\cite{Siegel:1988gd}, as making up the list of unitary and irreducible
representations of the conformal algebra $\so(2,d)$ which can lead to
a free conformal field theory. They were later identified by
Angelopoulos and Laoues \cite{Angelopoulos:1980wg,
  Angelopoulos:1997ij, Laoues:1998ik, Angelopoulos:1999bz} as being
part of the same class of particular representations first singled out
by Dirac \cite{Dirac:1963ta}, which is what is understood by
singletons nowadays. Initially, what lead Dirac to single out the
$\so(2,3)$ representations $\D\big(\tfrac12; \blb 0 \brb\big)$ and
$\D\big(1; \blb \tfrac12 \brb\big)$ studied in \cite{Dirac:1963ta} as
``remarkable'' is the fact that, contrarily to the usual UIRs of {\it
  compact} orthogonal algebras, the former are labelled by an highest
weight whose components are not both integers or both half-integers
but rather one is an integer and the other is an half-integer. In
other words, the highest weight defining this representation is not
integral dominant. On top of that, the other intriguing feature of
these representations, which was later elaborated on significantly by
Flato and Fr\o{}nsdal, is the fact they correspond respectively to a
scalar and a spinor field in AdS which do not propagate any local
degree of freedom in the bulk. This last property is the most striking
from a field theoretical point of view.  Indeed, the fact that
representations of the $\so(2,d)$ algebra can be interpreted both as
fields in AdS$_{d+1}$, i.e. the bulk, and as conformal fields on
$d$-dimensional Minkowski, i.e. the (conformal) boundary of
AdS$_{d+1}$ is at the core of the AdS/CFT correspondence. This last
characteristic translates into a defining property of the $\so(2,d)$
singleton modules, namely that they remain irreducible when restricted
to either one of the subalgebras $\so(2,d-1)$, $\so(1,d)$ or
$\iso(1,d-1)$. This is reviewed below after we define the singletons
as unitary and irreducible $\so(2,d)$ modules.\\

First, let us recall that the unitarity conditions for generalised
Verma modules of $\so(2,d)$ (i.e. in its discrete series of
representations) induced from the compact subalgebra
$\so(2)\oplus\so(d)$ were derived independently in
\cite{Metsaev:1995re, Metsaev:1997nj, Metsaev:1998xg} and in
\cite{Ferrara:2000nu} (where the more general result of
\cite{Enright1983} giving unitarity conditions for highest weight
modules of Hermitian algebras was applied to $\so(2,d)$). The outcome
of these analyses is that the irreducible modules $\D\big(\Delta\,;\,
\Bell\big)$ which are unitary are:
\begin{itemize}\label{unitarity_condition}
\item $\Bell=\zero$: modules with $\Delta \geqslant \tfrac{d-2}2$;
\item $\Bell=\half$: modules with $\Delta \geqslant \tfrac{d-1}2$;
\item $\Bell = \blb s^p, s_{p+1}, \dots, s_r \brb$ with $1 \geqslant s
  > s_{p+1} \geqslant \dots \geqslant s_r$: modules with $\Delta
  \geqslant s + d - p -1$.
\end{itemize}

With these unitarity bounds in mind for $\so(2,d)$ generalised Verma
modules, let us move onto the definition of {\it unitary} singletons:

\begin{definition}[Singleton]\label{def:unitary_singletons}
  A spin-$s$ singleton is defined as the $\so(2,d)$ module:
  \begin{equation}
    \D\big(s+\tfrac d2-1; \blb s^r \brb \big)\, ,
  \end{equation}
  for $s=0,\tfrac12$ in arbitrary dimensions, and $s \in N$ when
  $d=2r$. Introducing the minimal energies of the scalar and spinor singleton
  \begin{equation}
    \ez := \frac{d-2}2\, , \quad \text{and} \quad \eoh := \frac{d-1}2
    \equiv \ez +\tfrac12\, ,
  \end{equation}
  all of the above modules can be denoted as $\D\big( \ez+s\,;\, \blb
  s^r \brb \big)$. Depending on the value of $s$ the structure of the
  the above module changes drastically:
  \begin{itemize}
  \item If $s=0$ or $\tfrac12$, then
    \begin{equation}
      \rac := \D\big( \ez; \zero \big) \cong \frac{\V\big( \ez; \zero
        \big)}{\D\big( d-\ez; \zero \big)}\, , \quad \di := \D\big(
      \eoh; \half \big) \cong \frac{\V\big( \eoh; \half \big)}{\D\big(
        d-\eoh; \half \big)}\, ,
    \end{equation}
    where $\D\big( d-\ez; \zero \big) = \V\big( d-\ez; \zero \big)$
    and $\D\big( d-\eoh; \half \big) = \V\big( d-\eoh; \half \big)$.
    Their character read \cite{Flato:1978qz, Dolan:2005wy}:
    \begin{equation}
      \chi^{\so(2,d)}_{\rac}(q, \vec x) = q^{\ez}\, (1-q^2)\, \Pd d
      (q, \vec x)\, , \quad \text{and} \quad \chi^{\so(2,d)}_{\di}(q,
      \vec x) = q^{\eoh}\, (1-q)\, \chi^{\so(d)}_{\half}(\vec x)\, \Pd
      d (q, \vec x)\, .
    \end{equation}
  \item If $s\geqslant 1$ (and $d=2r$), then:
    \begin{equation}
      \D\big( s + \tfrac d2 - 1\,;\, \blb s^r \brb\big) \cong
      \frac{\V\big( s + \tfrac d2 - 1\,;\, \blb s^r \brb\big)}{\D\big(
        s + \tfrac d2\,;\, \blb s^{r-1}, s-1 \brb\big)}\, .
      \label{module_hs_singletons}
    \end{equation}
    In this case, the structure of the maximal submodule $\D\big( s +
    \tfrac d2\,;\, \blb s^{r-1}, s-1 \brb \big)$ is more involved, the
    maximal submodule $\D\big( s+\tfrac d2\,;\, \blb s^{r-1}, s-1 \brb
    \big)$ can be defined through the sequence of quotients of
    generalized Verma modules:
    \begin{equation}
      \D\big(s+\tfrac d2 - 1 + k\,;\, \blb s^{r-k}, (s-1)^{k} \brb
      \big) := \frac{\V\big(s+\tfrac d2 - 1 + k\,;\, \blb s^{r-k},
        (s-1)^{k} \brb \big)}{\D\big(s+\tfrac d2 + k\,;\, \blb
        s^{r-k-1}, (s-1)^{k+1} \brb \big)}\, ,
    \end{equation}
    with $k=1, \dots, r-1$ and $\D\big(s+\tfrac d2 + r-1\,;\, \blb
    (s-1)^{r} \brb \big) \equiv \V\big(s+\tfrac d2 + r-1\,;\, \blb
    (s-1)^{r} \brb \big)$ is an irreducible module. For more details
    on the structure of irreducible generalised $\so(2,d)$ Verma
    module, see the classification displayed in
    \cite{Shaynkman:2004vu}. Their character read \cite{Dolan:2005wy}:
    \begin{equation}
      \chi^{\so(2,d)}_{[s+\tfrac d2 -1;(s^r)]}(q, \vec x) =
      q^{s+\tfrac d2 -1} \Big( \chi^{\so(d)}_{(s^r)}(\vec x) +
      \sum_{k=1}^{r} (-)^{k} q^{k}
      \chi^{\so(d)}_{(s^{r-k},(s-1)^{k})}(\vec x) \Big) \Pd d (q, \vec
      x)\, .
      \label{character_singleton}
    \end{equation}
  \end{itemize}
\end{definition}

\begin{remark}
  As advertised, all of the above modules corresponding to singletons
  are unitary. One can notice that they actually saturate the
  unitarity bound and are all irreps of twist $\ez$ (the twist $\tau$
  being the absolute value of the difference between the minimal
  energy $\Delta$ and the spin of $s$ of a $\so(2,d)$ irrep,
  $\tau:=\rvert \Delta-s \rvert$).
\end{remark}

All of the above $\so(2,d)$ modules share a couple of defining
properties recalled below:

\begin{theorem}[Properties of the singletons \cite{Angelopoulos:1980wg, Angelopoulos:1997ij}]\label{th:uni_singleton}
  A singleton on AdS$_{d+1}$ is a module $\D\big(s+\tfrac d2 - 1; \blb
  s^r \brb \big)$ of $\so(2,d)$ with $s=0,\tfrac12$ for any values $d$
  and $s\in \N$ for $d=2r$, enjoying the following properties:
  \begin{itemize}
  \item[(i)] It decomposes into a (infinite) single direct sum of
    $\so(2) \oplus \so(d)$ (finite-dimensional) modules in which each
    irrep of $\so(d)$ appears only once (is multiplicity free) and
    with a different $\so(2)$ weight, i.e.
    \begin{equation}
      \D\big(s+\tfrac d2-1; \blb s^r \brb \big) \cong
      \bigoplus_{\sigma=0}^\infty \D_{\so(2) \oplus \so(d)} \big(
      s+\tfrac d2-1+\sigma; \blb s+\sigma,s^{r-1} \brb \big)\, .
      \label{decompo_hs_sing}
    \end{equation}
    This was proven originally in \cite{Ehrman1957} for the $d=3$ case
    where only the spin-$0$ and spin-$\tfrac12$ singletons exist, and
    extended to arbitrary dimensions and for singletons of arbitrary
    spin in \cite{Angelopoulos:1997ij}.
  \item[(ii)] It branches into a single irreducible
    module\,\footnote{Or at most two, which is the case of the scalar
      singleton as will be illustrated later (in
      \hyperref[prop:dirac_l]{Proposition \ref{prop:dirac_l}}).} of
    the subalgebras $\iso(1,d-1)$, $\so(1,d)$ or $\so(2,d-1)$, in
    which case this branching rule reads:
    \begin{equation}
      \D\big(s+r-1\,;\, \blb s^r \brb \big) \quad \branch \quad
      \D\big( s+r-1\,;\, \blb s^{r-1} \brb \big)\, , \quad \text{for}
      \quad d=2r \quad \text{and} \quad s \geqslant 1\, ,
      \label{branching_uni}
    \end{equation}
    and where the $\so(2,d-1)$ module $\D\big( s+r-1\,;\, \blb s^{r-1}
    \brb \big)$ correspond to a massless field of spin $\blb s^{r-1}
    \brb$ in AdS$_d$. Conversely, singletons can be seen as the only
    $\iso(1,d-1)$, $\so(1,d)$ or $\so(2,d-1)$ modules that can be
    lifted to a module of $\so(2,d)$ (which is $\D\big(s+\tfrac d2 -
    1; \blb s^r \brb \big)$). From this point of view, this property
    can be restated as ``Singletons are the only (massless) particles,
    or gauge fields, in $d$-dimensional Minkowski, de Sitter or
    anti-de Sitter spacetime which also admit conformal symmetries'',
    as they are the only representations of the isometry algebra of
    the $d$-dimensional maximally symmetric spaces that can be lifted
    to a representation of the conformal algebra in
    $d$-dimensions. Again, this property was first proven in $d=3$ in
    \cite{Angelopoulos:1980wg} and later extended to arbitrary
    dimensions in \cite{Angelopoulos:1997ij, Metsaev:1995jp}. This was
    revisited recently in \cite{Barnich:2015tma}.
  \end{itemize}
\end{theorem}
\begin{proof}
  Considering that the couple of defining properties of the singletons
  are already known, we will only sketch the idea of their proofs --
  that can be found in the original papers -- by focusing on the
  simpler, low-dimensional, case of $\so(2,4)$ spin-$s$ singletons,
  leaving the general case in arbitrary dimensions to
  \hyperref[app:unitary_hs_singleton]{Appendix
    \ref{app:unitary_hs_singleton}}.
  \begin{itemize}
  \item[(i)] This decomposition can be proven by showing that the
    character of the module $\D\big(s+r-1; \blb s^r \brb \big)$ can be
    rewritten in the form:
    \begin{equation}
      \chi_{[s+r-1;(s^r)]}^{\so(2,d)}(q, \vec x) =
      \sum_{\sigma=0}^\infty q^{s+r-1+\sigma} \,
      \chi^{\so(d)}_{(s+\sigma, s^{r-1})}(\vec x)\, ,
    \end{equation}
    which is indeed the character of the direct sum of $\so(2) \oplus
    \so(d)$ modules displayed in \eqref{decompo_hs_sing}. This was
    proven in \cite{Dolan:2005wy}, and in practice the idea is simply
    to use the property of the ``universal'' function $\Pd d (q, \vec
    x)$ that it can be rewritten as:
    \begin{equation}
      \Pd d (q, \vec x) = \sum_{\sigma,n=0}^\infty q^{\sigma+2n}\,
      \chi^{\so(d)}_{(\sigma)}(\vec x)\, ,
      \label{p_function}
    \end{equation}
    and then perform the tensor product between the $\so(d)$
    characters appearing in the character
    $\chi_{[s+r-1;(s^r)]}^{\so(2,d)}(q, \vec x)$ with
    $\chi^{\so(d)}_{(s)}(\vec x)$. Let us to do that explicitely for
    $d=4$, where we can take advantage of the exceptional isomorphism
    $\so(4) \cong \so(3) \oplus \so(3)$ to deal with $\so(4)$ tensor
    products:
    \begin{eqnarray}
      \chi^{\so(2,4)}_{[s+1; (s,s)]}(q, \vec x) & = & q^{s+1}\Big(
      \chi^{\so(4)}_{(s,s)}(\vec x) - q \chi^{\so(4)}_{(s,s-1)}(\vec
      x) +q^2 \chi^{\so(4)}_{(s-1,s-1)}(\vec x) \Big) \Pd 4 (q, \vec
      x) \\ & = & \sum_{n=0}^\infty q^{s+1+2n} \Big(
      \sum_{\sigma=0}^{2s} q^{\sigma} \sum_{k=0}^\sigma
      \chi^{\so(4)}_{(s+k,s+k-\sigma)}(\vec x) +
      \sum_{\sigma=2s+1}^\infty q^\sigma \sum_{k=0}^{2s}
      \chi^{\so(4)}_{(\sigma+k-s,k-s)}(\vec x) \\ && -
      \sum_{\sigma=0}^{2s-1} q^{\sigma+1} \sum_{k=0}^{\sigma}
      \chi^{\so(4)}_{(s+k,s+k-\sigma-1)}(\vec x) -
      \sum_{\sigma=1}^{2s-1} q^{\sigma+1} \sum_{k=0}^{\sigma}
      \chi^{\so(4)}_{(s+k-1,s+k-\sigma)}(\vec x) \\ && \qquad -
      \sum_{\sigma=2s}^\infty q^{\sigma+1} \sum_{k=0}^{2s-1} \Big[
        \chi^{\so(4)}_{(\sigma+k+1-s,k-s)}(\vec x) +
        \chi^{\so(4)}_{(\sigma+k-s,k+1-s)}(\vec x) \Big] \\ && +
      \sum_{\sigma=0}^{2s-2} q^{\sigma+2} \sum_{k=0}^{\sigma}
      \chi^{\so(4)}_{(s+k-1,s+k-\sigma-1)}(\vec x) +
      \sum_{\sigma=2s-1}^\infty q^{\sigma+2} \sum_{k=0}^{2s-2}
      \chi^{\so(4)}_{(\sigma+k+1-s,k+1-s)}(\vec x) \Big) \nonumber\\ &
      = & \sum_{n=0}^\infty q^{s+1+2n} \Big( \sum_{\sigma=0}^\infty
      q^{\sigma} \chi^{\so(4)}_{(s+\sigma,s)}(\vec x) -
      \sum_{\sigma=1}^\infty q^{\sigma+1}
      \chi^{\so(4)}_{(s+\sigma-1,s)}(\vec x) \Big) \\ & = &
      \sum_{\sigma=0}^\infty q^{s+1+\sigma}
      \chi^{\so(4)}_{(s+\sigma,s)}(\vec x)\, .
    \end{eqnarray}

    This decomposition can be illustrated by drawing a ``weight
    diagram'', representing the $\so(2)$ weight of
    $\so(2)\oplus\so(d)$ modules as a function of the first component
    of their $\so(d)$ weights, see
    \hyperref[fig:weight_diagram]{Figure \ref{fig:weight_diagram}}
    below.

    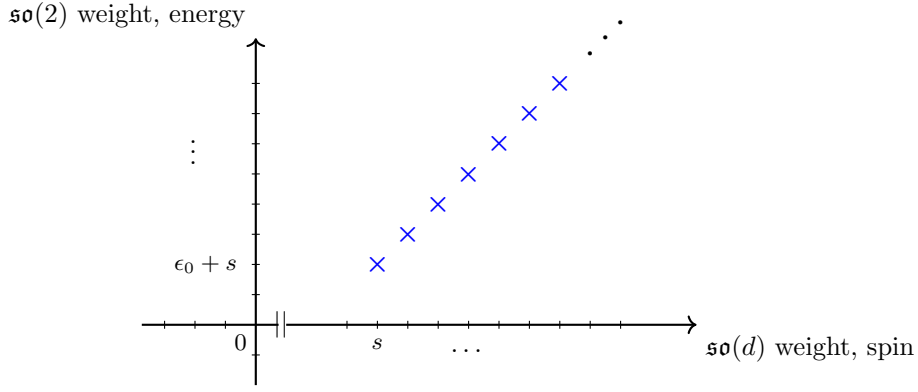
\begin{figure}[!ht]
      \center
      \begin{tikzpicture}
      \draw[thick] (-1.5,0) -- (0.3,0) node {$\,| |$};
      \draw[thick,->] (0.4,0) -- (5.8,0) node[anchor=north west]
           {$\so(d)$ weight, spin}; \draw[thick,->] (0,-0.8) --
           (0,3.8) node[anchor=south east] {$\so(2)$ weight, energy};

        \draw (-0.2, 0) node[below] {\small $0$};
        \draw (-0.4, 1.5pt) -- (-0.4, -1.5pt);
        \draw (-0.8, 1.5pt) -- (-0.8, -1.5pt);
        \draw (-1.2, 1.5pt) -- (-1.2, -1.5pt);
        \draw (1.2, 1.5pt) -- (1.2, -1.5pt);
        \draw (1.6, 1.5pt) -- (1.6, -1.5pt) node[below=0.2] {\small $s$};
        \draw (2, 1.5pt) -- (2, -1.5pt);
        \draw (2.4, 1.5pt) -- (2.4, -1.5pt);
        \draw (2.8, 1.5pt) -- (2.8, -1.5pt) node[below=4] {$\dots$};
        \draw (3.2, 1.5pt) -- (3.2, -1.5pt);
        \draw (3.6, 1.5pt) -- (3.6, -1.5pt);
        \draw (4, 1.5pt) -- (4, -1.5pt);
        \draw (4.4, 1.5pt) -- (4.4, -1.5pt);
        \draw (4.8, 1.5pt) -- (4.8, -1.5pt);

        \draw (-1.5pt, 0.4) -- (1.5pt, 0.4);
        \draw (-1.5pt, 0.8) -- (1.5pt, 0.8) node[left=5] {\small $\ez+s$};
        \draw (-1.5pt, 1.2) -- (1.5pt, 1.2);
        \draw (-1.5pt, 1.6) -- (1.5pt, 1.6);
        \draw (-1.5pt, 2) -- (1.5pt, 2);
        \draw (-1.5pt, 2.4) -- (1.5pt, 2.4) node[left=20] {$\vdots$};
        \draw (-1.5pt, 2.8) -- (1.5pt, 2.8);
        \draw (-1.5pt, 3.2) -- (1.5pt, 3.2);
        \draw (-1.5pt, -0.4) -- (1.5pt, -0.4);

        \draw (1.6, 0.8) node[color=blue] {$\boldsymbol{\times}$};
        \draw (2, 1.2) node[color=blue] {$\boldsymbol{\times}$};
        \draw (2.4, 1.6) node[color=blue] {$\boldsymbol{\times}$};
        \draw (2.8, 2) node[color=blue] {$\boldsymbol{\times}$};
        \draw (3.2, 2.4) node[color=blue] {$\boldsymbol{\times}$};
        \draw (3.6, 2.8) node[color=blue] {$\boldsymbol{\times}$};
        \draw (4, 3.2) node[color=blue] {$\boldsymbol{\times}$};
        \draw (4.4, 3.6) node[color=black] {$\boldsymbol{\cdot}$};
        \draw (4.6, 3.8) node[color=black] {$\boldsymbol{\cdot}$};
        \draw (4.8, 4) node[color=black] {$\boldsymbol{\cdot}$};
        
      \end{tikzpicture}
      \caption{Weight diagram of the spin-$s$ singleton.}
      \label{fig:weight_diagram}
    \end{figure}

    The fact the weight diagram of singletons is made out of a single
    line, noticed in the case of the Dirac singletons of $\so(2,3)$ in
    \cite{Ehrman1957} and later extended to singletons in arbitrary
    dimensions\,\footnote{Actually, it was even adopted as a
      definition of singletons in \cite{Angelopoulos:1999bz}.} in
    \cite{Angelopoulos:1997ij}, is the reason for the name
    ``singletons'' \cite{Bekaert:2011js}.

  \item[(ii)] In order to prove the branching rule from $\so(2,d)$ to
    $\so(2,d-1)$, we will compare the $\so(2) \oplus \so(d-1)$
    decomposition of the $\so(2,d)$ spin-$s$ singleton on the one
    hand, obtained by branching\,\footnote{The branching rules for
      $\so(d)$ irreps are recalled in
      \hyperref[app:branching]{Appendix \ref{app:branching}}.} the
    $\so(d)$ components of the $\so(2) \oplus \so(d)$ of these modules
    displayed in the previous item onto $\so(d-1)$, to the $\so(2)
    \oplus \so(d-1)$ decomposition of the $\so(2,d-1)$ module
    $\D\big(s+r-1; \blb s^{r-1} \brb\big)$ describing a massless field
    with spin $\blb s^{r-1} \brb$. For the sake of brevity, we will
    only detail the low dimensional case of $\so(2,4)$ spin-$s$
    singletons which captures the idea of the proof, and leave the
    treatment of the arbitrary dimension case to
    \hyperref[app:unitary_hs_singleton]{Appendix
      \ref{app:unitary_hs_singleton}}.\\

    \noindent Let us start by deriving the $\so(2) \oplus \so(3)$
    decomposition of the $\so(2,4)$ spin-$s$ singleton module $\D\big(
    s+1; \blb s, s \brb\big)$:
    \begin{eqnarray}
      \D\big( s+1; \blb s, s \brb \big) & \cong &
      \bigoplus_{\sigma=0}^\infty \D_{\so(2) \oplus \so(4)}\big(
      s+1+\sigma; \blb s+\sigma, s \brb \big) \\ & \branchmod{4}{3} &
      \bigoplus_{\sigma=0}^\infty \bigoplus_{k=0}^\sigma
      \D_{\so(2)\oplus\so(3)}\big( s+1+\sigma; \blb s+k \brb \big)
      \\ & \cong & \bigoplus_{\sigma,n=0}^\infty \D_{\so(2) \oplus
        \so(3)}\big( s+1+\sigma+n; \blb s+\sigma \brb \big)
    \end{eqnarray}
    Next, we need to derive the $\so(2) \oplus \so(3)$ of a massless
    spin-$s$ field corresponding to the $\so(2,3)$ module $\D\big(
    s+1; \blb s \brb\big)$. To do so, we will rewrite its character in
    a way that makes this decomposition explicit:
    \begin{eqnarray}
      \chi^{\so(2,3)}_{[s+1;(s)]}(q, x) & = & q^{s+1} \Big(
      \chi^{\so(3)}_{(s)}(x) - q \chi^{\so(3)}_{(s-1)}(x) \Big) \, \Pd
      d (q, x) \\ & = & \sum_{\sigma,n=0}^\infty q^{s+1+\sigma+2n}
      \Big( \sum_{\tau=|s-\sigma|}^{s+\sigma}
      \chi^{\so(3)}_{(\tau)}(x) - q
      \sum_{\tau=|s-1-\sigma|}^{s-1+\sigma} \chi^{\so(3)}_{(\tau)}(x)
      \Big) \\ & = & \sum_{n=0}^\infty q^{s+1+2n}\, \Big(
      \chi^{\so(3)}_{(s)} + \sum_{\sigma=1}^\infty q^\sigma \big[
        \sum_{\tau=|s-\sigma|}^{s+\sigma} \chi^{\so(3)}_{(\tau)}(x) -
        \sum_{\tau=|s-\sigma|}^{s+\sigma-2} \chi^{\so(3)}_{(\tau)}(x)
        \big] \Big) \\ & = & \sum_{n=0}^\infty q^{s+1+2n}\, (1+q)
      \sum_{\sigma=0}^\infty q^\sigma \chi^{\so(3)}_{(s+\sigma)}(x) =
      \sum_{\sigma, n=0}^\infty q^{s+1+\sigma+n}\,
      \chi^{\so(3)}_{(s+\sigma)}(x)\, ,
    \end{eqnarray}
    where we used the property \eqref{p_function} of the function $\Pd
    4 (q, \vec x)$, namely
    \begin{equation}
      \Pd 4 (q, \vec x) = \sum_{s,n=0}^\infty q^{s+2n}
      \chi^{\so(4)}_{(s)}(\vec x)\, .
    \end{equation}
    This proves that the decomposition of the $\so(2,3)$ module of a
    massless spin-$s$ field in AdS$_4$ reads:
    \begin{equation}
      \D\big( s+1; \blb s \brb \big) \cong
      \bigoplus_{\sigma,n=0}^\infty \D_{\so(2)\oplus\so(3)}\big(
      s+1+\sigma+n; \blb s+\sigma \brb \big)\, ,
      \label{decompo_massless_ads4}
    \end{equation}
    which coincide with the $\so(2) \oplus \so(3)$ decomposition
    obtained after branching the $\so(2,4)$ spin-$s$ singleton module
    onto $\so(2,3)$, i.e. we indeed have
    \begin{equation}
      \D\big( s+1; \blb s, s \brb \big) \quad \branchmod{4}{3} \quad
      \D\big( s+1; \blb s \brb \big)\, .
    \end{equation}
    This can be graphically seen by implementing the branching rule of
    the weight diagram in
    \hyperref[fig:weight_diagram]{Fig. \ref{fig:weight_diagram}}. Indeed,
    the branching rule for the $\so(2r)$ irrep $\blb s+\sigma, s^{r-1}
    \brb$ is:
    \begin{equation}
      \blb s+\sigma, s^{r-1} \brb \quad \branch \quad
      \bigoplus_{k=0}^\sigma \blb s+k, s^{r-2} \brb\, ,
    \end{equation}
    which means that one should add on each line of the weight diagram
    (representing the $\so(d)$ modules appearing at fixed energy, or
    $\so(2)$ weight) in
    \hyperref[fig:weight_diagram]{Fig. \ref{fig:weight_diagram}} a dot
    at each value of $\ell_1$ to the left of the orignal one until
    $\ell_1=s$ is reached. By doing so, an infinite wedge whose tip
    has coordinates $(E=s+\ez,\ell_1=s)$ precisely corresponding to
    the weight diagram of a massless field of spin given by a
    rectangular Young diagram on maximal height and length $s$ as can
    be seen from \eqref{decompo_massless_ads4} for $d=3$ and in
    \hyperref[app:unitary_hs_singleton]{Appendix
      \ref{app:unitary_hs_singleton}} for arbitrary odd values of $d$.
    \begin{figure}[!ht]
      \begin{minipage}[c]{0.4\textwidth}
        \center
        \begin{tikzpicture}
          \draw[thick] (-1.5,0) -- (0.3,0) node {$\,| |$};
          \draw[thick,->] (0.4,0) -- (5.8,0) node[anchor=north west]
               {$\ell_1$}; \draw[thick,->] (0,-0.8) --
               (0,3.8) node[anchor=south east] {$E$};
               
               \draw (-0.2, 0) node[below] {\small $0$};
               \draw (-0.4, 1.5pt) -- (-0.4, -1.5pt);  
               \draw (-0.8, 1.5pt) -- (-0.8, -1.5pt); 
               \draw (-1.2, 1.5pt) -- (-1.2, -1.5pt); 
               \draw (1.2, 1.5pt) -- (1.2, -1.5pt);
               \draw (1.6, 1.5pt) -- (1.6, -1.5pt) node[below=0.2] {\small $s$};
               \draw (2, 1.5pt) -- (2, -1.5pt);
               \draw (2.4, 1.5pt) -- (2.4, -1.5pt);
               \draw (2.8, 1.5pt) -- (2.8, -1.5pt) node[below=4] {$\dots$};
               \draw (3.2, 1.5pt) -- (3.2, -1.5pt);
               \draw (3.6, 1.5pt) -- (3.6, -1.5pt);
               \draw (4, 1.5pt) -- (4, -1.5pt);
               \draw (4.4, 1.5pt) -- (4.4, -1.5pt);
               \draw (4.8, 1.5pt) -- (4.8, -1.5pt);

               \draw (-1.5pt, 0.4) -- (1.5pt, 0.4);

               \draw (-1.5pt, 0.8) -- (1.5pt, 0.8) node[left=5] {\small $\ez+s$};
               \draw (-1.5pt, 1.2) -- (1.5pt, 1.2);
               \draw (-1.5pt, 1.6) -- (1.5pt, 1.6);
               \draw (-1.5pt, 2) -- (1.5pt, 2);
               \draw (-1.5pt, 2.4) -- (1.5pt, 2.4) node[left=20] {$\vdots$};
               \draw (-1.5pt, 2.8) -- (1.5pt, 2.8);
               \draw (-1.5pt, 3.2) -- (1.5pt, 3.2);
               \draw (-1.5pt, -0.4) -- (1.5pt, -0.4);
               
               \draw (1.6, 0.8) node[color=blue] {$\boldsymbol{\times}$};
               \draw (2, 1.2) node[color=blue] {$\boldsymbol{\times}$};
               \draw (2.4, 1.6) node[color=blue] {$\boldsymbol{\times}$};
               \draw (2.8, 2) node[color=blue] {$\boldsymbol{\times}$};
               \draw (3.2, 2.4) node[color=blue] {$\boldsymbol{\times}$};
               \draw (3.6, 2.8) node[color=blue] {$\boldsymbol{\times}$};
               \draw (4, 3.2) node[color=blue] {$\boldsymbol{\times}$};
               \draw (4.4, 3.6) node[color=black] {$\boldsymbol{\cdot}$};
               \draw (4.6, 3.8) node[color=black] {$\boldsymbol{\cdot}$};
               \draw (4.8, 4) node[color=black] {$\boldsymbol{\cdot}$};
        \end{tikzpicture}
      \end{minipage}
      \qquad $\branch$ \quad
      \begin{minipage}[c]{0.4\textwidth}
        \center
        \begin{tikzpicture}
          \draw[thick] (-1.5,0) -- (0.3,0) node {$\,| |$};
          \draw[thick,->] (0.4,0) -- (5.8,0) node[anchor=north west]
               {$\ell_1$}; \draw[thick,->] (0,-0.8) --
               (0,3.8) node[anchor=south east] {$E$};
               
               \draw (-0.2, 0) node[below] {\small $0$};
               \draw (-0.4, 1.5pt) -- (-0.4, -1.5pt);  
               \draw (-0.8, 1.5pt) -- (-0.8, -1.5pt); 
               \draw (-1.2, 1.5pt) -- (-1.2, -1.5pt); 
               \draw (1.2, 1.5pt) -- (1.2, -1.5pt);
               \draw (1.6, 1.5pt) -- (1.6, -1.5pt) node[below=0.2] {\small $s$};
               \draw (2, 1.5pt) -- (2, -1.5pt);
               \draw (2.4, 1.5pt) -- (2.4, -1.5pt);
               \draw (2.8, 1.5pt) -- (2.8, -1.5pt) node[below=4] {$\dots$};
               \draw (3.2, 1.5pt) -- (3.2, -1.5pt);
               \draw (3.6, 1.5pt) -- (3.6, -1.5pt);
               \draw (4, 1.5pt) -- (4, -1.5pt);
               \draw (4.4, 1.5pt) -- (4.4, -1.5pt);
               \draw (4.8, 1.5pt) -- (4.8, -1.5pt);
               
               \draw (-1.5pt, 0.4) -- (1.5pt, 0.4); 
               \draw (-1.5pt, 0.8) -- (1.5pt, 0.8) node[left=5] {\small $\ez+s$}; 
               \draw (-1.5pt, 1.2) -- (1.5pt, 1.2);
               \draw (-1.5pt, 1.6) -- (1.5pt, 1.6);
               \draw (-1.5pt, 2) -- (1.5pt, 2);
               \draw (-1.5pt, 2.4) -- (1.5pt, 2.4) node[left=20] {$\vdots$}; 
               \draw (-1.5pt, 2.8) -- (1.5pt, 2.8); 
               \draw (-1.5pt, 3.2) -- (1.5pt, 3.2); 
               \draw (-1.5pt, -0.4) -- (1.5pt, -0.4); 
               
               \draw (1.6, 0.8) node[color=blue] {$\boldsymbol{\times}$};
               \draw (2, 1.2) node[color=blue] {$\boldsymbol{\times}$};
               \draw (1.6, 1.2) node[color=myblue] {$\boldsymbol{\times}$};
               \draw (2.4, 1.6) node[color=blue] {$\boldsymbol{\times}$};
               \draw (2, 1.6) node[color=myblue] {$\boldsymbol{\times}$};
               \draw (1.6, 1.6) node[color=myblue] {$\boldsymbol{\times}$};
               \draw (2.8, 2) node[color=blue] {$\boldsymbol{\times}$};
               \draw (2.4, 2) node[color=myblue] {$\boldsymbol{\times}$};
               \draw (2, 2) node[color=myblue] {$\boldsymbol{\times}$};
               \draw (1.6, 2) node[color=myblue] {$\boldsymbol{\times}$};
               \draw (3.2, 2.4) node[color=blue] {$\boldsymbol{\times}$};
               \draw (2.8, 2.4) node[color=myblue] {$\boldsymbol{\times}$};
               \draw (2.4, 2.4) node[color=myblue] {$\boldsymbol{\times}$};
               \draw (2, 2.4) node[color=myblue] {$\boldsymbol{\times}$};
               \draw (1.6, 2.4) node[color=myblue] {$\boldsymbol{\times}$};
               \draw (3.6, 2.8) node[color=blue] {$\boldsymbol{\times}$};
               \draw (3.2, 2.8) node[color=myblue] {$\boldsymbol{\times}$};
               \draw (2.8, 2.8) node[color=myblue] {$\boldsymbol{\times}$};
               \draw (2.4, 2.8) node[color=myblue] {$\boldsymbol{\times}$};
               \draw (2, 2.8) node[color=myblue] {$\boldsymbol{\times}$};
               \draw (1.6, 2.8) node[color=myblue] {$\boldsymbol{\times}$};
               \draw (4, 3.2) node[color=blue] {$\boldsymbol{\times}$};
               \draw (3.6, 3.2) node[color=myblue] {$\boldsymbol{\times}$};
               \draw (3.2, 3.2) node[color=myblue] {$\boldsymbol{\times}$};
               \draw (2.8, 3.2) node[color=myblue] {$\boldsymbol{\times}$};
               \draw (2.4, 3.2) node[color=myblue] {$\boldsymbol{\times}$};
               \draw (2, 3.2) node[color=myblue] {$\boldsymbol{\times}$};
               \draw (1.6, 3.2) node[color=myblue] {$\boldsymbol{\times}$};
               \draw (4.4, 3.6) node[color=black] {$\boldsymbol{\cdot}$};
               \draw (4.6, 3.8) node[color=black] {$\boldsymbol{\cdot}$};
               \draw (4.8, 4) node[color=black] {$\boldsymbol{\cdot}$};
               \draw (1.6, 3.6) node[color=gray] {$\boldsymbol{\cdot}$};
               \draw (1.6, 3.8) node[color=gray] {$\boldsymbol{\cdot}$};
               \draw (1.6, 4) node[color=gray] {$\boldsymbol{\cdot}$};
        \end{tikzpicture}
      \end{minipage}
      \caption{Left: $\so(2)\oplus\so(d)$ weight diagram of the
        spin-$s$ singletons (with in abscisse the first component of
        the $\so(d)$ weights, denoted $\ell_1$). Right:
        $\so(2)\oplus\so(d-1)$ weight diagram of the $\so(2,d-1)$
        module $\D\big(s+\tfrac d2 -1\,;\, \blb s^{r-1} \brb\big)$
        (with in abscisse the first component of the $\so(d-1)$
        weights, denoted $\ell_1$ as well). Lighter blue crosses
        $\color{myblue} \times$ for a given $\so(2)$ weight $E$
        represent the $\so(d-1)$ representations coming from the
        branching rule of the $\so(d)$ representations in the
        $\so(2)\oplus\so(d)$ module (with the same $\so(2)$ weight
        $E$) of the singleton decomposition represented by a darker
        blue cross $\color{blue} \times$.}
    \end{figure}
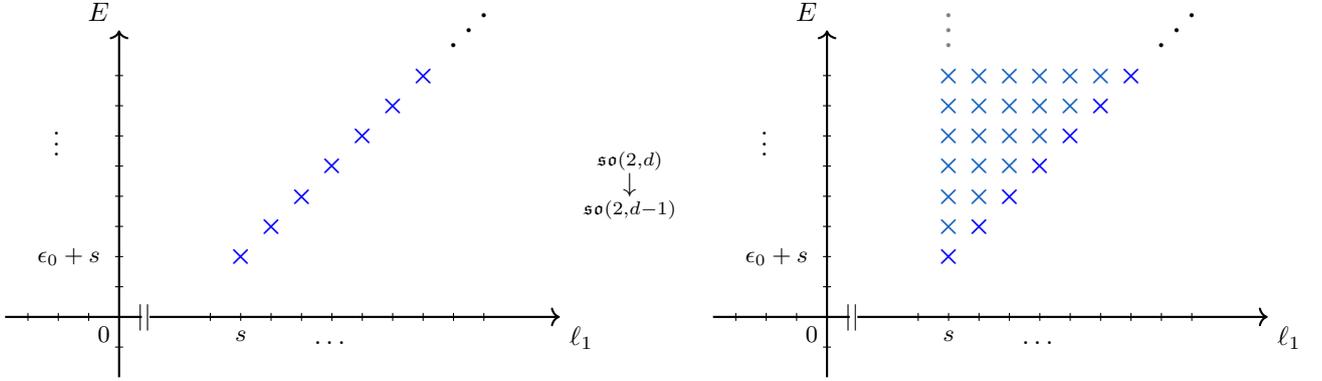
  \end{itemize}
\end{proof}

We did not, in the previous review of the proofs of the listed
properties in \hyperref[th:uni_singleton]{Theorem
  \ref{th:uni_singleton}}, cover the branching rule of the $\so(2,d)$
singletons onto $\iso(1,d-1)$ or $\so(1,d)$ for the following reasons:

\begin{itemize}
\item {\bf From $\so(2,d)$ to $\iso(1,d-1)$.} As far as the branching
  rule from $\so(2,d)$ to $\so(1,d-1)$ are concerned, it can be
  recovered, assuming that the following diagram is commutative:
  \begin{equation}
    \begin{tikzcd}
      \D_{\so(2,d)} \arrow{dd}[left]{\branch} \arrow{rrrdd}{
        \overset{\so(2,d)}{\underset{\iso(1,d-1)}{\downarrow}}} &
      \qquad & \qquad & \qquad \\ \qquad & \qquad & \qquad
      \\ \D_{\so(2,d-1)} \arrow{rrr}[below]{\so(2,d-1) \flimit
        \iso(1,d-1)} & \qquad & \qquad & \D_{\iso(1,d-1)}\, ,
    \end{tikzcd}
    \label{com_diagram}
  \end{equation}
  i.e. by combining the branching rule from $\so(2,d)$ to $\so(2,d-1)$
  and an In\"on\"u-Wigner contraction. That is to say, it is
  equivalent (i) to branch a representation $\D_{\so(2,d)}$ from
  $\so(2,d)$ onto $\so(2,d-1)$ and then perform a In\"on\"u-Wigner
  contraction by sending the cosmological constant $\lambda$ to zero
  to obtain a representation $\D_{\iso(1,d-1)}$ of $\iso(1,d-1)$, and
  (ii) to branch the $\so(2,d)$ module $\D_{\iso(1,d-1)}$ onto
  $\iso(1,d-1)$ to obtain the same module $\D_{\iso(1,d-1)}$ than
  previously. Under this assumption, we can use the branching rule
  \eqref{branching_uni} of the $\so(2,d)$ singleton module onto
  $\so(2,d-1)$ and then contracting it to a $\iso(1,d-1)$ instead of
  deriving the branching rule from $\so(2,d)$ onto $\iso(1,d-1)$. The
  In\"on\"u-Wigner contraction for massless fields in AdS$_{d+1}$
  (i.e. $\so(2,d)$ modules) is known as the Brink-Metsaev-Vasiliev
  mechanism \cite{Brink:2000ag}, which was proven in
  \cite{Boulanger:2008up, Boulanger:2008kw, Alkalaev:2009vm}. This
  mechanism states that massless $\so(2,d)$ UIRs of spin given by a
  $\so(d)$ Young diagram $\Y$ contracts to the direct sum of massless
  UIRs of the Poincar\'e algebra with spin given by all of the Young
  diagrams obtained from the branching rule of $\Y$ {\it except} those
  where boxes in the first block of $\Y$ have been
  removed. Higher-spin singleton, as well as the massless $\so(2,d-1)$
  module onto which they branch being labelled by a rectangular Young
  diagram, the BMV mechanism implies that they contract to a single
  $\iso(1,d-1)$ i.e.
  \begin{equation}
    \D_{\so(2,d)}\big(s+r-1\,;\, \blb s^r \brb\big) \branch
    \D_{\so(2,d-1)}\big( s+r-1 \,;\, \blb s^{r-1} \brb \big) \quad
    \flimit \quad \D_{\iso(1,d-1)}\big(m=0\,;\, \blb s^{r-1} \brb
    \big)\, ,
  \end{equation}
  as shown in \cite{Angelopoulos:1997ij}.
\item {\bf From $\so(2,d)$ to $\so(1,d)$.} The $\so(1,d)$ generalised
  Verma modules are induced by, and decompose into,
  $\so(1,1)\oplus\so(d-1)$ modules instead of $\so(2)\oplus\so(d-1)$
  in the case of $\so(2,d-1)$. As a consequence, the method used
  previously consisting in relying on the common
  $\so(2)\oplus\so(d-1)$ cannot be applied here and we will therefore
  refer to the original paper \cite{Angelopoulos:1997ij} for the proof
  of that branching rule.
\end{itemize}

\paragraph{Non-unitary, higher-order extension.}
Higher-order extension of the Dirac singletons (i.e. the scalar and
spinor ones) are non-unitary $\so(2,d)$ modules that share the crucial
field theoretical property of singletons mentioned above, namely they
correspond to AdS (scalar and spinor) field that do not propagate
local degree of freedom in the bulk. They have been considered in
\cite{Iazeolla:2008ix} as well as in \cite{Bekaert:2013zya} where the
confinement to the conformal boundary of these remarkable fields was
highlighted, but were excluded from the exhaustive
work\,\footnote{Although mentionned briefly in
  \cite{Angelopoulos:1999bz} as ``multipleton''.}
\cite{Angelopoulos:1997ij} because they fall below the unitary bound
for representations of $\so(2,d)$ (recalled
\hyperref[unitarity_condition]{previously}).

\begin{definition}[Higher-order Dirac singletons]\label{def:racdil}
  The scalar and spinor, order-$\ell$ Dirac singletons are the
  $\so(2,d)$ modules $\D\big( \ez^{(\ell)}; \zero\big)$ and $\D\big(
  \eoh^{(\ell)}; \half \big)$ respectively, where
  \begin{equation}
    \ez^{(\ell)} := \frac{d-2\ell}2\, , \quad \text{and} \quad
    \eoh^{(\ell)} := \frac{d+1-2\ell}{2} \equiv \ez^{(\ell)} +
    \tfrac12\, ,
  \end{equation}
  and which are defined as the quotient:
  \begin{equation}
    \rac_\ell := \D\big( \ez^{(\ell)}\,;\, \zero \big) \cong
    \frac{\V\big( \ez^{(\ell)}\,;\, \zero \big)}{\D\big(
      d-\ez^{(\ell)}\,;\, \zero \big)}\, , \quad \text{and} \quad
    \di_\ell := \D\big( \eoh^{(\ell)}\,;\, \half \big) \cong
    \frac{\V\big( \eoh^{(\ell)}\,;\, \half \big)}{\D\big(
      d-\eoh^{(\ell)}\,;\, \half \big)}\, .
  \end{equation}
  Their character read:
  \begin{equation}
    \chi^{\so(2,d)}_{\rac_\ell}(q, \vec x) = q^{\ez^{(\ell)}}\,
    (1-q^{2\ell})\, \Pd d (q, \vec x)\, , \quad \text{and} \quad
    \chi^{\so(2,d)}_{\di_\ell}(q, \vec x) = q^{\eoh^{(\ell)}}\,
    (1-q^{2\ell-1})\, \chi^{\so(d)}_{\half}(\vec x)\, \Pd d (q, \vec
    x)\, .
    \label{character_dirac_l}
  \end{equation}
  These modules are non-unitary for $\ell>1$, whereas they correspond
  to the original (unitary) Dirac singletons of
  \hyperref[def:unitary_singletons]{Definition
    \ref{def:unitary_singletons}} for $\ell=1$.
\end{definition}

On top of the confinement property, the $\rac_\ell$ and $\di_\ell$
singletons also possess properties analogous to those of their unitary
counterparts reviewed in \hyperref[th:uni_singleton]{Theorem
  \ref{th:uni_singleton}}. Specifically, they can be decomposed as
several direct sum of $\so(2)\oplus\so(d)$ modules, making up not only
one but now several lines in the weight diagram, and they obey a
branching rule (from $\so(2,d)$ to $\so(2,d-1)$) similar to that of
$\rac$ and $\di$. The properties of the higher-order Dirac singletons
are summed up below.

\begin{proposition}[Properties of $\rac_\ell$ and $\di_\ell$]
  \label{prop:dirac_l}
  The $\so(2)\oplus\so(d)$ decomposition of the order-$\ell$ scalar
  and spinor singletons respectively read\,\footnote{Notice that the
    $\so(2)\oplus\so(d)$ decomposition \eqref{rac_l_decompo} of the
    $\rac_\ell$ module was originally derived in
    \cite{Iazeolla:2008ix} (where these singletons are refered to as
    ``scalar $p$-linetons'')}:
  \begin{equation}
    \D\big( \ez^{(\ell)}\,;\, \zero\big) \cong
    \bigoplus_{k=0}^{\ell-1} \bigoplus_{\sigma=0}^\infty
    \D_{\so(2)\oplus\so(d)}\big( \ez^{(\ell)}+\sigma+2k\,;\, \blb
    \sigma \brb\big)\, ,
    \label{rac_l_decompo}
  \end{equation}
  and
  \begin{equation}
    \D\big( \eoh^{(\ell)}\,;\, \half \big) \cong
    \bigoplus_{k=0}^{2(\ell-1)} \bigoplus_{\sigma=0}^\infty
    \D_{\so(2)\oplus\so(d)}\big( \eoh^{(\ell)}+\sigma+k\,;\, \blb
    \sigma +\tfrac12, (\tfrac12)^{r-1} \brb\big)\, .
    \label{di_l_decompo}
  \end{equation}
  These two modules obey the following branching
  rules\,\footnote{Notice that the module on the right hand side of
    the branching rules \eqref{branching_rac_l} and
    \eqref{branching_di_l} for $k=0$ are not the order $\ell$
    singleton, due to the fact that $\ez^{(\ell)} = \tfrac{d-2\ell}2$
    is not the critical energy of the $\rac_\ell$ singleton of
    $\so(2,d-1)$.}:
  \begin{equation}
    \D\big( \ez^{(\ell)}\,;\, \zero\big) \quad \branch \quad
    \bigoplus_{k=0}^{2\ell-1} \D\big( \ez^{(\ell)} + k \,;\, \zero
    \big)\, ,
    \label{branching_rac_l}
  \end{equation}
  and
  \begin{equation}
    \D\big( \eoh^{(\ell)}\,;\, \half\big) \quad \branch \quad
    \bigoplus_{k=0}^{2(\ell-1)} \D\big( \eoh^{(\ell)} + k \,;\, \half
    \big)\, .
    \label{branching_di_l}
  \end{equation}
\end{proposition}
\begin{proof}
  As previsouly, we will use the property \eqref{p_function} of the
  function $\Pd d (q, \vec x)$ to rewrite the characters of the
  order-$\ell$ scalar and spinor singletons \eqref{character_dirac_l}
  as a sum of $\so(2)\oplus\so(d)$ characters, starting with the
  scalar $\rac_\ell$:
  \begin{eqnarray}
    \chi^{\so(2,d)}_{\rac_\ell}(q, \vec x) & = &
    q^{\ez^{(\ell)}}(1-q^{2\ell})\, \Pd d (q, \vec x) =
    \sum_{\sigma,n=0}^\infty q^{\ez^{(\ell)}+\sigma+2n}
    (1-q^{2\ell})\, \chi_{(\sigma)}^{\so(d)}(\vec x) \\ & = &
    \sum_{k=0}^{\ell-1} \sum_{\sigma=0}^\infty
    q^{\ez^{(\ell)}+2k+\sigma}\, \chi^{\so(d)}_{(\sigma)}(\vec x)
    \\ \Leftrightarrow \quad \rac_\ell & = & \bigoplus_{k=0}^{\ell-1}
    \bigoplus_{\sigma=0}^\infty \D_{\so(2)\oplus\so(d)}\big(
    \ez^{(\ell)} + 2k +\sigma\,;\, \blb \sigma \brb\big)\, .
  \end{eqnarray}
  For the $\di_\ell$ singleton we will also need the $\so(d)$ tensor
  product rule:
  \begin{equation}
    \blb \sigma \brb \otimes \half = \blb \sigma+\tfrac12,
    (\tfrac12)^{r-1} \brb \oplus \blb \sigma-\tfrac12,
    (\tfrac12)^{r-1} \brb\, , \qquad \text{for} \quad \sigma \geqslant
    1\, .
  \end{equation}
  Using the above identity and proceeding similarly to the scalar
  case, we end up with:
  \begin{eqnarray}
    \chi^{\so(2,d)}_{\di_\ell}(q, \vec x) & = &
    q^{\eoh^{(\ell)}}(1-q^{2\ell-1})\, \chi^{\so(d)}_{\half}(\vec x)\,
    \Pd d (q, \vec x) \\ & = & \sum_{n=0}^\infty q^{\eoh^{(\ell)}+2n}
    (1-q^{2\ell-1})\, \Big(\sum_{\sigma=0}^\infty q^\sigma
    \chi_{(\sigma+\tfrac12, (\tfrac12)^{r-1})}^{\so(d)}(\vec x) +
    \sum_{\sigma=1}^\infty q^\sigma \chi_{(\sigma-\tfrac12,
      (\tfrac12)^{r-1})}^{\so(d)}(\vec x) \Big) \\ & = & \sum_{n,
      \sigma=0}^\infty q^{\eoh^{(\ell)}+2n+\sigma} (1-q^{2\ell-1})\,
    (1+q)\, \chi_{(\sigma+\tfrac12, (\tfrac12)^{r-1})}^{\so(d)}(\vec
    x) \\ & = & \sum_{k=0}^{2(\ell-1)} \sum_{\sigma=0}^\infty
    q^{\eoh^{(\ell)}+k+\sigma}\, \chi^{\so(d)}_{(\sigma+\tfrac12,
      (\tfrac12)^{r-1})}(\vec x) \\ \Leftrightarrow \quad \di_\ell &
    = & \bigoplus_{k=0}^{2(\ell-1)} \bigoplus_{\sigma=0}^\infty
    \D_{\so(2)\oplus\so(d)}\big( \eoh^{(\ell)}+\sigma+k\,;\, \blb
    \sigma +\tfrac12, (\tfrac12)^{r-1} \brb\big)\, .
  \end{eqnarray}
  To prove the branching rule \eqref{branching_rac_l} and
  \eqref{branching_di_l}, we will follow the same strategy as
  previously, namely we will compare the $\so(2)\oplus\so(d-1)$
  decomposition of the two sides of these identities. This
  decomposition reads, for the $\rac_\ell$ singleton:
  \begin{eqnarray}
    \D\big( \ez^{(\ell)}\,;\, \zero\big) & \cong &
    \bigoplus_{k=0}^{\ell-1} \bigoplus_{\sigma=0}^\infty
    \D_{\so(2)\oplus\so(d)}\big( \ez^{(\ell)}+\sigma+2k\,;\, \blb
    \sigma \brb\big) \\ & \branch & \bigoplus_{k=0}^{\ell-1}
    \bigoplus_{\sigma=0}^\infty \bigoplus_{n=0}^\sigma
    \D_{\so(2)\oplus\so(d-1)}\big( \ez^{(\ell)}+\sigma+2k\,;\, \blb n
    \brb\big) \\ & \cong & \bigoplus_{k=0}^{\ell-1}
    \bigoplus_{\sigma=0}^\infty \bigoplus_{n=0}^\infty
    \D_{\so(2)\oplus\so(d-1)}\big( \ez^{(\ell)}+\sigma+2k+n\,;\, \blb
    \sigma \brb\big) \, ,
  \end{eqnarray}
  whereas for the $\di_\ell$ singleton:
  \begin{eqnarray}
    \D\big( \eoh^{(\ell)}\,;\, \half\big) & \cong &
    \bigoplus_{k=0}^{2(\ell-1)} \bigoplus_{\sigma=0}^\infty
    \D_{\so(2)\oplus\so(d)}\big( \eoh^{(\ell)}+\sigma+k\,;\, \blb
    \sigma+\tfrac12, (\tfrac12)^{r-1} \brb\big) \\ & \branch &
    \bigoplus_{k=0}^{2(\ell-1)} \bigoplus_{\sigma=0}^\infty
    \bigoplus_{n=0}^\sigma \D_{\so(2)\oplus\so(d-1)}\big(
    \eoh^{(\ell)}+\sigma+k\,;\, \blb n+\tfrac12, (\tfrac12)^{r-1}
    \brb\big) \\ & \cong & \bigoplus_{k=0}^{2(\ell-1)}
    \bigoplus_{\sigma=0}^\infty \bigoplus_{n=0}^\infty
    \D_{\so(2)\oplus\so(d-1)}\big( \eoh^{(\ell)}+\sigma+k+n\,;\, \blb
    \sigma+\tfrac12, (\tfrac12)^{r-1} \brb\big) \, .
  \end{eqnarray}
  On the other hand, the character of an irreducible $\so(2,d-1)$
  module $\D\big( \Delta\,;\, \zero\big) \equiv \V\big(\Delta\,;\,
  \zero\big)$, i.e. a generalised Verma module which does not contain
  a submodule\,\footnote{A scalar $\so(2,d-1)$ module
    $\V\big(\Delta\,;\, \zero\big)$ possesses a submodule only if
    $\Delta=\tfrac{d-1-2\ell}2$, whereas a spin one-half module
    $\V\big(\Delta\,;\, \half\big)$ possesses a submodule only if
    $\Delta=\tfrac{d-2\ell}2$. In other words, only the $\rac_\ell$
    and $\di_\ell$ modules are defined as quotients, see the
    classification in \cite{Shaynkman:2004vu}.} can be rewritten as:
  \begin{eqnarray}
    \chi^{\so(2,d-1)}_{[\Delta\,;\,\zero]}(q, \vec x) & = & q^\Delta\,
    \Pd{d-1}(q, \vec x) = \sum_{\sigma,n=0}^\infty
    q^{\Delta+2n+\sigma}\, \chi^{\so(d-1)}_{(\sigma)}(\vec x)
    \\ \Leftrightarrow \quad \D\big(\Delta\,;\, \zero\big) & \cong &
    \bigoplus_{\sigma=0}^\infty \bigoplus_{n=0}^\infty
    \D_{\so(2)\oplus\so(d-1)}\big( \Delta+\sigma+2n\,;\, \blb \sigma
    \brb\big)
  \end{eqnarray}
  As a consequence, 
  \begin{equation}
    \D\big( \ez^{(\ell)} + 2k \,;\, \zero \big) \oplus \D\big(
    \ez^{(\ell)} + 2k + 1 \,;\, \zero \big) \cong
    \bigoplus_{\sigma=0}^\infty \bigoplus_{n=0}^\infty
    \D_{\so(2)\oplus\so(d-1)}\big( \ez^{(\ell)}+\sigma+n+2k\,;\, \blb
    \sigma \brb\big)\, ,
  \end{equation}
  which proves \eqref{branching_rac_l}. Finally, an irreducible
  $\so(2,d-1)$ module $\D\big( \Delta\,;\, \half\big) \equiv
  \V\big(\Delta\,;\, \half\big)$ admits the following
  $\so(2)\oplus\so(d-1)$ decomposition:
  \begin{eqnarray}
    \chi^{\so(2,d-1)}_{[\Delta\,;\,\half]}(q, \vec x) & = & q^\Delta\,
    \chi^{\so(d-1)}_{\half}(\vec x) \Pd{d-1}(q, \vec x) \\ & = &
    \sum_{n=0}^\infty q^{\Delta+2n}\, \Big( \sum_{\sigma=0}^\infty
    q^\sigma \chi^{\so(d-1)}_{(\sigma+\tfrac12,
      (\tfrac12)^{r-1})}(\vec x) + \sum_{\sigma=1}^\infty q^{\sigma}
    \chi^{\so(d-1)}_{(\sigma-\tfrac12, (\tfrac12)^{r-1})}(\vec x)
    \Big) \\ & = & \sum_{\sigma=0}^\infty \sum_{n=0}^\infty
    q^{\Delta+\sigma+n}\, \chi^{\so(d-1)}_{(\sigma+\tfrac12,
      (\tfrac12)^{r-1})}(\vec x) \\ \Leftrightarrow \quad
    \D\big(\Delta\,;\, \half\big) & \cong &
    \bigoplus_{\sigma=0}^\infty \bigoplus_{n=0}^\infty
    \D_{\so(2)\oplus\so(d-1)}\big( \Delta+\sigma+n\,;\, \blb
    \sigma+\tfrac12, (\tfrac12)^{r-1} \brb\big)\, ,
  \end{eqnarray}
  thereby proving \eqref{branching_di_l}.
\end{proof}

The branching rule \eqref{branching_rac_l} and \eqref{branching_di_l}
reproduce that given in \cite{Angelopoulos:1997ij,
  Angelopoulos:1999bz} (and rederived in \cite{Artsukevich:2008vy})
for the Rac and Di singletons upon setting $\ell=1$, and extend them
to the higher-order Dirac singletons $\rac_\ell$ and $\di_\ell$.\\

From a CFT point of view, the order-$\ell$ scalar and spinor
singletons correspond to respectively a non-unitary fundamental scalar
or spinor fields of respective conformal weight $\ez^{(\ell)}$ and
$\eoh^{(\ell)}$, and respectively subject to an order $2\ell$ and
$2\ell-1$ wave equation (see e.g. \cite{Bekaert:2013zya} for more
details). The spectrum of current of these CFT contains an infinite
tower of partially conserved totally symmetric currents of arbitrary
spin, which should be dual to partially massless gauge fields in the
bulk \cite{Dolan:2001ih}.

\paragraph{Candidates for higher-spin higher-order singletons.} 
The extension we will be concerned with corresponds to the $\so(2,d)$
module, for $d=2r$:
\begin{equation}
  \D\big(s+\tfrac d2 - t\,;\, \blb s^r \brb\big) \cong
  \frac{\V\big(s+\tfrac d2 - t\,;\, \blb s^r
    \brb\big)}{\D\big(s+\tfrac d2\,;\, \blb s^{r-1}, s-t \brb\big)}\,
  , \quad \text{for} \quad 1 \leqslant t \leqslant s\, ,
  \label{module_hs_ho_singleton}
\end{equation}
whose structure is similar to the unitary spin-$s$ singletons for $s
\geqslant 1$ in the sense that the various submodule to be modded out
of $\V(s+\tfrac d2-t\,;\, \blb s^r \brb \big)$ are defined throught
the sequence:
\begin{equation}
  \D\big( s+\tfrac d2+k \,;\, \blb s^{r-k-1}, (s-1)^k, s-t \brb \big)
  := \frac{\V\big( s+\tfrac d2+k \,;\, \blb s^{r-k-1}, (s-1)^k, s-t
    \brb \big)}{\D\big( s+\tfrac d2+k+1 \,;\, \blb s^{r-k-2},
    (s-1)^{k+1}, s-t \brb \big)}\, ,
\end{equation}
for $0 \leqslant k \leqslant r-2$ and $\D\big( s+d-1 \,;\, \blb
(s-1)^{r-1}, s-t \brb \big) \equiv \V\big( s+d-1 \,;\, \blb
(s-1)^{r-1}, s-t \brb \big)$. In other words, except for the first
submodule which is obtained by increasing the $\so(2)$ weight of $t$
units and removing $t$ boxes from the last row of the rectangular
Young diagram $\blb s^r \brb$ labelling the irreducible module, the
sequence of nested submodules are related to one another by adding one
unit to the $\so(2)$ weight of the previous submodule and removing one
box in the row above the previously amputated row. Correspondingly,
the character of this module reads:
\begin{equation}
  \chi^{\so(2,d)}_{[s+\tfrac d2 -t;(s^r)]}(q, \vec x) = q^{s+\tfrac d2
    -t} \Big( \chi^{\so(d)}_{(s^r)}(\vec x) + \sum_{k=0}^{r-1}
  (-)^{k+1} q^{t+k} \chi^{\so(d)}_{(s^{r-1-k},(s-1)^k,s-t)}(\vec x)
  \Big) \Pd d (q, \vec x)\, .
  \label{character_ho_singleton}
\end{equation}

This definition encompasses the unitary spin-$s$ singletons, which
correspond to the case $t=1$ saturating the unitarity bound. For $t >
1$ (but always $t \leqslant s$), the module
\eqref{module_hs_ho_singleton} is non-unitary and describes a
depth-$t$ partially-massless field of spin $\blb s^r \brb$. The spin
being given by a rectangular Young diagram, we will refer to this
class of module as ``rectangular'' partially massless (RPM) fields of
spin $s$ and depth $t$. From the boundary point of view, the modules
\eqref{module_hs_ho_singleton} correspond to the curvature a conformal
field of spin $\blb s^{r-1} \brb$ (hence the curvature is given by a
tensor of symmetry described by a rectangular Young diagram of length
$s$ and height $r$) obeying a partial conservation law of order $t$,
i.e. taking $t$ symmetrised divergences of this curvature identically
vanishes on-shell (see e.g. \cite{Bae:2016xmv} where the $d=4$ and
$t=1$ case was discussed, and \cite{Vasiliev:2009ck} for a more
details on mixed symmetry conformal field in arbitrary dimensions).

\begin{remark}
  Notice that formally, the modules of the $\rac_\ell$ and $\di_\ell$
  singletons, as well as the module \eqref{module_hs_ho_singleton} that
  we propose here as a higher-spin generalisation of the higher-order
  scalar and spinor singletons, can be denoted as:
  \begin{equation}
    \D\big( s+\ez^{(t)}\,;\, \blb s^r \brb\big)\, , \quad \text{with}
    \quad s=0, \tfrac12\, , \quad \text{and} \quad s \in \N \quad
    \text{if} \quad d=2r
  \end{equation}
  with $\ez^{(t)} = \tfrac{d-2t}2$ as defined in
  \hyperref[def:racdil]{Definition \ref{def:racdil}}. On top of being
  notationally convenient, this coincidence is actually the reason why
  the modules \eqref{module_hs_ho_singleton} are ``natural''
  generalisations of the unitary higher-spin singletons: by
  introducing the parameter $t$ in this way, one considers a family of
  modules whose first representative is the unitary singletons whereas
  for $t>1$ the modules are non-unitary but their structure is almost
  the same than in the unitary case.
\end{remark}

Let us now study what are the counterpart of the properties displayed
in \hyperref[th:uni_hs_singlton]{Theorem \ref{th:uni_singleton}} for
unitary singletons and \hyperref[prop:dirac_l]{Proposition
  \ref{prop:dirac_l}} for the $\rac_\ell$ and $\di_\ell$ singletons,
starting with the $\so(2)\oplus\so(d)$ decomposition of
\eqref{module_hs_ho_singleton}.

\begin{proposition}[$\so(2)\oplus\so(d)$ decomposition]\label{prop:decompo_ho}
  The $\so(2,d)$ module $\D\big( s+r-t; \blb s^r \brb \big)$ for
  $d=2r$, describing a depth-$t$ and spin-$s$ RPM field, admits the
  following $\so(2) \oplus \so(d)$ decomposition:
  \begin{equation}
    \D\big( s+r-t; \blb s^r \brb \big) \quad \cong \quad
    \bigoplus_{\ell=0}^{t-1} \bigoplus_{n=0}^{t-1-\ell}
    \bigoplus_{\sigma=\ell}^\infty \D_{\so(2)\oplus\so(d)}\big(
    s+r-t+\sigma+2n; \blb s-\ell+\sigma, s^{r-2}, s-\ell \brb \big)\,
    .
    \label{decompo_ho_hs_sing}
  \end{equation}
  Equivalently, this property means that the character
  \eqref{character_ho_singleton} can bewritten as:
  \begin{equation}
    \chi^{\so(2,d)}_{[s+r-t;(s^r)]}(q, \vec x) = \sum_{\ell=0}^{t-1}
    \sum_{n=0}^{t-1-\ell} \sum_{\sigma=\ell}^\infty q^{s+r
      -t+\sigma+2n}
    \chi^{\so(d)}_{(s+\sigma-\ell,s^{r-2},s-\ell)}(\vec x)\, .
  \end{equation}
\end{proposition}
\begin{proof}
  As previously, we will only focus on the simpler $d=4$ case and
  leave the proof of this property in arbitrary dimension to
  \hyperref[app:nonunitary_hs_singleton]{Appendix
    \ref{app:nonunitary_hs_singleton}}. We will proceed in the exact
  same way as we did for unitary higher-spin singleton, that is we
  will use \eqref{p_function} in the character formula
  \eqref{character_ho_singleton}, so as to rewrite it in the following
  way:
  \begin{eqnarray}
    \chi^{\so(2,4)}_{[s+2-t; (s,s)]}(q, \vec x) & = & q^{s+2-t}\Big(
    \chi^{\so(4)}_{(s,s)}(\vec x) - q^{t} \chi^{\so(4)}_{(s,s-t)}(\vec
    x) +q^{t+1} \chi^{\so(4)}_{(s-1,s-t)}(\vec x) \Big) \Pd 4 (q, \vec
    x) \\ & = & \sum_{n=0}^\infty q^{s+2-t+2n} \Big(
    \sum_{\sigma=0}^{2s} q^{\sigma} \sum_{k=0}^\sigma
    \chi^{\so(4)}_{(s+k,s+k-\sigma)}(\vec x) +
    \sum_{\sigma=2s+1}^\infty q^\sigma \sum_{k=0}^{2s}
    \chi^{\so(4)}_{(\sigma+k-s,k-s)}(\vec x) \\ && - \sum_{m=0}^t
    \Big[ \sum_{\sigma=m}^{2s-t} q^{\sigma+t} \sum_{k=0}^{\sigma}
      \chi^{\so(4)}_{(s+k-m,s+k-\sigma-t+m)}(\vec x) +
      \sum_{\sigma=2s-t+1}^\infty q^{\sigma+t} \sum_{k=0}^{2s-t}
      \chi^{\so(4)}_{(\sigma+k+t-s-m,k+m-s)}(\vec x) \Big] \nonumber
    \\ && + \sum_{m=0}^{t-1} \Big[ \sum_{\sigma=m}^{2s-t-1}
      q^{\sigma+t+1} \sum_{k=0}^{\sigma}
      \chi^{\so(4)}_{(s+k-m-1,s+k-\sigma-t+m)}(\vec x) \\ && \qquad
      \qquad \qquad \qquad \qquad + \sum_{\sigma=2s-t}^\infty
      q^{\sigma+t+1} \sum_{k=0}^{2s-t-1}
      \chi^{\so(4)}_{(\sigma+k+t-s-m,k+m+1-s)}(\vec x) \Big] \Big)
    \nonumber \\ & = & \sum_{n=0}^\infty q^{s+2-t+2n} \Big(
    \sum_{\sigma=0}^{t-1} q^{\sigma} \sum_{k=0}^\sigma
    \chi^{\so(4)}_{(s+k,s+k-\sigma)}(\vec x) + \sum_{m=0}^{t-1}
    \sum_{\sigma=t}^{\infty} q^{\sigma}
    \chi^{\so(4)}_{(s+\sigma-m,s-m)}(\vec x) \label{inter1} \\ &&
    \qquad \qquad \qquad \qquad \qquad \qquad \qquad - \sum_{m=1}^{t}
    \sum_{\sigma=m}^\infty q^{\sigma+t}
    \chi^{\so(4)}_{(s+\sigma-m,s-t+m)}(\vec x) \Big) \\ & = &
    \sum_{n=0}^\infty q^{s+2-t+2n} \sum_{m=0}^{t-1} \Big(
    \sum_{\sigma=m}^{\infty} q^{\sigma}
    \chi^{\so(4)}_{(s+\sigma-m,s-m)}(\vec x) -
    \sum_{\sigma=t-m}^\infty q^{\sigma+t}
    \chi^{\so(4)}_{(s+\sigma-t+m,s-m)}(\vec x) \Big) \label{inter2}
    \\ & = & \sum_{n=0}^\infty q^{s+2-t+2n} \sum_{m=0}^{t-1}
    \sum_{\sigma=m}^{\infty} q^{\sigma} (1- q^{2(t-m)})
    \chi^{\so(4)}_{(s+\sigma-m,s-m)}(\vec x) \\ & = &
    \sum_{\ell=0}^{t-1} \sum_{n=0}^{t-1-\ell}
    \sum_{\sigma=\ell}^{\infty} q^{\sigma+s+2-t+2n}
    \chi^{\so(4)}_{(s-\ell+\sigma,s-\ell)}(\vec x)\, ,
    \label{final_decompo}
  \end{eqnarray}
  where we used
  \begin{equation}
    \sum_{\sigma=0}^{t-1} q^\sigma \sum_{k=0}^{\sigma}
    \chi^{\so(4)}_{(s+k,s+k-\sigma)}(\vec x) = \sum_{m=0}^{t-1}
    \sum_{\sigma=m}^{t-1} q^\sigma
    \chi^{\so(4)}_{(s+\sigma-m,s-m)}(\vec x)\, ,
  \end{equation}
  between \eqref{inter1} and \eqref{inter2}. Expression
  \eqref{final_decompo} shows that the depth-$t$ PM module
  $\D\big(s+2-t; \blb s, s \brb\big)$ decomposes as the direct sum of
  $\so(2)\oplus\so(4)$ modules:
  \begin{equation}
    \D\big(s+2-t\,;\, \blb s, s \brb\big) \quad \cong \quad
    \bigoplus_{\ell=0}^{t-1} \bigoplus_{n=0}^{t-1-\ell}
    \bigoplus_{\sigma=\ell}^\infty \D_{\so(2)\oplus\so(4)}\big(
    s+2-t+\sigma+2n\,;\, \blb s+\sigma-\ell, s-\ell \brb\big)\, .
  \end{equation}
\end{proof}

With the previous $\so(2)\oplus\so(d)$ decomposition at hand, we can
now derive the branching rule of the spin-$s$ depth-$t$ RPM module.

\begin{proposition}[Branching rule]\label{prop:branching_rule_ho}
  The $\so(2,d)$ module $\D\big( s+r-t\,;\, \blb s^r \brb \big)$ for
  $d=2r$, describing a depth-$t$ and spin-$s$ RPM field, branches onto
  the direct sum of $\so(2,d-1)$ modules $\D\big( s+r-\tau\,;\, \blb
  s^{r-1} \brb \big)$ with $\tau=1, \dots, t$ describing partially
  massless fields in AdS$_d$ of spin $\blb s^{r-1} \brb$ and with
  depth-$\tau$:
  \begin{equation}
    \D\big( s+r-t\,;\, \blb s^r \brb \big) \quad \branch \quad
    \bigoplus_{\tau=1}^t \D\big( s+r-\tau\,;\, \blb s^{r-1} \brb
    \big)\, .
    \label{branching_ho}
  \end{equation}
\end{proposition}
\begin{proof}
  Here again we will only display the proof for the low dimensional
  case $d=4$ in order to illustrate the general mechanism while being
  not too technically involved, and we leave the treatment in
  arbitrary dimensions to the
  \hyperref[app:nonunitary_hs_singleton]{Appendix
    \ref{app:nonunitary_hs_singleton}}.\\ 

  In order to prove the branching rule \eqref{branching_ho} for
  $\so(2,4)$, we will compare the $\so(2)\oplus \so(3)$ decomposition
  of the $\so(2,4)$ spin-$s$ and depth-$t$ singleton (obtained by
  first branching it onto $\so(2,3)$) to the $\so(2)\oplus \so(3)$
  decomposition of the $\so(2,3)$ spin-$s$ and depth-$\tau$ partially
  massless fields. Let us start with the latter, i.e. derive the
  $\so(2) \oplus \so(3)$ decomposition of the $\so(2,3)$ module
  $\D\big(s+2-\tau\,;\, \blb s \brb\big)$ using its character:
  \begin{eqnarray}
    \chi^{\so(2,3)}_{[s+2-\tau;(s)]}(q, x) & = & q^{s+2-\tau} \big(
    \chi^{\so(3)}_{(s)}(x) - q^\tau \chi^{\so(3)}_{(s-\tau)}(x) \big)
    \Pd 3 (q, x) \\ & = & \sum_{n,\sigma=0}^\infty
    q^{s+2-\tau+\sigma+2n} \Big( \sum_{k=|s-\sigma|}^{s+\sigma}
    \chi^{\so(3)}_{(k)}(x) - q^\tau
    \sum_{k=|s-\sigma-\tau|}^{s+\sigma-\tau} \chi^{\so(3)}_{(k)}(x)
    \Big) \\ & = & \sum_{n=0}^\infty q^{s+2-\tau+2n} \Big(
    \sum_{\sigma=0}^\infty q^\sigma \sum_{k=|s-\sigma|}^{s+\sigma}
    \chi^{\so(3)}_{(k)}(x) - \sum_{\sigma=\tau}^\infty q^\sigma
    \sum_{k=|s-\sigma|}^{s+\sigma-2\tau} \chi^{\so(3)}_{(k)}(x) \Big)
    \\ & = & \sum_{n=0}^\infty q^{s+2-\tau+2n} \Big(
    \sum_{\sigma=0}^{\tau-1} q^\sigma \sum_{k=|s-\sigma|}^{s+\sigma}
    \chi^{\so(3)}_{(k)}(x) + \sum_{\sigma=\tau}^\infty q^\sigma
    \sum_{k=s+\sigma-2\tau+1}^{s+\sigma} \chi^{\so(3)}_{(k)}(x) \Big)
    \\ & = & \sum_{\sigma,n=0}^\infty \sum_{k=0}^{\tau-1}
    q^{s+2-\tau+n+\sigma+k} \chi^{\so(3)}_{(s+\sigma-k)}(x)\, .
  \end{eqnarray}
  Hence, the $\so(2)\oplus \so(3)$ decomposition of a $\so(2,3)$
  spin-$s$ and depth-$\tau$ partially massless field reads:
  \begin{equation}
    \D\big( s+2-\tau\,;\, \blb s \brb \big) \quad \cong \quad
    \bigoplus_{\sigma=0}^\infty \bigoplus_{n=0}^\infty
    \bigoplus_{k=0}^{\tau-1} \D_{\so(2) \oplus \so(3)}\big(
    s+2-\tau+n+\sigma+k\,;\, \blb s+\sigma-k \brb \big)\, .
    \label{decompo_pm_3}
  \end{equation}
  This can be represented graphically by the weight diagram displayed
  in \hyperref[fig:weight_pm]{Figure \ref{fig:weight_pm}} for
  $\tau=3$.
  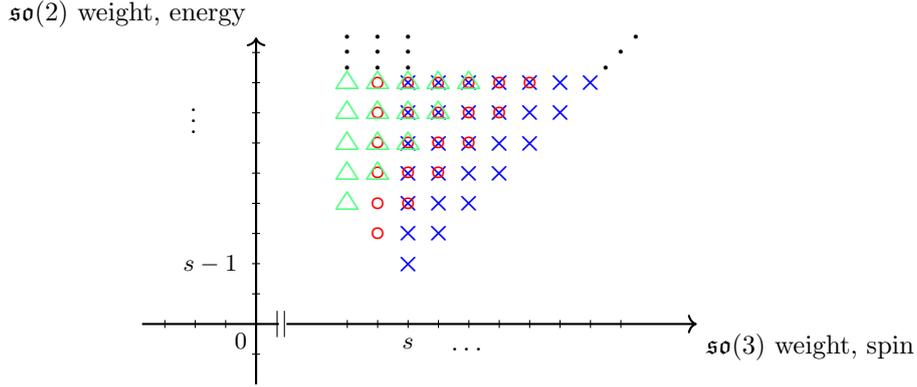
\begin{figure}[!ht]
    \center
    \begin{tikzpicture}

      \draw[thick] (-1.5,0) -- (0.3,0) node {$\,| |$};
      \draw[thick,->] (0.4,0) -- (5.8,0) node[anchor=north west]
           {$\so(3)$ weight, spin}; \draw[thick,->] (0,-0.8) --
           (0,3.8) node[anchor=south east] {$\so(2)$ weight, energy};
                
           \draw (-0.2, 0) node[below] {\small $0$};
           \draw (-0.4, 1.5pt) -- (-0.4, -1.5pt);  
           \draw (-0.8, 1.5pt) -- (-0.8, -1.5pt); 
           \draw (-1.2, 1.5pt) -- (-1.2, -1.5pt); 
           \draw (1.2, 1.5pt) -- (1.2, -1.5pt); 
           \draw (1.6, 1.5pt) -- (1.6, -1.5pt); 
           \draw (2, 1.5pt) -- (2, -1.5pt) node[below=0.2] {\small $s$}; 
           \draw (2.4, 1.5pt) -- (2.4, -1.5pt) ; 
           \draw (2.8, 1.5pt) -- (2.8, -1.5pt) node[below=4] {$\dots$}; 
           \draw (3.2, 1.5pt) -- (3.2, -1.5pt); 
           \draw (3.6, 1.5pt) -- (3.6, -1.5pt); 
           \draw (4, 1.5pt) -- (4, -1.5pt); 
           \draw (4.4, 1.5pt) -- (4.4, -1.5pt); 
           \draw (4.8, 1.5pt) -- (4.8, -1.5pt); 
           
           \draw (-1.5pt, 0.4) -- (1.5pt, 0.4); 
           \draw (-1.5pt, 0.8) -- (1.5pt, 0.8) node[left=5]{\small $s-1$};
           \draw (-1.5pt, 1.2) -- (1.5pt, 1.2);
           \draw (-1.5pt, 1.6) -- (1.5pt, 1.6);
           \draw (-1.5pt, 2) -- (1.5pt, 2);
           \draw (-1.5pt, 2.4) -- (1.5pt, 2.4); 
           \draw (-1.5pt, 2.8) -- (1.5pt, 2.8) node[left=20] {$\vdots$}; 
           \draw (-1.5pt, 3.2) -- (1.5pt, 3.2); 
           \draw (-1.5pt, 3.6) -- (1.5pt, 3.6); 
           \draw (-1.5pt, -0.4) -- (1.5pt, -0.4); 
           
           \draw (2, 0.8) node[color=blue] {$\boldsymbol{\times}$};
           \draw (2.4, 1.2) node[color=blue] {$\boldsymbol{\times}$};
           \draw (2.8, 1.6) node[color=blue] {$\boldsymbol{\times}$};
           \draw (3.2, 2) node[color=blue] {$\boldsymbol{\times}$};
           \draw (3.6, 2.4) node[color=blue] {$\boldsymbol{\times}$};
           \draw (4, 2.8) node[color=blue] {$\boldsymbol{\times}$};
           \draw (4.4, 3.2) node[color=blue] {$\boldsymbol{\times}$};
           \draw (4.6, 3.4) node[color=black] {$\boldsymbol{\cdot}$};
           \draw (4.8, 3.6) node[color=black] {$\boldsymbol{\cdot}$};
           \draw (5, 3.8) node[color=black] {$\boldsymbol{\cdot}$};

           \draw (2, 1.2) node[color=blue] {$\boldsymbol{\times}$};
           \draw (2.4, 1.6) node[color=blue] {$\boldsymbol{\times}$};
           \draw (2.8, 2) node[color=blue] {$\boldsymbol{\times}$};
           \draw (3.2, 2.4) node[color=blue] {$\boldsymbol{\times}$};
           \draw (3.6, 2.8) node[color=blue] {$\boldsymbol{\times}$};
           \draw (4, 3.2) node[color=blue] {$\boldsymbol{\times}$};

           \draw (2, 1.6) node[color=blue] {$\boldsymbol{\times}$};
           \draw (2.4, 2) node[color=blue] {$\boldsymbol{\times}$};
           \draw (2.8, 2.4) node[color=blue] {$\boldsymbol{\times}$};
           \draw (3.2, 2.8) node[color=blue] {$\boldsymbol{\times}$};
           \draw (3.6, 3.2) node[color=blue] {$\boldsymbol{\times}$};

           \draw (2, 2) node[color=blue] {$\boldsymbol{\times}$};
           \draw (2.4, 2.4) node[color=blue] {$\boldsymbol{\times}$};
           \draw (2.8, 2.8) node[color=blue] {$\boldsymbol{\times}$};
           \draw (3.2, 3.2) node[color=blue] {$\boldsymbol{\times}$};

           \draw (2, 2.4) node[color=blue] {$\boldsymbol{\times}$};
           \draw (2.4, 2.8) node[color=blue] {$\boldsymbol{\times}$};
           \draw (2.8, 3.2) node[color=blue] {$\boldsymbol{\times}$};

           \draw (2, 2.8) node[color=blue] {$\boldsymbol{\times}$};
           \draw (2.4, 3.2) node[color=blue] {$\boldsymbol{\times}$};

           \draw (2, 3.2) node[color=blue] {$\boldsymbol{\times}$};
           \draw (2, 3.4) node[color=black] {$\boldsymbol{\cdot}$};
           \draw (2, 3.6) node[color=black] {$\boldsymbol{\cdot}$};
           \draw (2, 3.8) node[color=black] {$\boldsymbol{\cdot}$};
           \draw (1.6, 1.2) node[color=red] {$\boldsymbol{\circ}$};
           \draw (2, 1.6) node[color=red] {$\boldsymbol{\circ}$};
           \draw (2.4, 2) node[color=red] {$\boldsymbol{\circ}$};
           \draw (2.8, 2.4) node[color=red] {$\boldsymbol{\circ}$};
           \draw (3.2, 2.8) node[color=red] {$\boldsymbol{\circ}$};
           \draw (3.6, 3.2) node[color=red] {$\boldsymbol{\circ}$};

           \draw (1.6, 1.6) node[color=red] {$\boldsymbol{\circ}$};
           \draw (2, 2) node[color=red] {$\boldsymbol{\circ}$};
           \draw (2.4, 2.4) node[color=red] {$\boldsymbol{\circ}$};
           \draw (2.8, 2.8) node[color=red] {$\boldsymbol{\circ}$};
           \draw (3.2, 3.2) node[color=red] {$\boldsymbol{\circ}$};

           \draw (1.6, 2) node[color=red] {$\boldsymbol{\circ}$};
           \draw (2, 2.4) node[color=red] {$\boldsymbol{\circ}$};
           \draw (2.4, 2.8) node[color=red] {$\boldsymbol{\circ}$};
           \draw (2.8, 3.2) node[color=red] {$\boldsymbol{\circ}$};

           \draw (1.6, 2.4) node[color=red] {$\boldsymbol{\circ}$};
           \draw (2, 2.8) node[color=red] {$\boldsymbol{\circ}$};
           \draw (2.4, 3.2) node[color=red] {$\boldsymbol{\circ}$};

           \draw (1.6, 2.8) node[color=red] {$\boldsymbol{\circ}$};
           \draw (2, 3.2) node[color=red] {$\boldsymbol{\circ}$};

           \draw (1.6, 3.2) node[color=red] {$\boldsymbol{\circ}$};
           \draw (1.6, 3.4) node[color=black] {$\boldsymbol{\cdot}$};
           \draw (1.6, 3.6) node[color=black] {$\boldsymbol{\cdot}$};
           \draw (1.6, 3.8) node[color=black] {$\boldsymbol{\cdot}$};
           \draw (1.2, 1.6) node[color=mygreen] {$\boldsymbol{\triangle}$};
           \draw (1.6, 2) node[color=mygreen] {$\boldsymbol{\triangle}$};
           \draw (2, 2.4) node[color=mygreen] {$\boldsymbol{\triangle}$};
           \draw (2.4, 2.8) node[color=mygreen] {$\boldsymbol{\triangle}$};
           \draw (2.8, 3.2) node[color=mygreen] {$\boldsymbol{\triangle}$};

           \draw (1.2, 2) node[color=mygreen] {$\boldsymbol{\triangle}$};
           \draw (1.6, 2.4) node[color=mygreen] {$\boldsymbol{\triangle}$};
           \draw (2, 2.8) node[color=mygreen] {$\boldsymbol{\triangle}$};
           \draw (2.4, 3.2) node[color=mygreen] {$\boldsymbol{\triangle}$};

           \draw (1.2, 2.4) node[color=mygreen] {$\boldsymbol{\triangle}$};
           \draw (1.6, 2.8) node[color=mygreen] {$\boldsymbol{\triangle}$};
           \draw (2, 3.2) node[color=mygreen] {$\boldsymbol{\triangle}$};

           \draw (1.2, 2.8) node[color=mygreen] {$\boldsymbol{\triangle}$};
           \draw (1.6, 3.2) node[color=mygreen] {$\boldsymbol{\triangle}$};

           \draw (1.2, 3.2) node[color=mygreen] {$\boldsymbol{\triangle}$};
           \draw (1.2, 3.4) node[color=black] {$\boldsymbol{\cdot}$};
           \draw (1.2, 3.6) node[color=black] {$\boldsymbol{\cdot}$};
           \draw (1.2, 3.8) node[color=black] {$\boldsymbol{\cdot}$};
    \end{tikzpicture}
    \caption{Weight diagram of a spin-$s$ partially massless field of
      depth $\tau=3$. The contribution of the sum over $k$ in
      \eqref{decompo_pm_3} are represented by: blue crosses
      $\color{blue} \times$ for $k=0$, red circles $\color{red} \circ$
      for $k=1$ and green triangles $\color{mygreen} \triangle$ for
      $k=2$.}
    \label{fig:weight_pm}
  \end{figure}

  Now starting with the $\so(2)\oplus\so(4)$ decomposition
  \eqref{decompo_ho_hs_sing} of the spin-$s$ depth-$t$ PM $\so(2,4)$
  module, we can derive its $\so(2)\oplus\so(3)$ decomposition:
  \begin{eqnarray}
    \D\big( s+2-t\,;\, \blb s,s \brb \big) & \cong &
    \bigoplus_{\ell=0}^{t-1} \bigoplus_{n=0}^{t-1-\ell}
    \bigoplus_{\sigma=\ell}^\infty \D_{\so(2) \oplus \so(4)}\big(
    s+2-t+\sigma+2n\,;\, \blb s+\sigma-\ell,s-\ell \brb \big) \\ &
    \branchmod{4}{3} & \bigoplus_{\ell=0}^{t-1}
    \bigoplus_{n=0}^{t-1-\ell} \bigoplus_{\sigma=\ell}^\infty
    \bigoplus_{k=0}^\sigma \D_{\so(2) \oplus \so(3)}\big(
    s+2-t+\sigma+2n\,;\, \blb s+k-\ell \brb \big) \label{interm_dec}
    \\ & \cong & \bigoplus_{\tau=1}^t \bigoplus_{k=0}^{\tau-1}
    \bigoplus_{\sigma=0}^\infty \bigoplus_{n=0}^\infty \D_{\so(2)
      \oplus \so(3)}\big( s+2-\tau+n+\sigma+k\,;\, \blb s+\sigma-k
    \brb \big)\, ,
  \end{eqnarray}
  which matches the direct sum of the $\so(2)\oplus\so(3)$
  decomposition of the spin-$s$ partially massless modules of depth
  $\tau=1, \dots, t$, i.e.
  \begin{equation}
    \D\big( s+2-t\,;\, \blb s,s \brb \big) \quad \branchmod{4}{3}
    \quad \bigoplus_{\tau=1}^t \D\big( s+2-\tau\,;\, \blb s \brb
    \big)\, .
  \end{equation}
  This branching rule can also be represented graphically, by drawing
  on the one hand the $\so(2)\oplus\so(3)$ weight diagram of the
  spin-$s$ and depth-$t$ RPM field as read from \eqref{interm_dec} and
  on the other hand by drawing the $\so(2)\oplus\so(3)$ weight
  diagrams of the partially massless spin-$s$ modules of depth
  $\tau=1, \dots, t$, and comparing the two diagrams. This is done for
  the $t=2$ case in \hyperref[fig:weight_ho_singleton]{Figure
    \ref{fig:weight_ho_singleton}} below.
  \begin{figure}[!ht]
    \begin{minipage}[c]{0.4\textwidth}
      \center
      \begin{tikzpicture}
        \draw[thick] (-1.5,0) -- (0.3,0) node {$\,| |$};
        \draw[thick,->] (0.4,0) -- (5.8,0) node[anchor=north west]
             {$\ell_1$}; \draw[thick,->] (0,-0.8) --
             (0,3.8) node[anchor=south east] {$E$};
             
           \draw (-0.2, 0) node[below] {\small $0$};
           \draw (-0.4, 1.5pt) -- (-0.4, -1.5pt);  
           \draw (-0.8, 1.5pt) -- (-0.8, -1.5pt); 
           \draw (-1.2, 1.5pt) -- (-1.2, -1.5pt); 
           \draw (1.2, 1.5pt) -- (1.2, -1.5pt); 
           \draw (1.6, 1.5pt) -- (1.6, -1.5pt); 
           \draw (2, 1.5pt) -- (2, -1.5pt) node[below=0.2] {\small $s$}; 
           \draw (2.4, 1.5pt) -- (2.4, -1.5pt) ; 
           \draw (2.8, 1.5pt) -- (2.8, -1.5pt) node[below=4] {$\dots$}; 
           \draw (3.2, 1.5pt) -- (3.2, -1.5pt); 
           \draw (3.6, 1.5pt) -- (3.6, -1.5pt); 
           \draw (4, 1.5pt) -- (4, -1.5pt); 
           \draw (4.4, 1.5pt) -- (4.4, -1.5pt); 
           \draw (4.8, 1.5pt) -- (4.8, -1.5pt); 
           
           \draw (-1.5pt, 0.4) -- (1.5pt, 0.4); 
           \draw (-1.5pt, 0.8) -- (1.5pt, 0.8) node[left=5]{\small $s$};
           \draw (-1.5pt, 1.2) -- (1.5pt, 1.2);
           \draw (-1.5pt, 1.6) -- (1.5pt, 1.6);
           \draw (-1.5pt, 2) -- (1.5pt, 2);
           \draw (-1.5pt, 2.4) -- (1.5pt, 2.4); 
           \draw (-1.5pt, 2.8) -- (1.5pt, 2.8) node[left=20] {$\vdots$}; 
           \draw (-1.5pt, 3.2) -- (1.5pt, 3.2); 
           \draw (-1.5pt, 3.6) -- (1.5pt, 3.6); 
           \draw (-1.5pt, -0.4) -- (1.5pt, -0.4); 
           
           \draw (2, 0.8) node[color=blue] {$\boldsymbol{\times}$};
           \draw (2.4, 1.2) node[color=blue] {$\boldsymbol{\times}$};
           \draw (2.8, 1.6) node[color=blue] {$\boldsymbol{\times}$};
           \draw (3.2, 2) node[color=blue] {$\boldsymbol{\times}$};
           \draw (3.6, 2.4) node[color=blue] {$\boldsymbol{\times}$};
           \draw (4, 2.8) node[color=blue] {$\boldsymbol{\times}$};
           \draw (4.4, 3.2) node[color=blue] {$\boldsymbol{\times}$};
           \draw (4.6, 3.4) node[color=black] {$\boldsymbol{\cdot}$};
           \draw (4.8, 3.6) node[color=black] {$\boldsymbol{\cdot}$};
           \draw (5, 3.8) node[color=black] {$\boldsymbol{\cdot}$};

           \draw (2, 1.2) node[color=blue] {$\boldsymbol{\times}$};
           \draw (2.4, 1.6) node[color=blue] {$\boldsymbol{\times}$};
           \draw (2.8, 2) node[color=blue] {$\boldsymbol{\times}$};
           \draw (3.2, 2.4) node[color=blue] {$\boldsymbol{\times}$};
           \draw (3.6, 2.8) node[color=blue] {$\boldsymbol{\times}$};
           \draw (4, 3.2) node[color=blue] {$\boldsymbol{\times}$};

           \draw (2, 1.6) node[color=blue] {$\boldsymbol{\times}$};
           \draw (2.4, 2) node[color=blue] {$\boldsymbol{\times}$};
           \draw (2.8, 2.4) node[color=blue] {$\boldsymbol{\times}$};
           \draw (3.2, 2.8) node[color=blue] {$\boldsymbol{\times}$};
           \draw (3.6, 3.2) node[color=blue] {$\boldsymbol{\times}$};

           \draw (2, 2) node[color=blue] {$\boldsymbol{\times}$};
           \draw (2.4, 2.4) node[color=blue] {$\boldsymbol{\times}$};
           \draw (2.8, 2.8) node[color=blue] {$\boldsymbol{\times}$};
           \draw (3.2, 3.2) node[color=blue] {$\boldsymbol{\times}$};

           \draw (2, 2.4) node[color=blue] {$\boldsymbol{\times}$};
           \draw (2.4, 2.8) node[color=blue] {$\boldsymbol{\times}$};
           \draw (2.8, 3.2) node[color=blue] {$\boldsymbol{\times}$};

           \draw (2, 2.8) node[color=blue] {$\boldsymbol{\times}$};
           \draw (2.4, 3.2) node[color=blue] {$\boldsymbol{\times}$};

           \draw (2, 3.2) node[color=blue] {$\boldsymbol{\times}$};
           \draw (2, 3.4) node[color=black] {$\boldsymbol{\cdot}$};
           \draw (2, 3.6) node[color=black] {$\boldsymbol{\cdot}$};
           \draw (2, 3.8) node[color=black] {$\boldsymbol{\cdot}$};
           \draw (1.6, 1.2) node[color=red] {$\boldsymbol{\circ}$};
           \draw (2, 1.6) node[color=red] {$\boldsymbol{\circ}$};
           \draw (2.4, 2) node[color=red] {$\boldsymbol{\circ}$};
           \draw (2.8, 2.4) node[color=red] {$\boldsymbol{\circ}$};
           \draw (3.2, 2.8) node[color=red] {$\boldsymbol{\circ}$};
           \draw (3.6, 3.2) node[color=red] {$\boldsymbol{\circ}$};

           \draw (2, 1.2) node[color=red] {$\boldsymbol{\circ}$};
           \draw (2.4, 1.6) node[color=red] {$\boldsymbol{\circ}$};
           \draw (2.8, 2) node[color=red] {$\boldsymbol{\circ}$};
           \draw (3.2, 2.4) node[color=red] {$\boldsymbol{\circ}$};
           \draw (3.6, 2.8) node[color=red] {$\boldsymbol{\circ}$};
           \draw (4, 3.2) node[color=red] {$\boldsymbol{\circ}$};

           \draw (1.6, 1.6) node[color=red] {$\boldsymbol{\circ}$};
           \draw (2, 2) node[color=red] {$\boldsymbol{\circ}$};
           \draw (2.4, 2.4) node[color=red] {$\boldsymbol{\circ}$};
           \draw (2.8, 2.8) node[color=red] {$\boldsymbol{\circ}$};
           \draw (3.2, 3.2) node[color=red] {$\boldsymbol{\circ}$};

           \draw (1.6, 2) node[color=red] {$\boldsymbol{\circ}$};
           \draw (2, 2.4) node[color=red] {$\boldsymbol{\circ}$};
           \draw (2.4, 2.8) node[color=red] {$\boldsymbol{\circ}$};
           \draw (2.8, 3.2) node[color=red] {$\boldsymbol{\circ}$};

           \draw (1.6, 2.4) node[color=red] {$\boldsymbol{\circ}$};
           \draw (2, 2.8) node[color=red] {$\boldsymbol{\circ}$};
           \draw (2.4, 3.2) node[color=red] {$\boldsymbol{\circ}$};

           \draw (1.6, 2.8) node[color=red] {$\boldsymbol{\circ}$};
           \draw (2, 3.2) node[color=red] {$\boldsymbol{\circ}$};

           \draw (1.6, 3.2) node[color=red] {$\boldsymbol{\circ}$};
           \draw (1.6, 3.4) node[color=black] {$\boldsymbol{\cdot}$};
           \draw (1.6, 3.6) node[color=black] {$\boldsymbol{\cdot}$};
           \draw (1.6, 3.8) node[color=black] {$\boldsymbol{\cdot}$};
           \draw (2, 1.6) node[color=mygreen] {$\boldsymbol{\triangle}$};
           \draw (2.4, 2) node[color=mygreen] {$\boldsymbol{\triangle}$};
           \draw (2.8, 2.4) node[color=mygreen] {$\boldsymbol{\triangle}$};
           \draw (3.2, 2.8) node[color=mygreen] {$\boldsymbol{\triangle}$};
           \draw (3.6, 3.2) node[color=mygreen] {$\boldsymbol{\triangle}$};

           \draw (2, 2) node[color=mygreen] {$\boldsymbol{\triangle}$};
           \draw (2.4, 2.4) node[color=mygreen] {$\boldsymbol{\triangle}$};
           \draw (2.8, 2.8) node[color=mygreen] {$\boldsymbol{\triangle}$};
           \draw (3.2, 3.2) node[color=mygreen] {$\boldsymbol{\triangle}$};

           \draw (2, 2.4) node[color=mygreen] {$\boldsymbol{\triangle}$};
           \draw (2.4, 2.8) node[color=mygreen] {$\boldsymbol{\triangle}$};
           \draw (2.8, 3.2) node[color=mygreen] {$\boldsymbol{\triangle}$};

           \draw (2, 2.8) node[color=mygreen] {$\boldsymbol{\triangle}$};
           \draw (2.4, 3.2) node[color=mygreen] {$\boldsymbol{\triangle}$};

           \draw (2, 3.2) node[color=mygreen] {$\boldsymbol{\triangle}$};
      \end{tikzpicture}
    \end{minipage}
    \qquad \qquad $\leftrightarrow$ \qquad
    \begin{minipage}[r]{0.4\textwidth}
      \center
      \begin{tikzpicture}

        \draw[thick] (-1.5,0) -- (0.3,0) node {$\,| |$};
        \draw[thick,->] (0.4,0) -- (5.8,0) node[anchor=north west]
             {$\ell_1$}; \draw[thick,->] (0,-0.8) --
             (0,3.8) node[anchor=south east] {$E$};
                
           \draw (-0.2, 0) node[below] {\small $0$};
           \draw (-0.4, 1.5pt) -- (-0.4, -1.5pt);  
           \draw (-0.8, 1.5pt) -- (-0.8, -1.5pt); 
           \draw (-1.2, 1.5pt) -- (-1.2, -1.5pt); 
           \draw (1.2, 1.5pt) -- (1.2, -1.5pt); 
           \draw (1.6, 1.5pt) -- (1.6, -1.5pt); 
           \draw (2, 1.5pt) -- (2, -1.5pt) node[below=0.2] {\small $s$}; 
           \draw (2.4, 1.5pt) -- (2.4, -1.5pt) ; 
           \draw (2.8, 1.5pt) -- (2.8, -1.5pt) node[below=4] {$\dots$}; 
           \draw (3.2, 1.5pt) -- (3.2, -1.5pt); 
           \draw (3.6, 1.5pt) -- (3.6, -1.5pt); 
           \draw (4, 1.5pt) -- (4, -1.5pt); 
           \draw (4.4, 1.5pt) -- (4.4, -1.5pt); 
           \draw (4.8, 1.5pt) -- (4.8, -1.5pt); 
           
           \draw (-1.5pt, 0.4) -- (1.5pt, 0.4); 
           \draw (-1.5pt, 0.8) -- (1.5pt, 0.8) node[left=5]{\small $s$};
           \draw (-1.5pt, 1.2) -- (1.5pt, 1.2);
           \draw (-1.5pt, 1.6) -- (1.5pt, 1.6);
           \draw (-1.5pt, 2) -- (1.5pt, 2);
           \draw (-1.5pt, 2.4) -- (1.5pt, 2.4); 
           \draw (-1.5pt, 2.8) -- (1.5pt, 2.8) node[left=20] {$\vdots$}; 
           \draw (-1.5pt, 3.2) -- (1.5pt, 3.2); 
           \draw (-1.5pt, 3.6) -- (1.5pt, 3.6); 
           \draw (-1.5pt, -0.4) -- (1.5pt, -0.4); 
           
           \draw (2, 0.8) node[color=black] {$\boldsymbol{\times}$};
           \draw (2.4, 1.2) node[color=black] {$\boldsymbol{\times}$};
           \draw (2.8, 1.6) node[color=black] {$\boldsymbol{\times}$};
           \draw (3.2, 2) node[color=black] {$\boldsymbol{\times}$};
           \draw (3.6, 2.4) node[color=black] {$\boldsymbol{\times}$};
           \draw (4, 2.8) node[color=black] {$\boldsymbol{\times}$};
           \draw (4.4, 3.2) node[color=black] {$\boldsymbol{\times}$};
           \draw (4.6, 3.4) node[color=black] {$\boldsymbol{\cdot}$};
           \draw (4.8, 3.6) node[color=black] {$\boldsymbol{\cdot}$};
           \draw (5, 3.8) node[color=black] {$\boldsymbol{\cdot}$};

           \draw (2, 1.2) node[color=black] {$\boldsymbol{\times}$};
           \draw (2.4, 1.6) node[color=black] {$\boldsymbol{\times}$};
           \draw (2.8, 2) node[color=black] {$\boldsymbol{\times}$};
           \draw (3.2, 2.4) node[color=black] {$\boldsymbol{\times}$};
           \draw (3.6, 2.8) node[color=black] {$\boldsymbol{\times}$};
           \draw (4, 3.2) node[color=black] {$\boldsymbol{\times}$};

           \draw (2, 1.6) node[color=black] {$\boldsymbol{\times}$};
           \draw (2.4, 2) node[color=black] {$\boldsymbol{\times}$};
           \draw (2.8, 2.4) node[color=black] {$\boldsymbol{\times}$};
           \draw (3.2, 2.8) node[color=black] {$\boldsymbol{\times}$};
           \draw (3.6, 3.2) node[color=black] {$\boldsymbol{\times}$};

           \draw (2, 2) node[color=black] {$\boldsymbol{\times}$};
           \draw (2.4, 2.4) node[color=black] {$\boldsymbol{\times}$};
           \draw (2.8, 2.8) node[color=black] {$\boldsymbol{\times}$};
           \draw (3.2, 3.2) node[color=black] {$\boldsymbol{\times}$};

           \draw (2, 2.4) node[color=black] {$\boldsymbol{\times}$};
           \draw (2.4, 2.8) node[color=black] {$\boldsymbol{\times}$};
           \draw (2.8, 3.2) node[color=black] {$\boldsymbol{\times}$};

           \draw (2, 2.8) node[color=black] {$\boldsymbol{\times}$};
           \draw (2.4, 3.2) node[color=black] {$\boldsymbol{\times}$};

           \draw (2, 3.2) node[color=black] {$\boldsymbol{\times}$};
           \draw (2, 3.4) node[color=black] {$\boldsymbol{\cdot}$};
           \draw (2, 3.6) node[color=black] {$\boldsymbol{\cdot}$};
           \draw (2, 3.8) node[color=black] {$\boldsymbol{\cdot}$};
           \draw (1.6, 1.2) node[color=black] {$\boldsymbol{\+}$};
           \draw (2, 1.6) node[color=blue] {$\boldsymbol{\circ}$};
           \draw (2.4, 2) node[color=blue] {$\boldsymbol{\circ}$};
           \draw (2.8, 2.4) node[color=blue] {$\boldsymbol{\circ}$};
           \draw (3.2, 2.8) node[color=blue] {$\boldsymbol{\circ}$};
           \draw (3.6, 3.2) node[color=blue] {$\boldsymbol{\circ}$};

           \draw (2, 1.2) node[color=blue] {$\boldsymbol{\circ}$};
           \draw (2.4, 1.6) node[color=blue] {$\boldsymbol{\circ}$};
           \draw (2.8, 2) node[color=blue] {$\boldsymbol{\circ}$};
           \draw (3.2, 2.4) node[color=blue] {$\boldsymbol{\circ}$};
           \draw (3.6, 2.8) node[color=blue] {$\boldsymbol{\circ}$};
           \draw (4, 3.2) node[color=blue] {$\boldsymbol{\circ}$};

           \draw (1.6, 1.6) node[color=black] {$\boldsymbol{\+}$};
           \draw (2, 2) node[color=blue] {$\boldsymbol{\circ}$};
           \draw (2.4, 2.4) node[color=blue] {$\boldsymbol{\circ}$};
           \draw (2.8, 2.8) node[color=blue] {$\boldsymbol{\circ}$};
           \draw (3.2, 3.2) node[color=blue] {$\boldsymbol{\circ}$};

           \draw (1.6, 2) node[color=black] {$\boldsymbol{\+}$};
           \draw (2, 2.4) node[color=blue] {$\boldsymbol{\circ}$};
           \draw (2.4, 2.8) node[color=blue] {$\boldsymbol{\circ}$};
           \draw (2.8, 3.2) node[color=blue] {$\boldsymbol{\circ}$};

           \draw (1.6, 2.4) node[color=black] {$\boldsymbol{\+}$};
           \draw (2, 2.8) node[color=blue] {$\boldsymbol{\circ}$};
           \draw (2.4, 3.2) node[color=blue] {$\boldsymbol{\circ}$};

           \draw (1.6, 2.8) node[color=black] {$\boldsymbol{\+}$};
           \draw (2, 3.2) node[color=blue] {$\boldsymbol{\circ}$};

           \draw (1.6, 3.2) node[color=black] {$\boldsymbol{\+}$};
           \draw (1.6, 3.4) node[color=black] {$\boldsymbol{\cdot}$};
           \draw (1.6, 3.6) node[color=black] {$\boldsymbol{\cdot}$};
           \draw (1.6, 3.8) node[color=black] {$\boldsymbol{\cdot}$};
           \draw (2, 1.6) node[color=black] {$\boldsymbol{\+}$};
           \draw (2.4, 2) node[color=black] {$\boldsymbol{\+}$};
           \draw (2.8, 2.4) node[color=black] {$\boldsymbol{\+}$};
           \draw (3.2, 2.8) node[color=black] {$\boldsymbol{\+}$};
           \draw (3.6, 3.2) node[color=black] {$\boldsymbol{\+}$};

           \draw (2, 2) node[color=black] {$\boldsymbol{\+}$};
           \draw (2.4, 2.4) node[color=black] {$\boldsymbol{\+}$};
           \draw (2.8, 2.8) node[color=black] {$\boldsymbol{\+}$};
           \draw (3.2, 3.2) node[color=black] {$\boldsymbol{\+}$};

           \draw (2, 2.4) node[color=black] {$\boldsymbol{\+}$};
           \draw (2.4, 2.8) node[color=black] {$\boldsymbol{\+}$};
           \draw (2.8, 3.2) node[color=black] {$\boldsymbol{\+}$};

           \draw (2, 2.8) node[color=black] {$\boldsymbol{\+}$};
           \draw (2.4, 3.2) node[color=black] {$\boldsymbol{\+}$};

           \draw (2, 3.2) node[color=black] {$\boldsymbol{\+}$};
      \end{tikzpicture}
    \end{minipage}
    \caption{Left: $\so(2)\oplus\so(3)$ weight diagram of the spin-$s$
      and depth $t=2$ PM field. The contribution to the sum over
      $\ell$ and $n$ in \eqref{interm_dec} are represented by: blue
      crosses $\color{blue} \times$ for $\ell=n=0$, green triangles
      $\color{mygreen} \triangle$ for $\ell=0$, $n=1$ and red circles
      $\color{red} \circ$ for $\ell=1$, $n=0$. The lowest
      energy/$\so(2)$ weight in this diagram is $s+2-t=s$ for
      $t=2$. Right: Superimposed $\so(2)\oplus\so(3)$ weight diagrams
      of the spin-$s$ and depth $\tau=1$ (blue circles $\color{blue}
      \circ$) and $\tau=2$ (black crosses $\times$ and $\+$) partially
      massless modules. The lowest energy/$\so(2)$ weight in this
      diagram is $s+2-\tau=s$ for $\tau=2$.}
    \label{fig:weight_ho_singleton}
  \end{figure}
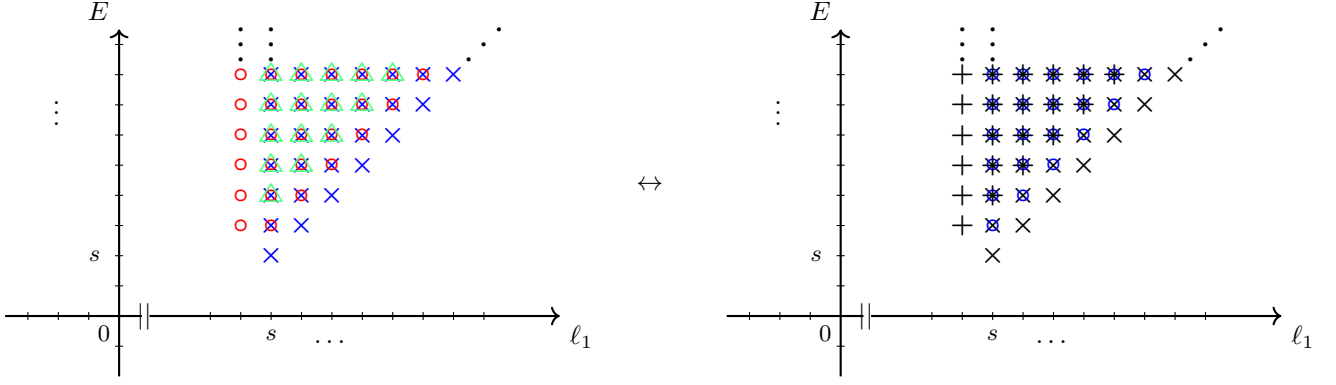

\end{proof}

\begin{remark}
  Notice that the previous
  \hyperref[prop:branching_rule_ho]{Proposition
    \ref{prop:branching_rule_ho}} encompasses the case of unitary
  higher-spin singleton, corresponding to $t=1$. The above
  decomposition reduce, in this special case $t=1$ to those previously
  derived and summed up in \hyperref[th:uni_singleton]{Theorem
    \ref{th:uni_singleton}}
\end{remark}

\paragraph{From $\so(2,d)$ to $\iso(1,d-1)$.}
Again assuming that the diagram \eqref{com_diagram} is commutative,
the branching of the spin-$s$ and depth-$t$ RPM can be obtained by
performing an In\"on\"u-Wigner contraction of the $\so(2,d-1)$
modules. Applying the BMV mechanism to a $\so(2,d-1)$ partially
massless fields of depth-$t$ and spin given by a maximal height
rectangular Young diagram yields \cite{Brink:2000ag, Boulanger:2008up,
  Boulanger:2008kw}:
\begin{equation}
  \D_{\so(2,d-1)}\big(s+d-r-t\,;\, \blb s^{r-1} \brb\big) \quad
  \flimit \quad \bigoplus_{\tau=0}^{t-1} \ \D_{\iso(1,d-1)}\big(
  m=0\,;\, \blb s^{r-2}, s-\tau \brb\big)\, .
\end{equation}
As a consequence, the branching rule of the $\so(2,d)$ spin-$s$ and
depth-$t$ RPM module onto $\iso(1,d-1)$ reads:
\begin{equation}
  \D_{\so(2,d)}\big(s+r-t\,;\, \blb s^r \brb\big) \quad \branchiso
  \quad \bigoplus_{\tau=0}^{t-1} \ (t-\tau) \ \D_{\iso(1,d-1)}\big(
  m=0 \,;\, \blb s^{r-2}, s-\tau \brb\big)\, .
  \label{branching_ho_iso}
\end{equation}

At this point, a few comments are in order. As emphasised in the first
part of this section, the crucial properties of unitary singletons is
that $(i)$ they constitute the class of representations that can be
lifted from $\so(2,d-1)$ to $\so(2,d)$, i.e. they are AdS fields that
are also conformal, and $(ii)$ they describe AdS fields which are
``confined'' to its (conformal) boundary. The first property
translates, for unitary singletons, into the fact that these
$\so(2,d)$ modules remain irreducible when restricted to $\so(2,d-1)$
-- except in the case of the scalar singleton whose branching rule
actually contains two modules. The second property is related to the
fact that the singleton modules also remain irreducible when further
contracting to the Poincar\'e algebra $\iso(1,d-1)$ (thereby
indicating that the AdS$_{d+1}$ field does not propagate degrees of
freedom in the bulk).\\

In the case of the RPM fields of spin-$s$ and depth-$t$ studied in the
present note, it seems difficult to consider them as a suitable
higher-order (i.e. non-unitary) extension of higher-spin singletons
due to the fact that their branching rule \eqref{branching_ho} shows
the appearance of $t$ modules. Indeed, the presence of multiple
modules in \eqref{branching_ho} for $t>1$ prevent us from reading this
decomposition ``backward'' (from right to left) as the property for a
{\it single} field in AdS$_d$ corresponding to a $\so(2,d-1)$ module
that can be lifted to a $\so(2,d)$ module thereby illustrating that
this AdS$_d$ field is also conformal. Notice that this is in
accordance with \cite{Metsaev:1995jp} where conformal AdS fields were
classified, and confirmed in the more recent analysis
\cite{Barnich:2015tma} where, without insisting on unitarity, the
authors were lead to rule out partially massless fields from the class
of AdS fields which can be lifted to conformal representations. On top
of that, the contraction of \eqref{branching_ho} to $\iso(1,d-1)$
given in \eqref{branching_ho_iso} produces several modules, some of
them even appearing with a multiplicity greater than one, which seems
to indicate that the ``confinement'' property of unitary singletons is
also lost when relaxing the unitarity condition in the way proposed
here (i.e. considering the modules $\D\big(s+\tfrac d2-t;\blb s^r
\brb\big)$ with $t>1$). It would nevertheless be interesting to study
a field theoretical realisation of these modules to explicitely see
how this property is lost when passing from $t=1$ to $t>1$.

\section{Flato-Fr\o{}nsdal theorem}
\label{sec:flato_fronsdal}
Let us now particularise the discussion to the $d=4$ case, where we
can take advantage of the low dimensional isomorphism $\so(4) \cong
\so(3) \oplus \so(3)$ to decompose the tensor product of two spin-$s$
and depth-$t$ RPM fields. \\

The tensor product of two higher-spin unitary singletons was
considered (in arbitrary dimensions) in \cite{Dolan:2005wy}, and reads
in the special case $d=4$:
\begin{equation}
  \D\big(s+1; \blb s, s \brb_0\big)^{\otimes 2} \cong
  \bigoplus_{\sigma=0}^{2s} \D\big(2s+2; \blb \sigma, \sigma
  \brb_0\big) \oplus \bigoplus_{\sigma=2s+1}^\infty \D\big(\sigma+2;
  \blb \sigma, 2s \brb_0\big) \oplus \bigoplus_{\sigma=2s}^\infty\,
  2\, \D\big(\sigma+2; \blb \sigma \brb\big)\, .
  \label{unitary_tensor_product}
\end{equation}
Considering singletons of fixed chirality, i.e. $\D\big(s+1; \blb s, s
\brb_\epsilon\big)$, the decomposition of their tensor product then
reads:
\begin{equation}
  \D\big(s+1; \blb s, s \brb_\epsilon\big)^{\otimes 2} \cong
  \bigoplus_{\sigma=0}^{2s} \D\big(2s+2; \blb \sigma, \sigma
  \brb_\epsilon\big) \oplus \bigoplus_{\sigma=2s+1}^\infty
  \D\big(\sigma+2; \blb \sigma, 2s \brb_\epsilon\big)\, ,
\end{equation}
i.e. it contributes to the above tensor product by producing the
infinite tower of mixed symmetry massless fields $\D\big(\sigma+2;
\blb \sigma, 2s \brb_\epsilon\big)$ and the finite tower of massive
fields $\D\big( 2s+2; \blb \sigma, \sigma \brb_\epsilon\big)$. The
tensor product of two spin-$s$ singletons of opposite chirality, on
the other hand, contribute to \eqref{unitary_tensor_product} by
producing the infinite tower of totally symmetric massless fields
$\D\big( \sigma+2; \blb \sigma \brb\big)$:
\begin{equation}
  \D\big(s+1; \blb s, s \brb_+\big) \otimes \D\big(s+1; \blb s, s
  \brb_-\big) \cong \bigoplus_{\sigma=2s}^\infty \D\big( \sigma+2;
  \blb \sigma \brb\big)\, .
\end{equation}

\begin{remark}
  The Higher-Spin algebra on which such a theory is based
  \cite{Bae:2016xmv} can be decomposed as:
  \begin{equation}
    \hs_s^{(4)} \quad \cong \quad \bigoplus_{\sigma=2s}^\infty \quad
       {\footnotesize \gyoung(_5{\sigma-1},_5{\sigma-1})} \quad \oplus
       \quad \bigoplus_{\sigma=2s+1}^\infty \quad {\footnotesize
         \gyoung(_5{\sigma-1},_5{\sigma-1},_3{2s})}\, ,
  \end{equation}
  In other words, it is composed of all the Killing tensor of the
  massless fields appearing in the decomposition of two spin-$s$
  singletons.
\end{remark}

The tensor product of two higher-order Dirac singletons was worked out
in arbitrary dimensions in \cite{Bekaert:2013zya, Basile:2014wua}, and
hereafter we give the decomposition for the tensor product of two
spin-$s$ and depth-$t$ RPM fields, considered as a possible
generalisation of those higher-order singletons, in the special case
$d=4$.

\begin{theorem}[Flato-Fr\o{}nsdal theorem for rectangular partially massless fields]
  \label{th:flato_fronsdal}
  The tensor product of two $\so(2,4)$ rectangular partially massless
  fields of spin-$s$ and depth-$t$ decomposes as:
  \begin{itemize}
  \item If they are of the same chirality $\epsilon$:
    \begin{eqnarray}
      \D\big(s+2-t; \blb s, s \brb_\epsilon\big)^{\otimes 2} & \cong &
      \bigoplus_{\tau=1}^{t} \bigoplus_{m=0}^{2(\tau-1)}
      \bigoplus_{n=0}^{\nu_m^\tau}
      \bigoplus_{\sigma=2s+2\tau-1-m-n}^\infty
      \D\big(\sigma+4-2\tau+m; \blb \sigma, 2s-m \brb_\epsilon\big)
      \nonumber \\ && \quad \oplus \bigoplus_{\tau=1}^{t}
      \bigoplus_{m=0}^{2(\tau-1)} \bigoplus_{n=0}^{\nu_m^\tau}
      \bigoplus_{k=-\mu_{m,n}^\tau}^{2s-n} \D\big(2s+4-2\tau+m; \blb
      k+m, k \brb_\epsilon\big)\, ,
      \label{same_chirality}
    \end{eqnarray}
    where $\nu_m^\tau := \min(m,2(\tau-1)-m)$ and 
    \begin{equation}
      \mu_{m,n}^\tau := \left\{
      \begin{aligned}
        \min(n, \nu_m^\tau-n) \qquad \, , & \quad \text{if} \quad m <
        \tau\, , \\ m-\tau+1+\min(n,\nu_m^\tau-n)\, , & \quad
        \text{if} \quad m\geqslant \tau\, .
      \end{aligned}
      \right.
    \end{equation}
  \item If they are of opposite chirality:
    \begin{equation}
      \D\big(s+2-t; \blb s, s \brb_+\big) \otimes \D\big(s+2-t; \blb
      s, s \brb_-\big) \cong \bigoplus_{\tau=1}^t
      \bigoplus_{m=0}^{t-\tau} \bigoplus_{n=0}^{t-\tau-m}
      \bigoplus_{\sigma=2s-m}^\infty \D\big(\sigma+4-2\tau-n; \blb
      \sigma, n \brb_0\big)\, .
      \label{opposite_chirality}
    \end{equation}
    Notice that in the above decomposition \eqref{opposite_chirality}
    of two singletons of opposite chirality, the irreps describing
    {\rm totally symmetric} partially massless fields, i.e. of spin
    given by a single row Young diagram, only appear once despite what
    the notation $\blb \sigma \brb_0$ would normally suggests.
  \end{itemize}
\end{theorem}
\begin{proof}
  In order to prove the above decomposition, we will use the two
  expressions of the character of a spin-$s$ and depth-$t$ RPM:
  \begin{eqnarray}
    \chi^{\so(2,4)}_{[s+2-t;(s,s)]}(q, \vec x) & = & q^{s+2-t} \Big(
    \chi^{\so(4)}_{(s,s)}(\vec x) - q^t \chi^{\so(4)}_{(s,s-t)}(\vec
    x) + q^{t+1} \chi^{\so(4)}_{(s-1,s-t)}(\vec x) \Big) \Pd 4 (q,
    \vec x) \label{exp1} \\ & = & \sum_{\ell=0}^{t-1}
    \sum_{n=0}^{t-1-\ell} \sum_{\sigma=\ell}^\infty q^{s+2-t+\sigma+2n}
    \chi^{\so(4)}_{(s-\ell+\sigma,s-\ell)}(\vec x)\, , \label{exp2}
  \end{eqnarray}
  and will decompose their product as the sum of the characters of the
  different modules appearing in \eqref{same_chirality} and
  \eqref{opposite_chirality}. To do so, the idea is simply to look at
  the product of \eqref{exp1} and \eqref{exp2}, decompose the tensor
  product of the $\so(4)$ characters, and finally recognize the
  resulting expression as a sum of characters of:
  \begin{itemize}
  \item Partially massless fields of depth-$\tau$ and spin given by a
    two-row Young diagram $\blb \sigma, n \brb$ which read:
    \begin{equation}
      \chi^{\so(2,4)}_{[\sigma+3-\tau;(\sigma,n)]}(q, \vec x) =
      q^{\sigma+3-\tau} \Pd 4 (q, \vec x) \Big(
      \chi^{\so(4)}_{(\sigma,n)}(\vec x) - q^\tau
      \chi^{\so(4)}_{(\sigma-\tau, n)}(\vec x) \Big)\, ,
    \end{equation}
  \item Massive fields of minimal energy $\Delta$ and spin given a
    two-row Young diagram $\blb k, l \brb$ which read:
    \begin{equation}
      \chi^{\so(2,4)}_{[\Delta;(k,l)]}(q, \vec x) = q^\Delta \Pd 4 (q,
      \vec x) \chi^{\so(4)}_{(k,l)}(\vec x)\, .
    \end{equation}
  \end{itemize}
  We will not display here the full computations for the sake of
  conciseness.
\end{proof}

\section{Conclusion}
\label{sec:discu}
In this note, we considered a class of non-unitary $\so(2,d)$ modules
(for $d=2r$) parametrised by an integer $t$, as possible extensions of
the higher-spin singletons. These $\so(2,d)$ modules describe
partially massless fields of spin $\blb s^r \brb$ and depth-$t$, and
restrict (for $d=2r$) to a sum of partially massless $\so(2,d-1)$
modules of spin $\blb s^{r-1} \brb$ and depth $\tau=1, \dots, t$,
thereby naturally generalising the case of unitary singletons
corresponding to $t=1$. Due to the fact that the branching rule
\eqref{branching_ho} shows that these modules cannot be considered as
AdS$_d$ field preserved by {\it conformal} symmetries, and that the
branching rule \eqref{branching_ho_iso} onto $\iso(1,d-1)$ (deduced
from \eqref{branching_ho} after a In\"on\"u-Wigner contraction) seems
to indicate that those fields are not ``confined'' to the boundary of
AdS, the family of $\so(2,d)$ module $\D\big(s+\tfrac d2-t; \blb s^r
\brb\big)$ does not appear to share the defining properties of
singletons for $t>1$. \\

The decomposition of their tensor product in the low-dimensional
$\so(2,4)$ case contains partially massless fields of the same type
than in the unitary ($t=1$) case, i.e. fields of spin $\blb \sigma
\brb$ with $\sigma\geqslant 2s$ and spin $\blb \sigma, 2s \brb$ with
$\sigma\geqslant 2s+1$, as could be expected from comparison with what
happens for the $\rac_\ell$ and $\di_\ell$ singletons. However, for
$t>1$, partially massless fields with a different spin also appear,
namely of the type $\blb \sigma, n \brb$ with $n$ either taking the
values $1,2, \dots, t-1$ or $2(s-t+1), \dots, 2s$. It is also worth
noticing that only for $s=1$ the decomposition in
\hyperref[th:flato_fronsdal]{Theorem \ref{th:flato_fronsdal}} contains
a conserved spin-$2$ current (i.e. the module $\D\big( 4\,;\, \blb 2
\brb\big)$). It would be interesting to extend this tensor product
decomposition to arbitrary dimensions.

\section*{Acknowledgments}

I am grateful to Xavier Bekaert and Nicolas Boulanger for suggesting
this work in the first place, as well as for various discussions on
the properties of higher-spin singletons and their comments on a
previous version of this paper. I am also grateful for the insightful
comments of an anonymous referee. This work was supported by a joint
grant ``50/50'' Universit\'e Fran\c{c}ois Rabelais Tours -- R\'egion
Centre / UMONS.

\appendix
\section{Branching rules and tensor products of $\so(d)$}
\label{app:sod}
In this appendix we recall the branching and tensor product rules for
$\so(d)$ irreps, as well as detail the proofs of the branching rules
\eqref{branching_uni} and \eqref{branching_ho} and the decomposition
\eqref{decompo_ho_hs_sing}.
\subsection{Branching rules for $\so(d)$}
\label{app:branching}
For $d=2r+1$, the $\so(d)$ irrep $\blb s_1, \dots, s_r \brb$ branches
onto $\so(d-1)$ as:
\begin{equation}
  \blb s_1, \dots, s_r \brb \quad \branching \quad
  \bigoplus_{t_1=s_2}^{s_1} \dots \bigoplus_{t_{r-1}=s_r}^{s_{r-1}}
  \bigoplus_{t_r=-s_r}^{s_r} \blb t_1, \dots, t_r \brb\, ,
\end{equation}
whereas for $d=2r$, the branching rule reads:
\begin{equation}
  \blb s_1, \dots, s_r \brb \quad \branching \quad
  \bigoplus_{t_1=s_2}^{s_1} \dots \bigoplus_{t_{r-1}=s_r}^{s_{r-1}}
  \blb t_1, \dots, t_{r-1} \brb\, .
\end{equation}

\subsection{Computing $\so(d)$ tensor products}
\label{app:tensor_sod}
\begin{minipage}[l]{0.7\textwidth}
  In order to prove the decomposition \eqref{decompo_hs_sing} and
  \eqref{decompo_ho_hs_sing}, as well as the similar decomposition for
  (partially) massless fields with spin given by a rectangular Young
  diagram of arbitrary length $s$, we first need to know know how to
  decompose the tensor product of two $\so(d)$ Young diagrams, one of
  which being a single row of arbitrary length and the other one being
  an ``almost'' rectangular diagram, i.e. of the form:
\end{minipage}
\hspace{-30pt}
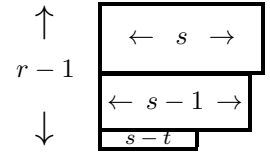
\begin{minipage}[r]{0.25\textwidth}
  \vspace{-10pt}
  \begin{tikzpicture}
    \put(65,-15){\Large $\uparrow$}
    \put(58,-32){$r-1$}
    \put(65,-55){\Large $\downarrow$}
    
    \put(90,-30){\framebox(60,25){$\leftarrow \,\,\, s \,\,\, \rightarrow$}}
    \put(90,-51.5){\framebox(55,20){$\leftarrow \, s-1 \, \rightarrow$}}
    \put(90,-58.5){\framebox(35,6){\footnotesize $s-t$}}
  \end{tikzpicture}
  \label{fig:a}
  \vspace{60pt}
\end{minipage}

\vspace{20pt} 

To do so, it is convenient to express the tensor product rule for
$\so(d)$ in terms of that of $\gl(d)$ which is considerably
simpler. The rule for decomposing a tensor product of two $\gl(d)$
irreps labelled by two Young diagrams $\blambda := \blb \lambda_1,
\dots, \lambda_d \brb$ and $\bmu = \blb \mu_1, \dots, \mu_d \brb$,
known as the Littlewood-Richardson rule, goes as follows (see
e.g. \cite{Bekaert:2006py}):
\begin{itemize}
\item First, assign to the boxes of each rows of one of the Young
  diagrams (say $\bmu$) a label which keeps track of the order of the
  rows (for instance, if the labels are letters of the alphabet, then
  each of the $\mu_1$ boxes of the first row of $\bmu$ are assigned
  the label ``a'', each of the $\mu_2$ boxes of the second row of
  $\bmu$ are assigned the label ``b'', etc\dots);
\item Then, glue the boxes of $\bmu$ to $\blambda$ in all possible
  ways such that the resulting diagram obey the following constraints:
  \begin{itemize}
  \item Boxes in the same column should {\it not} have the same label;
  \item When reading the row of the obtained Young diagram from right
    to left, and its columns from top to bottom, the number of boxes
    encountered should be decreasing with their label (i.e. less boxes
    of the second label are encountered than with the first label,
    less with the third than the second, etc\dots);
  \item The resulting diagram should always be a legitimate Young
    diagram, i.e. the length of the rows is decreasing from top to
    bottom, and it is composed of at most $d$ rows.
  \end{itemize}
\end{itemize}

For orthogonal algebras $\so(2r+1)$, the tensor product of two irreps,
$\Bell = \blb \ell_1, \dots, \ell_r \brb$ and $\blb \sigma \brb$ can
be computed as follows:
\begin{itemize}
\item[(i)] Branch each of the two Young diagrams into $\so(2r)$ and
  pair them by number of boxes removed from the original ones until
  the products of one of these diagrams are exhausted;
\item[(ii)] Compute the tensor product between these pairs of diagrams
  using the Littlewood-Richardson rule recalled above;
\item[(iii)] Discard the Young diagrams which are not acceptable for
  $\so(d)$, i.e. those for which the sum of the height of their first
  two columns is {\it stricly} greater than $d$.
\end{itemize}
The tensor product $\Bell \otimes \blb \sigma \brb$ can therefore be
represented as follows:
\begin{equation}
  \Bell \otimes \blb \sigma \brb = \bigoplus_{m=0}^{\min(\sigma,
    \ell_1-\ell_r)} \bigoplus_{n_1=0}^{\ell_1-\ell_2}
  \bigoplus_{n_2=0}^{\ell_2-\ell_3} \dots \bigoplus_{n_r=0}^{\ell_r}
  \blb \ell_1-n_1, \dots, \ell_r-n_r \brb \ogl \blb \sigma-m \brb
  \Bigg|_{\so(d) \ \mathrm{OK}}\, .
\end{equation}

\begin{example}
  Consider the tensor product between the $\so(5)$ irreps
  $\Yboxdim{5pt} \gyoung(;;,;;)$ and $\Yboxdim{5pt} \gyoung(;;)$. The
  branching rule for these representations is:
  \begin{equation}
    \gyoung(;;,;;) \quad \branchingmod{5}{4} \quad \gyoung(;;,;;)
    \quad \oplus \quad \gyoung(;;,;) \quad \oplus \quad \gyoung(;;)\,
    , \quad \text{and} \quad \gyoung(;;) \quad \branchingmod{5}{4}
    \quad \gyoung(;;) \quad \oplus \quad \gyoung(;) \quad \oplus \quad
    \bullet\, .
  \end{equation}
  Now computing the tensor products between those product paired by
  number of boxes removed, and using the Littlewood-Richardson rule
  yields:
  \begin{eqnarray}
    \gyoung(;;,;;) \otimes_{\gl} \gyoung(;;) & = & \gyoung(;;;;,;;)
    \quad \oplus \quad \gyoung(;;;,;;,;) \quad \oplus \quad
    \gyoung(;;,;;,;;) \\ \gyoung(;;,;) \otimes_{\gl}
    \gyoung(;) \hspace{8pt} & = & \gyoung(;;;,;) \quad \oplus \quad
    \gyoung(;;,;;) \quad \oplus \quad \gyoung(;;,;,;) \\ \gyoung(;;)
    \otimes_{\gl} \bullet \hspace{12pt} & = & \gyoung(;;)\, .
  \end{eqnarray}
  Among the diagrams obtained above, one is not a legitimate $\so(5)$
  Young diagram (as the sum of the height of its two first columns is
  greater than $5$), namely $\Yboxdim{5pt} \gyoung(;;,;;,;;)$. Using
  the fact that two $\so(d)$ Young diagrams whose first column is of
  height $c$ and $d-c$ are equivalent, the initial tensor product
  finally reads:
  \begin{equation}
    \gyoung(;;,;;) \otimes \gyoung(;;) \quad = \quad \gyoung(;;;;,;;)
    \quad \oplus \quad \gyoung(;;;,;;) \quad \oplus \quad
    \gyoung(;;,;;) \quad \oplus \quad \gyoung(;;;,;) \quad \oplus
    \quad \gyoung(;;,;) \quad \oplus \quad \gyoung(;;)\, .
  \end{equation}
\end{example}

Notice that the well-known tensor product rule for $\so(3)$ can be
recovered from the above algorithm. Given two $\so(3)$ irreps $\blb s
\brb$ and $\blb s' \brb$, i.e. two one-row Young diagram of respective
length $s$ and $s'$, their tensor product decomposes as:
\begin{equation}
  \blb s \brb \otimes \blb s' \brb = \bigoplus_{m=0}^{\min(s,s')} \blb
  s-m \brb \otimes_{\gl(3)} \blb s'-m \brb \Big|_{\so(3) \ \text{OK}}
  = \bigoplus_{m=0}^{\min(s,s')} \bigoplus_{k=0}^{\min(s'-m,s-m)} \blb
  s+s'-2m-k, k \brb \Big|_{\so(3) \ \text{OK}} \, .
\end{equation}
The only Young diagrams that are $\so(3)$ acceptable in the previous
equation are those for which $k=0$ or $1$, i.e. the second row
contains no more than one box. In the latter case, such a Young
diagram is equivalent to the one where the second row is absent,
i.e. $\blb s, 1 \brb \cong \blb s \brb$. As a consequence, the
decomposition reads:
\begin{equation}
  \blb s \brb \otimes \blb s' \brb = \bigoplus_{m=0}^{2\min(s,s')}
  \blb s+s'-m \brb = \bigoplus_{k=\rvert s-s' \rvert}^{s+s'} \blb k
  \brb\, ,
\end{equation}
which is indeed the tensor product rule for $\so(3)$.

\section{Technical proofs}
\label{app:technical}
\subsection{Proof of the branching rule for unitary HS singletons}
\label{app:unitary_hs_singleton}
From now on, we will set $d=2r$. In order to prove the branching rule:
\begin{equation}
  \D\big( s+r-1\,;\, \blb s^r \brb\big) \quad \branch \quad \D\big(
  s+r-1\,;\, \blb s^{r-1} \brb\big)\, ,
\end{equation}
we will need to derive the $\so(2)\oplus\so(d-1)$ decomposition of the
$\so(2,d)$ singleton module $\D\big( s+r-1\,;\, \blb s^r \brb\big)$
and of the $\so(2,d-1)$ spin $\blb s^{r-1} \brb$ massless field module
$\D\big( s+r-1\,;\, \blb s^{r-1} \brb\big)$.

\paragraph{Decomposition of massless modules.}
To obtain the $\so(2) \oplus \so(d-1)$ decomposition of the
$\so(2,d-1)$ spin $\blb s^{r-1} \brb$ massless field module $\D\big(
s+r-1\,;\, \blb s^{r-1} \brb\big)$, we will use its character and
rewrite it as a sum of $\so(2)\oplus\so(d-1)$ characters. Using the
property \eqref{p_function} of the function $\Pd{d-1}(q, \vec x)$, the
character of this module becomes:
\begin{eqnarray}
  \chi^{\so(2,d-1)}_{[s+r-1; (s^{r-1})]}(q, \vec x) & = & q^{s+r-1}\,
  \Pd{d-1}(q, \vec x) \Big( \chi^{\so(d-1)}_{(s^{r-1})}(\vec x) +
  \sum_{k=1}^{r-1} (-1)^{k} q^{k} \chi^{\so(d-1)}_{(s^{r-1-k},
    (s-1)^k)}(\vec x) \Big) \label{character_massless_s} \\ & = &
  \sum_{\sigma,n=0}^\infty q^{s+r-1+\sigma+2n}\,
  \chi^{\so(d-1)}_{(\sigma)}(\vec x) \sum_{k=0}^{r-1} (-q)^{k}
  \chi^{\so(d-1)}_{(s^{r-1-k}, (s-1)^k)}(\vec x) \, . \nonumber
\end{eqnarray}
Now we can use the tensor product rule recalled previously for
$\so(d-1)$ with $d-1=2r-1$ odd to reduce the above expression. It
turns out that most of the terms in its alternating sum cancel one
another. To see that, let us have a look at three consecutive terms in
the above sum, that we will denote by ``RHS''. In order to make the
expression more readable, we will also write $\Bell = \blb \ell_1,
\dots, \ell_r \brb$ for the character of the $\so(d-1)$ representation
$\Bell$. A typical triplet of terms in the alternating sum composing
the character \eqref{character_massless_s} reads:
\begin{eqnarray}
  \mathrm{RHS} & = & \sum_{\sigma=0}^\infty \Big( q^{\sigma+k} \blb
  s^{r-k-1}, (s-1)^k \brb - q^{\sigma+k+1} \blb s^{r-k-2}, (s-1)^{k+1}
  \brb + q^{\sigma+k+2} \blb s^{r-k-3}, (s-1)^{k+2} \brb \Big) \otimes
  \blb \sigma \brb \nonumber \\ & = & \sum_{\sigma=0}^\infty
  q^{\sigma+k} \sum_{m=0}^{\min(\sigma,s-1)} \blb s^{r-k-1},
  (s-1)^{k-1}, s-1-m \brb \otimes_{\gl(d)} \blb \sigma-m
  \brb \label{1} \\ && \qquad \qquad \qquad + \sum_{\sigma=1}^\infty
  q^{\sigma+k} \sum_{m=0}^{\min(\sigma-1,s-1)} \blb s^{r-k-2},
  (s-1)^{k}, s-1-m \brb \otimes_{\gl(d)} \blb \sigma-1-m
  \brb \label{2} \\ && -\sum_{\sigma=0}^\infty q^{\sigma+k+1}
  \sum_{m=0}^{\min(\sigma,s-1)} \blb s^{r-k-2}, (s-1)^{k}, s-1-m \brb
  \otimes_{\gl(d)} \blb \sigma-m \brb \label{3} \\ && \qquad \qquad
  \qquad - \sum_{\sigma=1}^\infty q^{\sigma+k+1}
  \sum_{m=0}^{\min(\sigma-1,s-1)} \blb s^{r-k-3}, (s-1)^{k+1}, s-1-m
  \brb \otimes_{\gl(d)} \blb \sigma-1-m \brb \label{4} \\ && +
  \sum_{\sigma=0}^\infty q^{\sigma+k+2} \sum_{m=0}^{\min(\sigma,s-1)}
  \blb s^{r-k-3}, (s-1)^{k+1}, s-1-m \brb \otimes_{\gl(d)} \blb
  \sigma-m \brb \label{5} \\ && \qquad \qquad \qquad +
  \sum_{\sigma=1}^\infty q^{\sigma+k+2}
  \sum_{m=0}^{\min(\sigma-1,s-1)} \blb s^{r-k-4}, (s-1)^{k+2}, s-1-m
  \brb \otimes_{\gl(d)} \blb \sigma-1-m \brb\, . \label{6}
\end{eqnarray}
The second term (corresponding to the fourth and fifth line,
i.e. \eqref{3} and \eqref{4} above) can be rewritten as:
\begin{eqnarray}
  \text{Second term} & = & -\sum_{\sigma=1}^\infty q^{\sigma+k}
  \sum_{m=1}^{\min(\sigma,s)} \blb s^{r-k-2}, (s-1)^{k}, s-m \brb
  \otimes_{\gl(d)} \blb \sigma-m \brb \\ && \qquad -
  \sum_{\sigma=1}^\infty q^{\sigma+k+1} \sum_{m=1}^{\min(\sigma,s)}
  \blb s^{r-k-3}, (s-1)^{k+1}, s-m \brb \otimes_{\gl(d)} \blb \sigma-m
  \brb\, .
\end{eqnarray}
Both of these terms are compensated, the first line above by
rewritting \eqref{2} as:
\begin{equation}
  \eqref{2} = + \sum_{\sigma=1}^\infty q^{\sigma+k}
  \sum_{m=1}^{\min(\sigma,s)} \blb s^{r-k-2}, (s-1)^{k}, s-m \brb
  \otimes_{\gl(d)} \blb \sigma-m \brb\, ,
\end{equation}
and the second line by rewritting \eqref{5} as:
\begin{equation}
  \eqref{5} = + \sum_{\sigma=1}^\infty q^{\sigma+k+1}
  \sum_{m=1}^{\min(\sigma,s)} \blb s^{r-k-3}, (s-1)^{k+1}, s-m \brb
  \otimes_{\gl(d)} \blb \sigma-m \brb\, .
\end{equation}
Hence it appears that the second term in this triplet is completely
compensated, in part by the first terms and in part by the third
one. Because this succession of three terms repeats itself in the
alternating sum of the character \eqref{character_massless_s}, and
that the last term is completely cancelled by the previous one, only
the part of the first term that is not compensated by the second one
remains:
\begin{eqnarray}
  \text{First two terms} & = & \sum_{\sigma=0}^\infty q^{\sigma} \blb
  s^{r-1} \brb \otimes \blb \sigma \brb - \sum_{\sigma=0}^\infty
  q^{\sigma+1} \blb s^{r-2}, s-1 \brb \otimes \blb \sigma \brb \\ & =
  & \sum_{\sigma=0}^\infty q^{\sigma} \sum_{m=0}^{\min(\sigma,s)} \blb
  s^{r-2},s-m \brb \otimes_{\gl(d)} \blb \sigma-m \brb \\ && \qquad -
  \sum_{\sigma=0}^\infty q^{\sigma+1} \sum_{m=0}^{\min(\sigma,s-1)}
  \blb s^{r-2}, s-1-m \brb \otimes_{\gl(d)} \blb \sigma-m \brb
  \nonumber \\ && \qquad \qquad - \quad \text{term that will be
    compensated} \\ & = & \sum_{\sigma=0}^\infty q^{\sigma}
  \sum_{m=0}^{\min(\sigma,s)} \blb s^{r-2},s-m \brb \ogl \blb \sigma-m
  \brb \\ && \qquad - \sum_{\sigma=1}^\infty q^{\sigma}
  \sum_{m=1}^{\min(\sigma,s)} \blb s^{r-2}, s-m \brb \ogl \blb
  \sigma-m \brb \qquad \quad \\ & = & \sum_{\sigma=0}^\infty
  q^{\sigma} \blb s^{r-1} \brb \ogl \blb \sigma \brb =
  \sum_{\sigma=0}^\infty q^\sigma \blb s+\sigma, s^{r-2} \brb \oplus
  \sum_{\sigma=1}^\infty q^\sigma \blb s+\sigma-1, s^{r-2} \brb
  \label{tpgl} \\ & = & \sum_{\sigma=0}^\infty (1+q)\, q^\sigma \blb
    s+\sigma, s^{r-2} \brb\, .
\end{eqnarray}
Notice that only two terms in the $\gl(d-1)$ tensor product survived in
\eqref{tpgl}. Indeed, in full generality, it would produce:
\begin{equation}
  \blb s^{r-1} \brb \ogl \blb \sigma \brb =
  \bigoplus_{k=0}^{\min(s,\sigma)} \blb s+\sigma-k, s^{r-2}, k \brb\,
  ,
\end{equation}
however, only the first two terms are $\so(d-1)$ acceptable Young
diagrams, as they are the only ones for which the sum of the heights
of their first two columns is lower than $d-1$. The character of the
$\so(2,d-1)$ module $\D\big( s+r-1\,;\, \blb s^{r-1} \brb \big)$
therefore reads:
\begin{equation}
  \chi^{\so(2,d-1)}_{[s+r-1;(s^{r-1})]}(q, \vec x) = \sum_{n=0}^\infty
  \sum_{\sigma=0}^\infty q^{s+r-1+2n+\sigma}\, (1+q)\,
  \chi^{\so(d-1)}_{(s+\sigma,s^{r-2})}(\vec x) = \sum_{n=0}^\infty
  \sum_{\sigma=0}^\infty q^{s+r-1+n+\sigma}\,
  \chi^{\so(d-1)}_{(s+\sigma,s^{r-2})}(\vec x)\, ,
\end{equation}
which proves that the $\so(2,d-1)$ module of a spin $\blb s^{r-1}
\brb$ massless field reads:
\begin{equation}
  \D\big( s+r-1; \blb s^{r-1} \brb \big) \quad \cong \quad
  \bigoplus_{\sigma,n=0}^\infty \D_{\so(2) \oplus \so(d-1)}\big(
  s+r-1+\sigma+n; \blb s+\sigma, s^{r-2} \brb \big)\, .
\end{equation}

\paragraph{Decomposition of the singleton module.}
The spin-$s$ singleton module can be decomposed as an infinite direct
sum of (finite-dimensional) $\so(2) \oplus \so(d)$ modules as
\cite{Dolan:2005wy}:
\begin{equation}
  \D\big(s+r-1;(s^r)\big) \cong \bigoplus_{\sigma=0}^\infty \D_{\so(2)
    \oplus \so(d)} \big( s+r-1+\sigma; (s+\sigma,s^{r-1}) \big)\, .
\end{equation}
In order to branch this module from $\so(2,d)$ to $\so(2,d-1)$, one
can simply branch the $\so(d)$ part of the above decomposition onto
$\so(d-1)$, thereby yielding:
\begin{eqnarray}
  \D\big(s+r-1;(s^r)\big) & \branch & \bigoplus_{\sigma=0}^\infty
  \bigoplus_{\tau=0}^\sigma \D_{\so(2) \oplus \so(d-1)} \big(
  s+r-1+\sigma; (s+\tau,s^{r-2}) \big) \\ & \cong &
  \bigoplus_{\sigma=0}^\infty \bigoplus_{n=0}^\infty \D_{\so(2) \oplus
    \so(d-1)} \big( s+r-1+\sigma+n; (s+\sigma,s^{r-2}) \big) \, .
\end{eqnarray}
This $\so(2)\oplus\so(d-1)$ decomposition matches that of the
$\so(2,d-1)$ module of a spin $\blb s^{r-1} \brb$ massless field,
hence we proved:
\begin{equation}
  \D\big( s+r-1\,;\, \blb s^r \brb \big) \quad \branch \quad \D\big(
  s+r-1\,;\, \blb s^{r-1} \brb \big)\, .
\end{equation}

\subsection{Proof of the branching rule for rectangular partially massless fields}
\label{app:nonunitary_hs_singleton}
Following the same steps as in the previous section, we will now
proceed to proving the branching rule:
\begin{equation}
  \D\big( s+r-t \,;\, \blb s^r \brb \big) \quad \branch \quad
  \bigoplus_{\tau=1}^t \D\big( s+r-\tau \,;\, \blb s^{r-1} \brb
  \big)\, .
\end{equation}

\paragraph{Decomposition of the partially massless modules for $d=2r+1$.}
Let us start with the $\so(2)\oplus\so(d-1)$ decomposition of the
$\so(2,d-1)$ module $\D\big( s+r-\tau; \blb s^{r-1} \brb\big)$
corresponding to a partially massless field of spin $\blb s^{r-1}
\brb$ and depth-$\tau$. Using the same notational shortcuts than in
the previous section, its character can be rewritten as:
\begin{eqnarray}
  \chi^{\so(2,d-1)}_{[s+r-\tau;(s^{r-1})]}(q, \vec x) & = &
  \sum_{\sigma,n=0}^\infty q^{s+r-\tau+2n+\sigma} \blb \sigma \brb
  \otimes \Big( \blb s^{r-1} \brb + \sum_{k=0}^{r-2} (-1)^{k+1}
  q^{\tau+k} \blb s^{r-2-k}, (s-1)^k, s-\tau \brb \Big) \\ & = &
  \sum_{n=0}^\infty q^{s+r-\tau+2n} \Big( \sum_{\sigma=0}^\infty
  q^\sigma \sum_{m=0}^{\min(s,\sigma)} \blb \sigma-m \brb \ogl \blb
  s^{r-2}, s-m \brb \\ && \qquad - \sum_{p=0}^\tau
  \sum_{\sigma=p}^\infty q^{\sigma+\tau}
  \sum_{m=0}^{\min(s-\tau,\sigma-p)} \blb \sigma-p-m \brb \ogl \blb
  s^{r-3}, s-p,s-\tau-m \brb \nonumber \\ && + \sum_{k=1}^{r-2}
  (-1)^{k+1} \sum_{p=0}^{\tau-1} \sum_{\sigma=p}^\infty
  q^{\sigma+\tau+k} \times \\ && \qquad \qquad
  \sum_{m=0}^{\min(s-\tau,\sigma-p)} \blb \sigma-p-m \brb \ogl \blb
  s^{r-2-k}, (s-1)^{k-1}, s-1-p, s-\tau-m \brb \nonumber \\ && +
  \sum_{k=1}^{r-3} (-1)^{k+1} \sum_{p=0}^{\tau-1}
  \sum_{\sigma=p+1}^\infty q^{\sigma+\tau+k} \times \\ && \qquad
  \qquad \sum_{m=0}^{\min(s-\tau,\sigma-p-1)} \blb \sigma-p-1-m \brb
  \ogl \blb s^{r-3-k}, (s-1)^{k}, s-1-p, s-\tau-m \brb \Big) \nonumber
  \\ & = & \sum_{n=0}^\infty q^{s+r-\tau+2n} \Big(
  \sum_{\sigma=0}^\infty q^\sigma \sum_{m=0}^{\min(s,\sigma)} \blb
  \sigma-m \brb \ogl \blb s^{r-2}, s-m \brb \\ && \qquad -
  \sum_{p=0}^\tau \sum_{\sigma=p}^\infty q^{\sigma+\tau}
  \sum_{m=0}^{\min(s-\tau,\sigma-p)} \blb \sigma-p-m \brb \ogl \blb
  s^{r-3}, s-p,s-\tau-m \brb \\ && + \sum_{p=0}^{\tau-1}
  \sum_{\sigma=p}^\infty q^{\sigma+\tau+1}
  \sum_{m=0}^{\min(s-\tau,\sigma-p)} \blb \sigma-p-m \brb \ogl \blb
  s^{r-3}, s-1-p, s-\tau-m \brb \Big) \\ & = & \sum_{n=0}^\infty
  q^{s+r-\tau+2n} \Big( \sum_{\sigma=0}^\infty q^\sigma
  \sum_{m=0}^{\min(s,\sigma)} \blb \sigma-m \brb \ogl \blb s^{r-2},
  s-m \brb \\ && \qquad \qquad \qquad - \sum_{\sigma=0}^\infty
  q^{\sigma+\tau} \sum_{m=0}^{\min(s-\tau,\sigma)} \blb \sigma-m \brb
  \ogl \blb s^{r-2}, s-\tau-m \brb \Big) \\ & = &
  \sum_{\sigma,n=0}^\infty q^{s+r-\tau+\sigma+2n}
  \sum_{m=0}^{\min(\sigma,\tau-1)} \blb \sigma-m \brb \ogl \blb
  s^{r-2}, s-m \brb\, .
\end{eqnarray}
Leaving aside the sum of $n$, the above equation can be re-expressed
as:
\begin{equation}
  \sum_{\sigma=0}^\infty q^\sigma \sum_{m=0}^{\min(\sigma,\tau-1)}
  \blb \sigma-m \brb \ogl \blb s^{r-2}, s-m \brb = \sum_{k=0}^{\tau-1}
  \sum_{\sigma=k}^\infty q^\sigma \blb \sigma-k \brb \ogl \blb
  s^{r-2}, s-k \brb\, .
  \label{intermediary}
\end{equation}
Using the Littlewood-Richardson rule yields:
\begin{eqnarray}
  \sum_{\sigma=k}^\infty q^\sigma \blb \sigma-k \brb \ogl \blb
  s^{r-2}, s-k \brb & = & \sum_{p=0}^k \sum_{\sigma=k+p}^\infty
  q^\sigma\, (1+q)\, \blb s+\sigma-k-p, s^{r-3} \brb \\ & = &
  \sum_{p=0}^k \sum_{\sigma=2k-p}^\infty q^\sigma\, (1+q)\, \blb
  s+\sigma-2k+p, s^{r-3}, s-p \brb \\ & = & \sum_{p=0}^k
  \sum_{\sigma=p}^\infty q^{\sigma+2(k-p)}\, (1+q)\, \blb s+\sigma-p,
  s^{r-3}, s-p \brb\, ,
\end{eqnarray}
hence \eqref{intermediary} becomes:
\begin{equation}
  \sum_{k=0}^{\tau-1} \sum_{p=0}^k \sum_{\sigma=p}^\infty
  q^{\sigma+2(k-p)}\, (1+q)\, \blb s+\sigma-p, s^{r-3}, s-p \brb =
  \sum_{k=0}^{\tau-1} \sum_{p=0}^{\tau-1-k} \sum_{\sigma=p}^\infty
  q^{\sigma+2p}\, (1+q)\, \blb s+\sigma-k, s^{r-3}, s-k \brb\, .
\end{equation}
Finally, the character of the $\so(2,d-1)$ module $\D\big( s+r-\tau;
\blb s^{r-1} \brb \big)$ can be expressed as:
\begin{equation}
  \chi^{\so(2,d-1)}_{[s+r-\tau;(s^{r-1})]}(q, \vec x) =
  \sum_{n=0}^\infty \sum_{k=0}^{\tau-1} \sum_{p=0}^{\tau-1-k}
  \sum_{\sigma=k}^\infty q^{s+r-\tau+2p+\sigma+n}
  \chi^{\so(d-1)}_{(s+\sigma-k, s^{r-3}, s-k)}(\vec x)\, ,
\end{equation}
which proves that this module decomposes as:
\begin{equation}
  \D\big( s+r-\tau; \blb s^{r-1} \brb \big) \cong
  \bigoplus_{k=0}^{\tau-1} \bigoplus_{p=0}^{\tau-1-k}
  \bigoplus_{n=0}^\infty \bigoplus_{\sigma=k}^\infty \D_{\so(2) \oplus
    \so(d-1)}\big( s+r-\tau+\sigma+n+2p; \blb s+\sigma-k, s^{r-3}, s-k
  \brb \big) \, .
\end{equation}

\paragraph{Decomposition of the rectangular partially massless field module for $d=2r$.}
Before deriving the $\so(2)\oplus\so(d-1)$ decomposition of the
spin-$s$ and depth-$t$ RPM module, we will need to prove its
$\so(2)\oplus\so(d)$ decomposition:
\begin{equation}
  \D\big( s+r-t; \blb s^r \brb \big) \cong \bigoplus_{\ell=0}^{t-1}
  \bigoplus_{n=0}^{t-1-\ell} \bigoplus_{\sigma=\ell}^\infty
  \D_{\so(2)\oplus\so(d)}\big( s+r-t+\sigma+2n; \blb s-\ell+\sigma,
  s^{r-2}, s-\ell \brb \big)\, .
  \label{decompotoprove}
\end{equation}
To do so, we will need to use the Weyl character formula:
\begin{equation}
  \chi^{\so(d)}_{\Bell}(\vec x) = \frac{\sum_{w \in \W_{\so(d)}}
    \varepsilon(w) e^{(w(\Bell+\brho)-\brho,
      \bmu)}}{\prod_{\alpha\in\Phi_+} (1-e^{(\balpha,\bmu)})}\, ,
\end{equation}
where $\Phi_+$ is the set of positive roots of $\so(d)$, $\W_{\so(d)}$
its Weyl group, $\brho := \tfrac12 \sum_{\balpha \in \Phi_+} \balpha$
the Weyl vector, $\varepsilon(w)$ the signature of the Weyl group
element $w$, $(\,,)$ is inner product on the root space inherited from
the Killing form and $\bmu = \blb \mu_1, \dots, \mu_r \brb$ an
arbitrary root used to define the variables $x_i$ on which depends the
character via $x_i := e^{\mu_i}$. Defining
\begin{equation}
  \C_{\Bell}^{\so(d)}(\vec x) := e^{(\Bell,\bmu)}
  \prod_{\balpha\in\Phi_+} \frac{1}{1-e^{(\balpha, \bmu)}} =
  \prod_{i=1}^r x_i^{\ell_i} \prod_{\balpha\in\Phi_+}
  \frac{1}{1-e^{(\balpha, \bmu)}}\, ,
  \label{weyl_character}
\end{equation}
one can show
\begin{equation}
  w\big(\C_{\Bell}^{\so(d)}(\vec x)\big) := e^{(w(\Bell),\bmu)}
  \prod_{\balpha\in\Phi_+} \frac{1}{1-e^{(w(\balpha), \bmu)}} =
  \varepsilon(w) \C_{w \cdot \Bell}^{\so(d)}(\vec x)\, ,
\end{equation}
with $w\cdot \Bell = w(\Bell+\brho)-\brho$. Hence, the Weyl character
formula \eqref{weyl_character} can be rewritten as:
\begin{equation}
  \chi^{\so(d)}_{\Bell}(\vec x) = \sum_{w\in\W_{\so(d)}}
  w\big(\C_{\Bell}^{\so(d)}(\vec x) \big) := \fW_{\so(d)}\Big(
  \C^{\so(d)}_{\Bell}(\vec x) \Big)\, .
\end{equation}
The Weyl group for $\W_{\so(d)}$ for $d=2r$ is the semi-direct product
$\Sn_r \ltimes (\Z_2)^{r-1}$ where $\Sn_r$ is the permutation group of
$r$ elements. Any element of $\W_{\so(d)}$ acts on the $\so(d)$
characters as a combination of permutation of the variables $x_i$ and
(pair of) sign flip of their exponents (for more details, see
e.g. \cite{Dolan:2005wy} or the classical textbooks \cite{Fulton1991,
  Fuchs2003}). As a consequence, any function invariant under such
operation can go in and out of the Weyl symmetrizer $\fW_{\so(d)}$. In
particular, $\Pd d (q, \vec x)$ being invariant under any permutation
and any inversion ($x_i \rightarrow x_i^{-1}$) of the variables $x_i$,
it has the property:
\begin{equation}
  \fW_{\so(d)} \big( \C^{\so(d)}_{\Bell}(\vec x) \Pd d (q, \vec x)
  \big) = \fW_{\so(d)} \big( \C^{\so(d)}_{\Bell}(\vec x) \big) \Pd d
  (q, \vec x)\, .
  \label{propW}
\end{equation}
Using this property, as well as the fact that
$\C_{\Bell}^{\so(d)}(\vec x)$ verifies:
\begin{equation}
  \C_{(\ell_1, \dots, \ell_{j-1}, \ell_j+a, \ell_{j+1}, \dots,
    \ell_r)}^{\so(d)}(\vec x) = x_j^a\,
  \C_{(\ell_1,\dots,\ell_{j-1},\ell_j,\ell_{j+1},\dots,
    \ell_r)}^{\so(d)}(\vec x)\, ,
  \label{propC}
\end{equation}
we can now prove that the character of the $\so(2)\oplus\so(d)$
decomposition \eqref{decompotoprove} can be rewritten\,\footnote{This
  technique was originally used in \cite{Dolan:2005wy} (see appendix
  D) in order to derive the $\so(2) \oplus \so(d)$ decomposition of
  the unitary singleton modules. We merely adapt it here to the
  non-unitary case.}  as the character of the spin-$s$ and depth-$t$
RPM module \eqref{character_ho_singleton}:
\begin{eqnarray}
  \text{Decomposition} & = & \sum_{\ell=0}^{t-1} \sum_{n=0}^{t-1-\ell}
  \sum_{\sigma=\ell}^\infty q^{s+r-t+2n+\sigma}
  \chi^{\so(d)}_{(s+\sigma-\ell, s^{r-2}, s-\ell)}(\vec x) \\ & = &
  \sum_{\ell=0}^{t-1} q^{s+r-t+\ell} \frac{1-q^{2(t-\ell)}}{1-q^2}
  \fW_{\so(d)}\Big( \sum_{\sigma=0}^\infty q^\sigma x_1^\sigma
  \C^{\so(d)}_{(s^{r-1},s-\ell)}(\vec x) \Big) \\ & = &
  \frac{q^{s+r}}{1-q^2} \fW_{\so(d)}\Big( \sum_{\ell=0}^{t-1} \big[
    q^{-t+\ell} - q^{t-\ell} \big] \frac{x_r^{-\ell}}{1-qx_1}
  \C^{\so(d)}_{(s^r)}(\vec x) \Big) \\ & = & \frac{q^{s+r-t}}{1-q^2}
  \fW_{\so(d)}\Big( \frac{1 - (1-q^2) q^tx_r^{-t} - (1-q^{2t}) qx_r -
    q^{2t+2}}{(1-qx_1)(1-qx_r)(1-qx_r^{-1})} \C^{\so(d)}_{(s^r)}(\vec
  x) \Big) \\ & = & \frac{q^{s+r-t}}{1-q^2} \Pd d (q, \vec x) \times
  \\ && \quad \fW_{\so(d)}\Big( \big[ 1 - q^{2t+2} - (1-q^{2t}) qx_r -
    (1-q^2) q^tx_r^{-t} \big] (1-qx^{-1}) \prod_{i=2}^{r-1}
  (1-qx_i)(1-qx_i^{-1}) \C^{\so(d)}_{(s^r)}(\vec x) \Big) \nonumber
\end{eqnarray}

Each terms of the above sum as some power of $q$ times a $\so(d)$
character: using \eqref{weyl_character}, \eqref{propW} and
\eqref{propC} they are all of the form
\begin{equation}
  \fW_{\so(d)}\Big( q^\beta \prod_{i=1}^r x_i^{\alpha_i}
  \C_{(s^r)}^{\so(d)}(\vec x) \Big) = q^\beta\,
  \chi^{\so(d)}_{(s+\alpha_1, s+\alpha_2, \dots, s+\alpha_r)}(\vec
  x)\, .
\end{equation}
\begin{itemize}
\item The first piece, proportional to $(1-q^{2t+2})$ reads:
  \begin{eqnarray}
    \fW_{\so(d)}\Big( (1-qx^{-1}) & \prod_{i=2}^{r-1} &
    (1-qx_i)(1-qx_i^{-1})\, \C^{\so(d)}_{(s^r)}(\vec x) \Big) \\ & = &
    \sum_{\scriptstyle 0 \leqslant n \leqslant 2r-3 \atop \scriptstyle
      n := n_{1,-} + \sum_{i=2}^{r-1} n_{i,+} + n_{i,-}} (-q)^n
    \chi^{\so(d)}_{(s-n_{1,-},s+n_{2,+}-n_{2,-}, \dots,
      s+n_{r-1,+}-n_{r-1,-},s)}(\vec x)\, . \nonumber
  \end{eqnarray}
  Using the symmetry property
  \begin{equation}
    \chi^{\so(d)}_{(\ell_1, \dots, \ell_j, \ell-1, \ell+1, \ell_{j+3},
      \dots, \ell_r)}(\vec x) = - \chi^{\so(d)}_{(\ell_1, \dots,
      \ell_j, \ell, \ell, \ell_{j+3}, \dots, \ell_r)}(\vec x)\, ,
    \label{sym_prop}
  \end{equation}
  of the $\so(d)$ characters, one can check that the previous sum
  reduces to a single contribution:
  \begin{equation}
    \fW_{\so(d)}\Big( (1-qx^{-1}) \prod_{i=2}^{r-1}
    (1-qx_i)(1-qx_i^{-1})\, \C^{\so(d)}_{(s^r)}(\vec x) \Big) =
    \chi^{\so(d)}_{(s^r)}(\vec x)\, .
  \end{equation}
  For instance, for $r=3$:
  \begin{eqnarray}
    \fW_{\so(d)}\Big( (1-qx^{-1}) (1-qx_2)(1-qx_2^{-1})\,
    \C^{\so(d)}_{(s^r)}(\vec x) \Big) & = & (1+q^2)
    \chi^{\so(d)}_{(s,s,s)}(\vec x) -q \big(
    \chi^{\so(d)}_{(s,s+1,s)}(\vec x) + \chi^{\so(d)}_{(s,s-1,s)}(\vec
    x) \big) \nonumber \\ && \qquad -q (1+q^2)
    \chi^{\so(d)}_{(s-1,s,s)}(\vec x) \nonumber \\ && \qquad \quad
    +q^2 \big( \chi^{\so(d)}_{(s-1,s+1,s)}(\vec x) +
    \chi^{\so(d)}_{(s-1,s-1,s)}(\vec x) \big)\, . \nonumber
  \end{eqnarray}
  The symmetry property \eqref{sym_prop} implies that:
  \begin{equation}
    \chi^{\so(d)}_{(s,s+1,s)}(\vec x) = \chi^{\so(d)}_{(s,s-1,s)}(\vec
    x) = \chi^{\so(d)}_{(s-1,s-1,s)}(\vec x) =
    \chi^{\so(d)}_{(s-1,s,s)}(\vec x) = 0\, , \quad \text{and} \quad
    \chi^{\so(d)}_{(s-1,s+1,s)}(\vec x) = -
    \chi^{\so(d)}_{(s,s,s)}(\vec x)\, , \nonumber
  \end{equation}
  hence we indeed obtain:
  \begin{equation}
    \fW_{\so(d)}\Big( (1-qx^{-1}) (1-qx_2)(1-qx_2^{-1})\,
    \C^{\so(d)}_{(s^r)}(\vec x) \Big) = \chi^{\so(d)}_{(s,s,s)}(\vec
    x)\, .
  \end{equation}
\item The second piece, proportional to $(1-q^{2t})$, also reduces
  to a single contribution upon using the same symmetry property
  \eqref{sym_prop}:
  \begin{eqnarray}
    \fW_{\so(d)}\Big( qx_r (1-qx^{-1}) & \prod_{i=2}^{r-1} &
    (1-qx_i)(1-qx_i^{-1})\, \C^{\so(d)}_{(s^r)}(\vec x) \Big) \\ & = &
    \sum_{\scriptstyle 0 \leqslant n \leqslant 2r-3 \atop \scriptstyle
      n := n_{1,-} + \sum_{i=2}^{r-1} n_{i,+} + n_{i,-}} (-)^{n}
    q^{n+1} \chi^{\so(d)}_{(s-n_{1,-},s+n_{2,+}-n_{2,-}, \dots,
      s+n_{r-1,+}-n_{r-1,-},s+1)}(\vec x) \nonumber \\ & = & q^2\,
    \chi^{\so(d)}_{(s^r)}(\vec x)
  \end{eqnarray}
\item Finally, the third piece (proportional to $(1-q^2)$) contains
  more contributions:
  \begin{eqnarray}
    \fW_{\so(d)}\Big( q^tx_r^{-t} (1-qx^{-1}) & \prod_{i=2}^{r-1} &
    (1-qx_i)(1-qx_i^{-1})\, \C^{\so(d)}_{(s^r)}(\vec x) \Big) \\ & = &
    \sum_{\scriptstyle 0 \leqslant n \leqslant 2r-3 \atop \scriptstyle
      n := n_{1,-} + \sum_{i=2}^{r-1} n_{i,+} + n_{i,-}} (-1)^{n}
    q^{t+n} \chi^{\so(d)}_{(s-n_{1,-},s+n_{2,+}-n_{2,-}, \dots,
      s+n_{r-1,+}-n_{r-1,-},s-t)}(\vec x) \nonumber \\ & = &
    \sum_{k=0}^{r-1} (-1)^k q^{t+k}
    \chi^{\so(d)}_{(s^{r-1-k},(s-1)^{k},s-t)}(\vec x)\, .
  \end{eqnarray}
\end{itemize}
Putting those three piece together yields:
\begin{eqnarray}
  \sum_{\ell=0}^{t-1} \sum_{n=0}^{t-1-\ell} \sum_{\sigma=\ell}^\infty
  q^{s+r-t+2n+\sigma} \chi^{\so(d)}_{(s+\sigma-\ell, s^{r-2},
    s-\ell)}(\vec x) & = & \frac{q^{s+r-t}}{1-q^2} \Big(
  (1-q^{2t+2})\, \chi^{\so(d)}_{(s^r)}(\vec x) - (1-q^{2t})\, q^2\,
  \chi^{\so(d)}_{(s^r)}(\vec x) \\ & & - (1-q^2) \sum_{k=0}^{r-1}
  (-)^k q^{t+k}\, \chi^{\so(d)}_{(s^{r-1-k},(s-1)^k, s-t)}(\vec x)
  \Big) \Pd d (q, \vec x) \nonumber \\ & = & q^{s+r-t}\, \Pd d (q,
  \vec x)\, \Big( \chi^{\so(d)}_{(s^r)}(\vec x) \\ && \qquad \qquad -
  \sum_{k=0}^{r-1} (-)^k q^{t+k}\,
  \chi^{\so(d)}_{(s^{r-1-k},(s-1)^k,s-t)}(\vec x) \Big) \\ & = &
  \chi^{\so(2,d)}_{[s+r-t;(s^r)]}(q, \vec x)\, ,
\end{eqnarray}
thereby proving the $\so(2)\oplus\so(d)$ decomposition
\eqref{decompotoprove} of the spin-$s$ and depth-$t$ RPM module.\\

Finally, the $\so(2)\oplus\so(d-1)$ of this module reads:
\begin{eqnarray}
  \D\big( s+r-t; \blb s^r \brb \big) & \cong &
  \bigoplus_{\ell=0}^{t-1} \bigoplus_{n=0}^{t-1-\ell}
  \bigoplus_{\sigma=\ell}^\infty \D_{\so(2)\oplus\so(d)}\big(
  s+r-t+\sigma+2n; \blb s-\ell+\sigma, s^{r-2}, s-\ell \brb \big) \\ &
  \branch & \bigoplus_{\ell=0}^{t-1} \bigoplus_{n=0}^{t-1-\ell}
  \bigoplus_{\sigma=\ell}^\infty \bigoplus_{\tau=0}^{\sigma-\ell}
  \bigoplus_{k=0}^\ell \D_{\so(2)\oplus\so(d-1)}\big( s+r-t+\sigma+2n;
  \blb s+\tau, s^{r-3}, s-k \brb \big) \\ & \cong &
  \bigoplus_{\tau=1}^t \bigoplus_{k=0}^{\tau-1}
  \bigoplus_{p=0}^{\tau-1-k} \bigoplus_{n=0}^\infty
  \bigoplus_{\sigma=k}^\infty \D_{\so(2) \oplus \so(d-1)}\big(
  s+r-\tau+\sigma+n+2p; \blb s+\sigma-k, s^{r-3}, s-k \brb \big) \, ,
  \nonumber
\end{eqnarray}
thereby proving the branching rule:
\begin{equation}
  \D\big( s+r-t\,;\, \blb s^r \brb \big) \quad \branch \quad
  \bigoplus_{\tau=1}^t \D\big( s+r-\tau\,;\, \blb s^{r-1} \brb \big)\,
  .
\end{equation}

\newpage
\bibliographystyle{utphys.bst} 
\bibliography{biblio}

\providecommand{\href}[2]{#2}\begingroup\raggedright\begin{thebibliography}{100}

\bibitem{Metsaev:2016lhs}
R.~R. Metsaev, ``{Continuous spin gauge field in (A)dS space},''
  \href{http://dx.doi.org/10.1016/j.physletb.2017.02.027}{{\em Phys. Lett.}
  {\bf B767} (2017)  458--464},
\href{http://arxiv.org/abs/1610.00657}{{\tt arXiv:1610.00657 [hep-th]}}.

\bibitem{Metsaev:2017ytk}
R.~R. Metsaev, ``{Fermionic continuous spin gauge field in (A)dS space},''
\href{http://arxiv.org/abs/1703.05780}{{\tt arXiv:1703.05780 [hep-th]}}.

\bibitem{Bekaert:2017khg}
X.~Bekaert and E.~D. Skvortsov, ``{Elementary particles with continuous
  spin},''
\href{http://arxiv.org/abs/1708.01030}{{\tt arXiv:1708.01030 [hep-th]}}.

\bibitem{Angelopoulos:1980wg}
E.~Angelopoulos, M.~Flato, C.~Fronsdal, and D.~Sternheimer, ``{Massless
  Particles, Conformal Group and De Sitter Universe},''
\href{http://dx.doi.org/10.1103/PhysRevD.23.1278}{{\em Phys. Rev.} {\bf D23}
  (1981)  1278}.

\bibitem{Angelopoulos:1997ij}
E.~Angelopoulos and M.~Laoues, ``{Masslessness in n-dimensions},''
  \href{http://dx.doi.org/10.1142/S0129055X98000082}{{\em Rev. Math. Phys.}
  {\bf 10} (1998)  271--300},
\href{http://arxiv.org/abs/hep-th/9806100}{{\tt arXiv:hep-th/9806100
  [hep-th]}}.

\bibitem{Iazeolla:2008ix}
C.~Iazeolla and P.~Sundell, ``{A Fiber Approach to Harmonic Analysis of
  Unfolded Higher-Spin Field Equations},''
  \href{http://dx.doi.org/10.1088/1126-6708/2008/10/022}{{\em JHEP} {\bf 10}
  (2008)  022},
\href{http://arxiv.org/abs/0806.1942}{{\tt arXiv:0806.1942 [hep-th]}}.

\bibitem{Flato:1978qz}
M.~Flato and C.~Fronsdal, ``{One Massless Particle Equals Two Dirac Singletons:
  Elementary Particles in a Curved Space. 6.},''
\href{http://dx.doi.org/10.1007/BF00400170}{{\em Lett. Math. Phys.} {\bf 2}
  (1978)  421--426}.

\bibitem{Angelopoulos:1999bz}
E.~Angelopoulos and M.~Laoues, ``{Singletons on AdS(n)},'' in {\em {Conference
  Moshe Flato Dijon, France, September 5-8, 1999}}, pp.~3--23.
\newblock
1999.
\newblock

\bibitem{Vasiliev:2004cm}
M.~A. Vasiliev, ``{Higher spin superalgebras in any dimension and their
  representations},''
  \href{http://dx.doi.org/10.1088/1126-6708/2004/12/046}{{\em JHEP} {\bf 12}
  (2004)  046},
\href{http://arxiv.org/abs/hep-th/0404124}{{\tt arXiv:hep-th/0404124
  [hep-th]}}.

\bibitem{Dolan:2005wy}
F.~A. Dolan, ``{Character formulae and partition functions in higher
  dimensional conformal field theory},''
  \href{http://dx.doi.org/10.1063/1.2196241}{{\em J. Math. Phys.} {\bf 47}
  (2006)  062303},
\href{http://arxiv.org/abs/hep-th/0508031}{{\tt arXiv:hep-th/0508031
  [hep-th]}}.

\bibitem{Vasiliev:1988xc}
M.~A. Vasiliev, ``{Equations of Motion of Interacting Massless Fields of All
  Spins as a Free Differential Algebra},''
\href{http://dx.doi.org/10.1016/0370-2693(88)91179-3}{{\em Phys. Lett.} {\bf
  B209} (1988)  491--497}.

\bibitem{Vasiliev:1990en}
M.~A. Vasiliev, ``{Consistent equation for interacting gauge fields of all
  spins in (3+1)-dimensions},''
\href{http://dx.doi.org/10.1016/0370-2693(90)91400-6}{{\em Phys. Lett.} {\bf
  B243} (1990)  378--382}.

\bibitem{Vasiliev:1992av}
M.~A. Vasiliev, ``{More on equations of motion for interacting massless fields
  of all spins in (3+1)-dimensions},''
\href{http://dx.doi.org/10.1016/0370-2693(92)91457-K}{{\em Phys. Lett.} {\bf
  B285} (1992)  225--234}.

\bibitem{Vasiliev:2003ev}
M.~A. Vasiliev, ``{Nonlinear equations for symmetric massless higher spin
  fields in (A)dS(d)},''
  \href{http://dx.doi.org/10.1016/S0370-2693(03)00872-4}{{\em Phys. Lett.} {\bf
  B567} (2003)  139--151},
\href{http://arxiv.org/abs/hep-th/0304049}{{\tt arXiv:hep-th/0304049
  [hep-th]}}.

\bibitem{Bekaert:2010hw}
X.~Bekaert, N.~Boulanger, and P.~Sundell, ``{How higher-spin gravity surpasses
  the spin two barrier: no-go theorems versus yes-go examples},''
  \href{http://dx.doi.org/10.1103/RevModPhys.84.987}{{\em Rev. Mod. Phys.} {\bf
  84} (2012)  987--1009},
\href{http://arxiv.org/abs/1007.0435}{{\tt arXiv:1007.0435 [hep-th]}}.

\bibitem{Bekaert:2005vh}
X.~Bekaert, S.~Cnockaert, C.~Iazeolla, and M.~A. Vasiliev, ``{Nonlinear higher
  spin theories in various dimensions},'' in {\em {Higher spin gauge theories:
  Proceedings, 1st Solvay Workshop: Brussels, Belgium, 12-14 May, 2004}},
  pp.~132--197.
\newblock 2004.
\newblock
\href{http://arxiv.org/abs/hep-th/0503128}{{\tt arXiv:hep-th/0503128
  [hep-th]}}.
\newblock

\bibitem{Didenko:2014dwa}
V.~E. Didenko and E.~D. Skvortsov, ``{Elements of Vasiliev theory},''
\href{http://arxiv.org/abs/1401.2975}{{\tt arXiv:1401.2975 [hep-th]}}.

\bibitem{Maldacena:1997re}
J.~M. Maldacena, ``{The Large N limit of superconformal field theories and
  supergravity},'' \href{http://dx.doi.org/10.1023/A:1026654312961}{{\em Int.
  J. Theor. Phys.} {\bf 38} (1999)  1113--1133},
  \href{http://arxiv.org/abs/hep-th/9711200}{{\tt arXiv:hep-th/9711200
  [hep-th]}}.
[Adv. Theor. Math. Phys.2,231(1998)].

\bibitem{Witten:1998qj}
E.~Witten, ``{Anti-de Sitter space and holography},'' {\em Adv. Theor. Math.
  Phys.} {\bf 2} (1998)  253--291,
\href{http://arxiv.org/abs/hep-th/9802150}{{\tt arXiv:hep-th/9802150
  [hep-th]}}.

\bibitem{Gubser:1998bc}
S.~S. Gubser, I.~R. Klebanov, and A.~M. Polyakov, ``{Gauge theory correlators
  from noncritical string theory},''
  \href{http://dx.doi.org/10.1016/S0370-2693(98)00377-3}{{\em Phys. Lett.} {\bf
  B428} (1998)  105--114},
\href{http://arxiv.org/abs/hep-th/9802109}{{\tt arXiv:hep-th/9802109
  [hep-th]}}.

\bibitem{Sezgin:2002rt}
E.~Sezgin and P.~Sundell, ``{Massless higher spins and holography},''
  \href{http://dx.doi.org/10.1016/S0550-3213(02)00739-3,
  10.1016/S0550-3213(03)00267-0}{{\em Nucl. Phys.} {\bf B644} (2002)
  303--370}, \href{http://arxiv.org/abs/hep-th/0205131}{{\tt
  arXiv:hep-th/0205131 [hep-th]}}.
[Erratum: Nucl. Phys.B660,403(2003)].

\bibitem{Klebanov:2002ja}
I.~R. Klebanov and A.~M. Polyakov, ``{AdS dual of the critical O(N) vector
  model},'' \href{http://dx.doi.org/10.1016/S0370-2693(02)02980-5}{{\em Phys.
  Lett.} {\bf B550} (2002)  213--219},
\href{http://arxiv.org/abs/hep-th/0210114}{{\tt arXiv:hep-th/0210114
  [hep-th]}}.

\bibitem{Giombi:2013fka}
S.~Giombi and I.~R. Klebanov, ``{One Loop Tests of Higher Spin AdS/CFT},''
  \href{http://dx.doi.org/10.1007/JHEP12(2013)068}{{\em JHEP} {\bf 12} (2013)
  068},
\href{http://arxiv.org/abs/1308.2337}{{\tt arXiv:1308.2337 [hep-th]}}.

\bibitem{Giombi:2014yra}
S.~Giombi, I.~R. Klebanov, and A.~A. Tseytlin, ``{Partition Functions and
  Casimir Energies in Higher Spin AdS$_{d+1}$/CFT$_d$},''
  \href{http://dx.doi.org/10.1103/PhysRevD.90.024048}{{\em Phys. Rev.} {\bf
  D90} (2014) no.~2, 024048},
\href{http://arxiv.org/abs/1402.5396}{{\tt arXiv:1402.5396 [hep-th]}}.

\bibitem{Sezgin:2003pt}
E.~Sezgin and P.~Sundell, ``{Holography in 4D (super) higher spin theories and
  a test via cubic scalar couplings},''
  \href{http://dx.doi.org/10.1088/1126-6708/2005/07/044}{{\em JHEP} {\bf 07}
  (2005)  044},
\href{http://arxiv.org/abs/hep-th/0305040}{{\tt arXiv:hep-th/0305040
  [hep-th]}}.

\bibitem{Giombi:2009wh}
S.~Giombi and X.~Yin, ``{Higher Spin Gauge Theory and Holography: The
  Three-Point Functions},''
  \href{http://dx.doi.org/10.1007/JHEP09(2010)115}{{\em JHEP} {\bf 09} (2010)
  115},
\href{http://arxiv.org/abs/0912.3462}{{\tt arXiv:0912.3462 [hep-th]}}.

\bibitem{Giombi:2012ms}
S.~Giombi and X.~Yin, ``{The Higher Spin/Vector Model Duality},''
  \href{http://dx.doi.org/10.1088/1751-8113/46/21/214003}{{\em J. Phys.} {\bf
  A46} (2013)  214003},
\href{http://arxiv.org/abs/1208.4036}{{\tt arXiv:1208.4036 [hep-th]}}.

\bibitem{Giombi:2016ejx}
S.~Giombi, \href{http://dx.doi.org/10.1142/9789813149441_0003}{``{Higher Spin
  — CFT Duality},''} in {\em {Proceedings, Theoretical Advanced Study
  Institute in Elementary Particle Physics: New Frontiers in Fields and Strings
  (TASI 2015): Boulder, CO, USA, June 1-26, 2015}}, pp.~137--214.
\newblock 2017.
\newblock
\href{http://arxiv.org/abs/1607.02967}{{\tt arXiv:1607.02967 [hep-th]}}.
\newblock

\bibitem{Giombi:2016pvg}
S.~Giombi, I.~R. Klebanov, and Z.~M. Tan, ``{The ABC of Higher-Spin AdS/CFT},''
\href{http://arxiv.org/abs/1608.07611}{{\tt arXiv:1608.07611 [hep-th]}}.

\bibitem{Sleight:2016hyl}
C.~Sleight, {\em {Interactions in Higher-Spin Gravity: a Holographic
  Perspective}}.
\newblock PhD thesis, Munich U., 2016.
\newblock
\href{http://arxiv.org/abs/1610.01318}{{\tt arXiv:1610.01318 [hep-th]}}.
\newblock

\bibitem{Sleight:2017krf}
C.~Sleight, ``{Metric-like Methods in Higher Spin Holography},'' {\em PoS} {\bf
  Modave2016} (2017)  003,
\href{http://arxiv.org/abs/1701.08360}{{\tt arXiv:1701.08360 [hep-th]}}.

\bibitem{Bekaert:2014cea}
X.~Bekaert, J.~Erdmenger, D.~Ponomarev, and C.~Sleight, ``{Towards holographic
  higher-spin interactions: Four-point functions and higher-spin exchange},''
  \href{http://dx.doi.org/10.1007/JHEP03(2015)170}{{\em JHEP} {\bf 03} (2015)
  170},
\href{http://arxiv.org/abs/1412.0016}{{\tt arXiv:1412.0016 [hep-th]}}.

\bibitem{Ruhl:2004cf}
W.~Ruhl, ``{The Masses of gauge fields in higher spin field theory on
  AdS(4)},'' \href{http://dx.doi.org/10.1016/j.physletb.2004.11.050}{{\em Phys.
  Lett.} {\bf B605} (2005)  413--418},
\href{http://arxiv.org/abs/hep-th/0409252}{{\tt arXiv:hep-th/0409252
  [hep-th]}}.

\bibitem{Manvelyan:2004ii}
R.~Manvelyan and W.~Ruhl, ``{The Masses of gauge fields in higher spin field
  theory on the bulk of AdS(4)},''
  \href{http://dx.doi.org/10.1016/j.physletb.2005.03.061}{{\em Phys. Lett.}
  {\bf B613} (2005)  197--207},
\href{http://arxiv.org/abs/hep-th/0412252}{{\tt arXiv:hep-th/0412252
  [hep-th]}}.

\bibitem{Manvelyan:2005fp}
R.~Manvelyan and W.~Ruhl, ``{The Off-shell behaviour of propagators and the
  Goldstone field in higher spin gauge theory on AdS(d+1) space},''
  \href{http://dx.doi.org/10.1016/j.nuclphysb.2005.03.038}{{\em Nucl. Phys.}
  {\bf B717} (2005)  3--18},
\href{http://arxiv.org/abs/hep-th/0502123}{{\tt arXiv:hep-th/0502123
  [hep-th]}}.

\bibitem{Sleight:2016dba}
C.~Sleight and M.~Taronna, ``{Higher Spin Interactions from Conformal Field
  Theory: The Complete Cubic Couplings},''
  \href{http://dx.doi.org/10.1103/PhysRevLett.116.181602}{{\em Phys. Rev.
  Lett.} {\bf 116} (2016) no.~18, 181602},
\href{http://arxiv.org/abs/1603.00022}{{\tt arXiv:1603.00022 [hep-th]}}.

\bibitem{Bekaert:2015tva}
X.~Bekaert, J.~Erdmenger, D.~Ponomarev, and C.~Sleight, ``{Quartic AdS
  Interactions in Higher-Spin Gravity from Conformal Field Theory},''
  \href{http://dx.doi.org/10.1007/JHEP11(2015)149}{{\em JHEP} {\bf 11} (2015)
  149},
\href{http://arxiv.org/abs/1508.04292}{{\tt arXiv:1508.04292 [hep-th]}}.

\bibitem{Boulanger:2015ova}
N.~Boulanger, P.~Kessel, E.~D. Skvortsov, and M.~Taronna, ``{Higher spin
  interactions in four-dimensions: Vasiliev versus Fronsdal},''
  \href{http://dx.doi.org/10.1088/1751-8113/49/9/095402}{{\em J. Phys.} {\bf
  A49} (2016) no.~9, 095402},
\href{http://arxiv.org/abs/1508.04139}{{\tt arXiv:1508.04139 [hep-th]}}.

\bibitem{Skvortsov:2015lja}
E.~D. Skvortsov and M.~Taronna, ``{On Locality, Holography and Unfolding},''
  \href{http://dx.doi.org/10.1007/JHEP11(2015)044}{{\em JHEP} {\bf 11} (2015)
  044},
\href{http://arxiv.org/abs/1508.04764}{{\tt arXiv:1508.04764 [hep-th]}}.

\bibitem{Vasiliev:2016xui}
M.~A. Vasiliev, ``{Current Interactions and Holography from the 0-Form Sector
  of Nonlinear Higher-Spin Equations},''
\href{http://arxiv.org/abs/1605.02662}{{\tt arXiv:1605.02662 [hep-th]}}.

\bibitem{Taronna:2017jeq}
M.~Taronna, ``{A note on field redefinitions and higher-spin equations},''
\href{http://dx.doi.org/10.1051/epjconf/201612505025}{{\em EPJ Web Conf.} {\bf
  125} (2016)  05025}.

\bibitem{Sleight:2017pcz}
C.~Sleight and M.~Taronna, ``{Higher spin gauge theories and bulk locality: a
  no-go result},''
\href{http://arxiv.org/abs/1704.07859}{{\tt arXiv:1704.07859 [hep-th]}}.

\bibitem{Bonezzi:2017vha}
R.~Bonezzi, N.~Boulanger, D.~De~Filippi, and P.~Sundell, ``{Noncommutative
  Wilson lines in higher-spin theory and correlation functions of conserved
  currents for free conformal fields},''
\href{http://arxiv.org/abs/1705.03928}{{\tt arXiv:1705.03928 [hep-th]}}.

\bibitem{Maldacena:2011jn}
J.~Maldacena and A.~Zhiboedov, ``{Constraining Conformal Field Theories with A
  Higher Spin Symmetry},''
  \href{http://dx.doi.org/10.1088/1751-8113/46/21/214011}{{\em J. Phys.} {\bf
  A46} (2013)  214011},
\href{http://arxiv.org/abs/1112.1016}{{\tt arXiv:1112.1016 [hep-th]}}.

\bibitem{Alba:2013yda}
V.~Alba and K.~Diab, ``{Constraining conformal field theories with a higher
  spin symmetry in d=4},''
\href{http://arxiv.org/abs/1307.8092}{{\tt arXiv:1307.8092 [hep-th]}}.

\bibitem{Alba:2015upa}
V.~Alba and K.~Diab, ``{Constraining conformal field theories with a higher
  spin symmetry in $d > 3$ dimensions},''
  \href{http://dx.doi.org/10.1007/JHEP03(2016)044}{{\em JHEP} {\bf 03} (2016)
  044},
\href{http://arxiv.org/abs/1510.02535}{{\tt arXiv:1510.02535 [hep-th]}}.

\bibitem{Boulanger:2013zza}
N.~Boulanger, D.~Ponomarev, E.~D. Skvortsov, and M.~Taronna, ``{On the
  uniqueness of higher-spin symmetries in AdS and CFT},''
  \href{http://dx.doi.org/10.1142/S0217751X13501625}{{\em Int. J. Mod. Phys.}
  {\bf A28} (2013)  1350162},
\href{http://arxiv.org/abs/1305.5180}{{\tt arXiv:1305.5180 [hep-th]}}.

\bibitem{Siegel:1988gd}
W.~Siegel, ``{All Free Conformal Representations in All Dimensions},''
\href{http://dx.doi.org/10.1142/S0217751X89000819}{{\em Int. J. Mod. Phys.}
  {\bf A4} (1989)  2015}.

\bibitem{Leigh:2003gk}
R.~G. Leigh and A.~C. Petkou, ``{Holography of the N=1 higher spin theory on
  AdS(4)},'' \href{http://dx.doi.org/10.1088/1126-6708/2003/06/011}{{\em JHEP}
  {\bf 06} (2003)  011},
\href{http://arxiv.org/abs/hep-th/0304217}{{\tt arXiv:hep-th/0304217
  [hep-th]}}.

\bibitem{Skvortsov:2015pea}
E.~D. Skvortsov, \href{http://dx.doi.org/10.1142/9789813144101_0008}{``{On
  (Un)Broken Higher-Spin Symmetry in Vector Models},''} in {\em {Proceedings,
  International Workshop on Higher Spin Gauge Theories: Singapore, Singapore,
  November 4-6, 2015}}, pp.~103--137.
\newblock 2017.
\newblock
\href{http://arxiv.org/abs/1512.05994}{{\tt arXiv:1512.05994 [hep-th]}}.
\newblock

\bibitem{Maldacena:2012sf}
J.~Maldacena and A.~Zhiboedov, ``{Constraining conformal field theories with a
  slightly broken higher spin symmetry},''
  \href{http://dx.doi.org/10.1088/0264-9381/30/10/104003}{{\em Class. Quant.
  Grav.} {\bf 30} (2013)  104003},
\href{http://arxiv.org/abs/1204.3882}{{\tt arXiv:1204.3882 [hep-th]}}.

\bibitem{Boulanger:2011se}
N.~Boulanger and E.~D. Skvortsov, ``{Higher-spin algebras and cubic
  interactions for simple mixed-symmetry fields in AdS spacetime},''
  \href{http://dx.doi.org/10.1007/JHEP09(2011)063}{{\em JHEP} {\bf 09} (2011)
  063},
\href{http://arxiv.org/abs/1107.5028}{{\tt arXiv:1107.5028 [hep-th]}}.

\bibitem{Joung:2014qya}
E.~Joung and K.~Mkrtchyan, ``{Notes on higher-spin algebras: minimal
  representations and structure constants},''
  \href{http://dx.doi.org/10.1007/JHEP05(2014)103}{{\em JHEP} {\bf 05} (2014)
  103},
\href{http://arxiv.org/abs/1401.7977}{{\tt arXiv:1401.7977 [hep-th]}}.

\bibitem{Joseph1976}
A.~Joseph, ``{The minimal orbit in a simple Lie algebra and its associated
  maximal ideal},'' {\em Ann. Sci. {\'E}cole Norm. Sup.(4)} {\bf 9} (1976)
  no.~1, 1--29.

\bibitem{Bekaert:2009fg}
X.~Bekaert and M.~Grigoriev, ``{Manifestly conformal descriptions and higher
  symmetries of bosonic singletons},''
  \href{http://dx.doi.org/10.3842/SIGMA.2010.038}{{\em SIGMA} {\bf 6} (2010)
  038},
\href{http://arxiv.org/abs/0907.3195}{{\tt arXiv:0907.3195 [hep-th]}}.

\bibitem{Labastida:1987kw}
J.~M.~F. Labastida, ``{Massless Particles in Arbitrary Representations of the
  Lorentz Group},''
\href{http://dx.doi.org/10.1016/0550-3213(89)90490-2}{{\em Nucl. Phys.} {\bf
  B322} (1989)  185--209}.

\bibitem{Metsaev:1995re}
R.~R. Metsaev, ``{Massless mixed symmetry bosonic free fields in d-dimensional
  anti-de Sitter space-time},''
\href{http://dx.doi.org/10.1016/0370-2693(95)00563-Z}{{\em Phys. Lett.} {\bf
  B354} (1995)  78--84}.

\bibitem{Metsaev:1997nj}
R.~R. Metsaev, ``{Arbitrary spin massless bosonic fields in d-dimensional
  anti-de Sitter space},'' \href{http://dx.doi.org/10.1007/BFb0104614}{{\em
  Lect. Notes Phys.} {\bf 524} (1999)  331--340},
\href{http://arxiv.org/abs/hep-th/9810231}{{\tt arXiv:hep-th/9810231
  [hep-th]}}.

\bibitem{Metsaev:1998xg}
R.~R. Metsaev, ``{Fermionic fields in the d-dimensional anti-de Sitter
  space-time},'' \href{http://dx.doi.org/10.1016/S0370-2693(97)01446-9}{{\em
  Phys. Lett.} {\bf B419} (1998)  49--56},
\href{http://arxiv.org/abs/hep-th/9802097}{{\tt arXiv:hep-th/9802097
  [hep-th]}}.

\bibitem{Burdik:2000kj}
C.~Burdik, A.~Pashnev, and M.~Tsulaia, ``{The Lagrangian description of
  representations of the Poincare group},''
  \href{http://dx.doi.org/10.1016/S0920-5632(01)01568-7}{{\em Nucl. Phys. Proc.
  Suppl.} {\bf 102} (2001)  285--292},
  \href{http://arxiv.org/abs/hep-th/0103143}{{\tt arXiv:hep-th/0103143
  [hep-th]}}.
[,285(2000)].

\bibitem{Burdik:2001hj}
C.~Burdik, A.~Pashnev, and M.~Tsulaia, ``{On the Mixed symmetry irreducible
  representations of the Poincare group in the BRST approach},''
  \href{http://dx.doi.org/10.1142/S0217732301003826}{{\em Mod. Phys. Lett.}
  {\bf A16} (2001)  731--746},
\href{http://arxiv.org/abs/hep-th/0101201}{{\tt arXiv:hep-th/0101201
  [hep-th]}}.

\bibitem{Alkalaev:2003qv}
K.~B. Alkalaev, O.~V. Shaynkman, and M.~A. Vasiliev, ``{On the frame - like
  formulation of mixed symmetry massless fields in (A)dS(d)},''
  \href{http://dx.doi.org/10.1016/j.nuclphysb.2004.05.031}{{\em Nucl. Phys.}
  {\bf B692} (2004)  363--393},
\href{http://arxiv.org/abs/hep-th/0311164}{{\tt arXiv:hep-th/0311164
  [hep-th]}}.

\bibitem{Bekaert:2003az}
X.~Bekaert and N.~Boulanger, ``{On geometric equations and duality for free
  higher spins},'' \href{http://dx.doi.org/10.1016/S0370-2693(03)00409-X}{{\em
  Phys. Lett.} {\bf B561} (2003)  183--190},
\href{http://arxiv.org/abs/hep-th/0301243}{{\tt arXiv:hep-th/0301243
  [hep-th]}}.

\bibitem{Bekaert:2003zq}
X.~Bekaert and N.~Boulanger, ``{Mixed symmetry gauge fields in a flat
  background},'' in {\em {Proceedings, 5th International Workshop on
  Supersymmetries and Quantum Symmetries (SQS'03): Dubna, Russia, July 24 - 29,
  2003}}, pp.~37--42.
\newblock 2004.
\newblock
\href{http://arxiv.org/abs/hep-th/0310209}{{\tt arXiv:hep-th/0310209
  [hep-th]}}.
\newblock

\bibitem{Boulanger:2008up}
N.~Boulanger, C.~Iazeolla, and P.~Sundell, ``{Unfolding Mixed-Symmetry Fields
  in AdS and the BMV Conjecture: I. General Formalism},''
  \href{http://dx.doi.org/10.1088/1126-6708/2009/07/013}{{\em JHEP} {\bf 07}
  (2009)  013},
\href{http://arxiv.org/abs/0812.3615}{{\tt arXiv:0812.3615 [hep-th]}}.

\bibitem{Boulanger:2008kw}
N.~Boulanger, C.~Iazeolla, and P.~Sundell, ``{Unfolding Mixed-Symmetry Fields
  in AdS and the BMV Conjecture. II. Oscillator Realization},''
  \href{http://dx.doi.org/10.1088/1126-6708/2009/07/014}{{\em JHEP} {\bf 07}
  (2009)  014},
\href{http://arxiv.org/abs/0812.4438}{{\tt arXiv:0812.4438 [hep-th]}}.

\bibitem{Skvortsov:2008vs}
E.~D. Skvortsov, ``{Mixed-Symmetry Massless Fields in Minkowski space
  Unfolded},'' \href{http://dx.doi.org/10.1088/1126-6708/2008/07/004}{{\em
  JHEP} {\bf 07} (2008)  004},
\href{http://arxiv.org/abs/0801.2268}{{\tt arXiv:0801.2268 [hep-th]}}.

\bibitem{Alkalaev:2008gi}
K.~B. Alkalaev, M.~Grigoriev, and I.~{\relax Yu}. Tipunin, ``{Massless Poincare
  modules and gauge invariant equations},''
  \href{http://dx.doi.org/10.1016/j.nuclphysb.2009.08.007}{{\em Nucl. Phys.}
  {\bf B823} (2009)  509--545},
\href{http://arxiv.org/abs/0811.3999}{{\tt arXiv:0811.3999 [hep-th]}}.

\bibitem{Skvortsov:2009nv}
E.~D. Skvortsov, ``{Gauge fields in (A)dS(d) within the unfolded approach:
  algebraic aspects},'' \href{http://dx.doi.org/10.1007/JHEP01(2010)106}{{\em
  JHEP} {\bf 01} (2010)  106},
\href{http://arxiv.org/abs/0910.3334}{{\tt arXiv:0910.3334 [hep-th]}}.

\bibitem{Skvortsov:2009zu}
E.~D. Skvortsov, ``{Gauge fields in (A)dS(d) and Connections of its symmetry
  algebra},'' \href{http://dx.doi.org/10.1088/1751-8113/42/38/385401}{{\em J.
  Phys.} {\bf A42} (2009)  385401},
\href{http://arxiv.org/abs/0904.2919}{{\tt arXiv:0904.2919 [hep-th]}}.

\bibitem{Campoleoni:2009je}
A.~Campoleoni, ``{Metric-like Lagrangian Formulations for Higher-Spin Fields of
  Mixed Symmetry},'' \href{http://dx.doi.org/10.1393/ncr/i2010-10053-2}{{\em
  Riv. Nuovo Cim.} {\bf 33} (2010)  123--253},
\href{http://arxiv.org/abs/0910.3155}{{\tt arXiv:0910.3155 [hep-th]}}.

\bibitem{Alkalaev:2009vm}
K.~B. Alkalaev and M.~Grigoriev, ``{Unified BRST description of AdS gauge
  fields},'' \href{http://dx.doi.org/10.1016/j.nuclphysb.2010.04.004}{{\em
  Nucl. Phys.} {\bf B835} (2010)  197--220},
\href{http://arxiv.org/abs/0910.2690}{{\tt arXiv:0910.2690 [hep-th]}}.

\bibitem{Alkalaev:2011zv}
K.~Alkalaev and M.~Grigoriev, ``{Unified BRST approach to (partially) massless
  and massive AdS fields of arbitrary symmetry type},''
  \href{http://dx.doi.org/10.1016/j.nuclphysb.2011.08.005}{{\em Nucl. Phys.}
  {\bf B853} (2011)  663--687},
\href{http://arxiv.org/abs/1105.6111}{{\tt arXiv:1105.6111 [hep-th]}}.

\bibitem{Campoleoni:2012th}
A.~Campoleoni and D.~Francia, ``{Maxwell-like Lagrangians for higher spins},''
  \href{http://dx.doi.org/10.1007/JHEP03(2013)168}{{\em JHEP} {\bf 03} (2013)
  168},
\href{http://arxiv.org/abs/1206.5877}{{\tt arXiv:1206.5877 [hep-th]}}.

\bibitem{Metsaev:2005ar}
R.~R. Metsaev, ``{Cubic interaction vertices of massive and massless higher
  spin fields},'' \href{http://dx.doi.org/10.1016/j.nuclphysb.2006.10.002}{{\em
  Nucl. Phys.} {\bf B759} (2006)  147--201},
\href{http://arxiv.org/abs/hep-th/0512342}{{\tt arXiv:hep-th/0512342
  [hep-th]}}.

\bibitem{Metsaev:2007rn}
R.~R. Metsaev, ``{Cubic interaction vertices for fermionic and bosonic
  arbitrary spin fields},''
  \href{http://dx.doi.org/10.1016/j.nuclphysb.2012.01.022}{{\em Nucl. Phys.}
  {\bf B859} (2012)  13--69},
\href{http://arxiv.org/abs/0712.3526}{{\tt arXiv:0712.3526 [hep-th]}}.

\bibitem{Alkalaev:2010af}
K.~Alkalaev, ``{FV-type action for $AdS_5$ mixed-symmetry fields},''
  \href{http://dx.doi.org/10.1007/JHEP03(2011)031}{{\em JHEP} {\bf 03} (2011)
  031},
\href{http://arxiv.org/abs/1011.6109}{{\tt arXiv:1011.6109 [hep-th]}}.

\bibitem{Boulanger:2011qt}
N.~Boulanger, E.~D. Skvortsov, and {\relax Yu}.~M. Zinoviev, ``{Gravitational
  cubic interactions for a simple mixed-symmetry gauge field in AdS and flat
  backgrounds},'' \href{http://dx.doi.org/10.1088/1751-8113/44/41/415403}{{\em
  J. Phys.} {\bf A44} (2011)  415403},
\href{http://arxiv.org/abs/1107.1872}{{\tt arXiv:1107.1872 [hep-th]}}.

\bibitem{Alkalaev:2012rg}
K.~Alkalaev, ``{Mixed-symmetry tensor conserved currents and AdS/CFT
  correspondence},''
  \href{http://dx.doi.org/10.1088/1751-8113/46/21/214007}{{\em J. Phys.} {\bf
  A46} (2013)  214007},
\href{http://arxiv.org/abs/1207.1079}{{\tt arXiv:1207.1079 [hep-th]}}.

\bibitem{Alkalaev:2012ic}
K.~Alkalaev, ``{Massless hook field in AdS(d+1) from the holographic
  perspective},'' \href{http://dx.doi.org/10.1007/JHEP01(2013)018}{{\em JHEP}
  {\bf 01} (2013)  018},
\href{http://arxiv.org/abs/1210.0217}{{\tt arXiv:1210.0217 [hep-th]}}.

\bibitem{Chekmenev:2015kzf}
A.~Chekmenev and M.~Grigoriev, ``{Boundary values of mixed-symmetry massless
  fields in AdS space},''
  \href{http://dx.doi.org/10.1016/j.nuclphysb.2016.10.006}{{\em Nucl. Phys.}
  {\bf B913} (2016)  769--791},
\href{http://arxiv.org/abs/1512.06443}{{\tt arXiv:1512.06443 [hep-th]}}.

\bibitem{Bekaert:2013zya}
X.~Bekaert and M.~Grigoriev, ``{Higher order singletons, partially massless
  fields and their boundary values in the ambient approach},''
  \href{http://dx.doi.org/10.1016/j.nuclphysb.2013.08.015}{{\em Nucl. Phys.}
  {\bf B876} (2013)  667--714},
\href{http://arxiv.org/abs/1305.0162}{{\tt arXiv:1305.0162 [hep-th]}}.

\bibitem{Brust:2016zns}
C.~Brust and K.~Hinterbichler, ``{Partially Massless Higher-Spin Theory},''
  \href{http://dx.doi.org/10.1007/JHEP02(2017)086}{{\em JHEP} {\bf 02} (2017)
  086},
\href{http://arxiv.org/abs/1610.08510}{{\tt arXiv:1610.08510 [hep-th]}}.

\bibitem{Alkalaev:2014nsa}
K.~B. Alkalaev, M.~Grigoriev, and E.~D. Skvortsov, ``{Uniformizing higher-spin
  equations},'' \href{http://dx.doi.org/10.1088/1751-8113/48/1/015401}{{\em J.
  Phys.} {\bf A48} (2015) no.~1, 015401},
\href{http://arxiv.org/abs/1409.6507}{{\tt arXiv:1409.6507 [hep-th]}}.

\bibitem{Brust:2016xif}
C.~Brust and K.~Hinterbichler, ``{Partially Massless Higher-Spin Theory II:
  One-Loop Effective Actions},''
  \href{http://dx.doi.org/10.1007/JHEP01(2017)126}{{\em JHEP} {\bf 01} (2017)
  126},
\href{http://arxiv.org/abs/1610.08522}{{\tt arXiv:1610.08522 [hep-th]}}.

\bibitem{Brust:2016gjy}
C.~Brust and K.~Hinterbichler, ``{Free □$^{k}$ scalar conformal field
  theory},'' \href{http://dx.doi.org/10.1007/JHEP02(2017)066}{{\em JHEP} {\bf
  02} (2017)  066},
\href{http://arxiv.org/abs/1607.07439}{{\tt arXiv:1607.07439 [hep-th]}}.

\bibitem{Eastwood2008}
M.~Eastwood and T.~Leistner, ``{Higher symmetries of the square of the
  Laplacian},'' in {\em Symmetries and overdetermined systems of partial
  differential equations}, pp.~319--338.
\newblock Springer, 2008.

\bibitem{Eastwood:2002su}
M.~G. Eastwood, ``{Higher symmetries of the Laplacian},''
  \href{http://dx.doi.org/10.4007/annals.2005.161.1645}{{\em Annals Math.} {\bf
  161} (2005)  1645--1665},
\href{http://arxiv.org/abs/hep-th/0206233}{{\tt arXiv:hep-th/0206233
  [hep-th]}}.

\bibitem{Gover2012}
A.~R. Gover and J.~{\v{S}}ilhan, ``{Higher symmetries of the conformal powers
  of the Laplacian on conformally flat manifolds},'' {\em Journal of
  Mathematical Physics} {\bf 53} (2012) no.~3, 032301.

\bibitem{Michel2014}
J.-P. Michel, ``{Higher symmetries of the Laplacian via quantization},'' in
  {\em Annales de l'institut Fourier}, vol.~64, pp.~1581--1609.
\newblock 2014.

\bibitem{Joung:2015jza}
E.~Joung and K.~Mkrtchyan, ``{Partially-massless higher-spin algebras and their
  finite-dimensional truncations},''
  \href{http://dx.doi.org/10.1007/JHEP01(2016)003}{{\em JHEP} {\bf 01} (2016)
  003},
\href{http://arxiv.org/abs/1508.07332}{{\tt arXiv:1508.07332 [hep-th]}}.

\bibitem{Deser:1983mm}
S.~Deser and R.~I. Nepomechie, ``{Gauge Invariance Versus Masslessness in De
  Sitter Space},''
\href{http://dx.doi.org/10.1016/0003-4916(84)90156-8}{{\em Annals Phys.} {\bf
  154} (1984)  396}.

\bibitem{Deser:1983tm}
S.~Deser and R.~I. Nepomechie, ``{Anomalous Propagation of Gauge Fields in
  Conformally Flat Spaces},''
\href{http://dx.doi.org/10.1016/0370-2693(83)90317-9}{{\em Phys. Lett.} {\bf
  132B} (1983)  321--324}.

\bibitem{Higuchi:1986wu}
A.~Higuchi, ``{Symmetric Tensor Spherical Harmonics on the $N$ Sphere and Their
  Application to the De Sitter Group SO($N$,1)},''
  \href{http://dx.doi.org/10.1063/1.527513}{{\em J. Math. Phys.} {\bf 28}
  (1987)  1553}.
[Erratum: J. Math. Phys.43,6385(2002)].

\bibitem{Deser:2001us}
S.~Deser and A.~Waldron, ``{Partial masslessness of higher spins in (A)dS},''
  \href{http://dx.doi.org/10.1016/S0550-3213(01)00212-7}{{\em Nucl. Phys.} {\bf
  B607} (2001)  577--604},
\href{http://arxiv.org/abs/hep-th/0103198}{{\tt arXiv:hep-th/0103198
  [hep-th]}}.

\bibitem{Skvortsov:2006at}
E.~D. Skvortsov and M.~A. Vasiliev, ``{Geometric formulation for partially
  massless fields},''
  \href{http://dx.doi.org/10.1016/j.nuclphysb.2006.06.019}{{\em Nucl. Phys.}
  {\bf B756} (2006)  117--147},
\href{http://arxiv.org/abs/hep-th/0601095}{{\tt arXiv:hep-th/0601095
  [hep-th]}}.

\bibitem{Basile:2016aen}
T.~Basile, X.~Bekaert, and N.~Boulanger, ``{Mixed-symmetry fields in de Sitter
  space: a group theoretical glance},''
  \href{http://dx.doi.org/10.1007/JHEP05(2017)081}{{\em JHEP} {\bf 05} (2017)
  081},
\href{http://arxiv.org/abs/1612.08166}{{\tt arXiv:1612.08166 [hep-th]}}.

\bibitem{Gwak:2016sma}
S.~Gwak, J.~Kim, and S.-J. Rey, ``{Massless and Massive Higher Spins from
  Anti-de Sitter Space Waveguide},''
  \href{http://dx.doi.org/10.1007/JHEP11(2016)024}{{\em JHEP} {\bf 11} (2016)
  024},
\href{http://arxiv.org/abs/1605.06526}{{\tt arXiv:1605.06526 [hep-th]}}.

\bibitem{Gwak:2017cyu}
S.~Gwak, J.~Kim, and S.-J. Rey,
  \href{http://dx.doi.org/10.1142/9789813144101_0017}{``{Higgs Mechanism and
  Holography of Partially Massless Higher Spin Fields},''} in {\em
  {Proceedings, International Workshop on Higher Spin Gauge Theories:
  Singapore, Singapore, November 4-6, 2015}}, pp.~317--352.
\newblock
2017.
\newblock

\bibitem{Basile:2014wua}
T.~Basile, X.~Bekaert, and N.~Boulanger, ``{Flato-Fronsdal theorem for
  higher-order singletons},''
  \href{http://dx.doi.org/10.1007/JHEP11(2014)131}{{\em JHEP} {\bf 11} (2014)
  131},
\href{http://arxiv.org/abs/1410.7668}{{\tt arXiv:1410.7668 [hep-th]}}.

\bibitem{Beccaria:2014jxa}
M.~Beccaria, X.~Bekaert, and A.~A. Tseytlin, ``{Partition function of free
  conformal higher spin theory},''
  \href{http://dx.doi.org/10.1007/JHEP08(2014)113}{{\em JHEP} {\bf 08} (2014)
  113},
\href{http://arxiv.org/abs/1406.3542}{{\tt arXiv:1406.3542 [hep-th]}}.

\bibitem{Bekaert:2011js}
X.~Bekaert, ``{Singletons and their maximal symmetry algebras},'' in {\em
  {Modern Mathematical Physics. Proceedings, 6th Summer School: Belgrade,
  Serbia, September 14-23, 2010}}, pp.~71--89.
\newblock 2011.
\newblock
\href{http://arxiv.org/abs/1111.4554}{{\tt arXiv:1111.4554 [math-ph]}}.
\newblock

\bibitem{Fernando:2015tiu}
S.~Fernando and M.~Günaydin, ``{Massless conformal fields, $AdS_{d+1}/CFT_d$
  higher spin algebras and their deformations},''
  \href{http://dx.doi.org/10.1016/j.nuclphysb.2016.01.024}{{\em Nucl. Phys.}
  {\bf B904} (2016)  494--526},
\href{http://arxiv.org/abs/1511.02167}{{\tt arXiv:1511.02167 [hep-th]}}.

\bibitem{Laoues:1998ik}
M.~Laoues, ``{Massless particles in arbitrary dimensions},''
  \href{http://dx.doi.org/10.1142/S0129055X98000355}{{\em Rev. Math. Phys.}
  {\bf 10} (1998)  1079--1109},
\href{http://arxiv.org/abs/hep-th/9806101}{{\tt arXiv:hep-th/9806101
  [hep-th]}}.

\bibitem{Dirac:1963ta}
P.~A.~M. Dirac, ``{A Remarkable representation of the 3 + 2 de Sitter group},''
\href{http://dx.doi.org/10.1063/1.1704016}{{\em J. Math. Phys.} {\bf 4} (1963)
  901--909}.

\bibitem{Ferrara:2000nu}
S.~Ferrara and C.~Fronsdal, ``{Conformal fields in higher dimensions},'' in
  {\em {Recent developments in theoretical and experimental general relativity,
  gravitation and relativistic field theories. Proceedings, 9th Marcel
  Grossmann Meeting, MG'9, Rome, Italy, July 2-8, 2000. Pts. A-C}},
  pp.~508--527.
\newblock 2000.
\newblock
\href{http://arxiv.org/abs/hep-th/0006009}{{\tt arXiv:hep-th/0006009
  [hep-th]}}.
\newblock

\bibitem{Enright1983}
T.~Enright, R.~Howe, and N.~Wallach, ``{A classification of unitary highest
  weight modules},'' in {\em Representation theory of reductive groups},
  pp.~97--143.
\newblock Springer, 1983.

\bibitem{Shaynkman:2004vu}
O.~V. Shaynkman, I.~{\relax Yu}. Tipunin, and M.~A. Vasiliev, ``{Unfolded form
  of conformal equations in M dimensions and o(M + 2) modules},''
  \href{http://dx.doi.org/10.1142/S0129055X06002814}{{\em Rev. Math. Phys.}
  {\bf 18} (2006)  823--886},
\href{http://arxiv.org/abs/hep-th/0401086}{{\tt arXiv:hep-th/0401086
  [hep-th]}}.

\bibitem{Ehrman1957}
J.~B. Ehrman, ``{On the unitary irreducible representations of the universal
  covering group of the 3+ 2 deSitter group},'' in {\em Mathematical
  Proceedings of the Cambridge Philosophical Society}, vol.~53, pp.~290--303,
  Cambridge Univ Press.
\newblock 1957.

\bibitem{Metsaev:1995jp}
R.~R. Metsaev, ``{All conformal invariant representations of d-dimensional
  anti-de Sitter group},''
\href{http://dx.doi.org/10.1142/S0217732395001848}{{\em Mod. Phys. Lett.} {\bf
  A10} (1995)  1719--1731}.

\bibitem{Barnich:2015tma}
G.~Barnich, X.~Bekaert, and M.~Grigoriev, ``{Notes on conformal invariance of
  gauge fields},'' \href{http://dx.doi.org/10.1088/1751-8113/48/50/505402}{{\em
  J. Phys.} {\bf A48} (2015) no.~50, 505402},
\href{http://arxiv.org/abs/1506.00595}{{\tt arXiv:1506.00595 [hep-th]}}.

\bibitem{Brink:2000ag}
L.~Brink, R.~R. Metsaev, and M.~A. Vasiliev, ``{How massless are massless
  fields in AdS(d)},''
  \href{http://dx.doi.org/10.1016/S0550-3213(00)00402-8}{{\em Nucl. Phys.} {\bf
  B586} (2000)  183--205},
\href{http://arxiv.org/abs/hep-th/0005136}{{\tt arXiv:hep-th/0005136
  [hep-th]}}.

\bibitem{Artsukevich:2008vy}
A.~{\relax Yu}. Artsukevich and M.~A. Vasiliev, ``{On Dimensional Degression in
  AdS(d)},'' \href{http://dx.doi.org/10.1103/PhysRevD.79.045007}{{\em Phys.
  Rev.} {\bf D79} (2009)  045007},
\href{http://arxiv.org/abs/0810.2065}{{\tt arXiv:0810.2065 [hep-th]}}.

\bibitem{Dolan:2001ih}
L.~Dolan, C.~R. Nappi, and E.~Witten, ``{Conformal operators for partially
  massless states},''
  \href{http://dx.doi.org/10.1088/1126-6708/2001/10/016}{{\em JHEP} {\bf 10}
  (2001)  016},
\href{http://arxiv.org/abs/hep-th/0109096}{{\tt arXiv:hep-th/0109096
  [hep-th]}}.

\bibitem{Bae:2016xmv}
J.-B. Bae, E.~Joung, and S.~Lal, ``{A note on vectorial AdS$_{5}$/CFT$_{4}$
  duality for spin-$j$ boundary theory},''
  \href{http://dx.doi.org/10.1007/JHEP12(2016)077}{{\em JHEP} {\bf 12} (2016)
  077},
\href{http://arxiv.org/abs/1611.00112}{{\tt arXiv:1611.00112 [hep-th]}}.

\bibitem{Vasiliev:2009ck}
M.~A. Vasiliev, ``{Bosonic conformal higher-spin fields of any symmetry},''
  \href{http://dx.doi.org/10.1016/j.nuclphysb.2009.12.010}{{\em Nucl. Phys.}
  {\bf B829} (2010)  176--224},
\href{http://arxiv.org/abs/0909.5226}{{\tt arXiv:0909.5226 [hep-th]}}.

\bibitem{Bekaert:2006py}
X.~Bekaert and N.~Boulanger, ``{The Unitary representations of the Poincare
  group in any spacetime dimension},'' in {\em {2nd Modave Summer School in
  Theoretical Physics Modave, Belgium, August 6-12, 2006}}.
\newblock 2006.
\newblock
\href{http://arxiv.org/abs/hep-th/0611263}{{\tt arXiv:hep-th/0611263
  [hep-th]}}.
\newblock

\bibitem{Fulton1991}
W.~Fulton and J.~Harris, {\em {Representation Theory: A First Course}}.
\newblock Graduate Texts in Mathematics. Springer New York, 1991.

\bibitem{Fuchs2003}
J.~Fuchs and C.~Schweigert, {\em {Symmetries, Lie Algebras and Representations:
  A Graduate Course for Physicists}}.
\newblock Cambridge Monographs on Mathem. Cambridge University Press, 2003.

\end{thebibliography}\endgroup

\end{document}